\documentclass[11pt]{article}
\usepackage{hyperref}
\hypersetup{
    colorlinks = true,
    linkcolor = {red}
}
\usepackage{pset}

\title{Improved Bounds for Randomly Sampling Colorings via Linear Programming}
\author{
Sitan Chen\thanks{EECS, MIT, Cambridge, Massachusetts. {\tt sitanc@mit.edu}. This work was supported in part by NSF CAREER Award CCF-1453261 and NSF Large CCF-1565235.}
\and
Michelle Delcourt
\thanks{School of Mathematics,
University of Birmingham, Birmingham, UK. {\tt m.delcourt@bham.ac.uk}. Research supported by supported by EPSRC grant EP/P009913/1.}
\and
Ankur Moitra\thanks{Math \& CSAIL, MIT, Cambridge, Massachusetts. {\tt moitra@mit.edu}. This work was supported in part by NSF CAREER Award CCF-1453261, NSF Large CCF-1565235, a David and Lucile Packard Fellowship, an Alfred P. Sloan Fellowship, and an ONR Young Investigator Award.}\\
\and
Guillem Perarnau
\thanks{School of Mathematics,
University of Birmingham, Birmingham, UK.  {\tt g.perarnau@bham.ac.uk}.}
\and
Luke Postle
\thanks{Combinatorics and Optimization Department,
University of Waterloo, Waterloo, Ontario N2L 3G1, Canada. {\tt lpostle@uwaterloo.ca}. Partially supported by NSERC
under Discovery Grant No. 2014-06162.}}

\newcolumntype{Y}{>{\centering\arraybackslash}X}

\renewcommand{\E}{\mathbb{E}}
\newcommand{\Bad}[1]{\textsc{Bad}$_{#1}$}
\newcommand{\Sing}[1]{\textsc{Sing}$_{#1}$}
\newcommand{\Good}[1]{\textsc{Good}$_{#1}$}
\newcommand{\Badend}[1]{\textsc{BadEnd}$_{#1}$}
\newcommand{\Goodend}[1]{\textsc{GoodEnd}$_{#1}$}
\newcommand{\Tstop}{T_{\text{stop}}}
\renewcommand{\P}{\mathbb{P}}
\renewcommand{\Pr}{\mathbb{P}}

\newif\ifedit

\editfalse

\ifedit
\newcommand{\new}[1]{\textcolor{red}{#1}}
\newcommand{\sitan}[1]{\textcolor{blue}{#1}}
\definecolor{mygray}{gray}{0.8}
\newcommand{\trim}[1]{\textcolor{mygray}{#1}}
\else
\newcommand{\new}[1]{#1}
\newcommand{\sitan}[1]{#1}
\newcommand{\trim}[1]{}
\fi

\begin{document}

\maketitle

\begin{abstract}
\normalsize
A well-known conjecture in computer science and statistical physics is that Glauber dynamics on the set of $k$-colorings of a graph $G$ on $n$ vertices with maximum degree $\Delta$ is rapidly mixing for $k\ge\Delta+2$. In FOCS 1999, Vigoda \cite{vigoda2000improved} showed that the flip dynamics (and therefore also Glauber dynamics) is rapidly mixing for any $k>\frac{11}{6}\Delta$. It turns out that there is a natural barrier at $\frac{11}{6}$, below which there is no one-step coupling that is contractive with respect to the Hamming metric, even for the flip dynamics.

We use linear programming and duality arguments to fully characterize the obstructions to going beyond $\frac{11}{6}$. These extremal configurations turn out to be quite brittle, and in this paper we use this to give two proofs that the Glauber dynamics is rapidly mixing for any $k\ge(\frac{11}{6} - \epsilon_0)\Delta$ for some absolute constant $\epsilon_0>0$. This is the first improvement to Vigoda's result that holds for general graphs. Our first approach analyzes a variable-length coupling in which these configurations break apart with high probability before the coupling terminates, and our other approach analyzes a one-step path coupling with a new metric that counts the extremal configurations. Additionally, our results extend to list coloring, a widely studied generalization of coloring, where the previously best known results required $k > 2 \Delta$. 
\end{abstract}

\thispagestyle{empty}
\setcounter{page}{0}

\newpage
%!TEX root = fullpaper.tex

\section{Introduction}

\subsection{Background}

Here we study the problem of sampling random proper colorings of a bounded degree graph. More precisely, let $k$ be the number of colors and let $\Delta$ be the maximum degree. A long-standing open question is to give an algorithm that works for any $k \geq \Delta + 2$, when the space of proper colorings is first connected. Despite a long line of intensive investigation \cite{jerrum1995very, salas1997absence, dyer2003randomly, dyer2013randomly, hayes2003randomly, hayes2003non, molloy2004glauber, hayes2005coupling, frieze2006randomly, frieze2007survey}, the best known bounds are quite far from the conjecture. 

In fact, there is a natural Markov chain called the {\em Glauber dynamics} that is widely believed to work: in each step, choose a random node and recolor it with a random color not appearing among its neighbors. It is easy to see that its steady state distribution is uniform on all proper $k$-colorings, again provided that $k \geq \Delta + 2$. It is even conjectured that on an $n$ node graph, the mixing time is $O(n \log n)$ which would be tight \cite{hayes2005general}. We remark that rapidly mixing Markov chains for sampling random colorings immediately give a fully polynomial randomized approximation scheme (FPRAS) for counting the number of proper colorings. There is also interest in this question in combinatorics \cite{brightwell2002random} and in statistical physics, where it corresponds to approximating the partition function of the zero temperature anti-ferromagnetic Potts model \cite{potts1952some}.

Jerrum \cite{jerrum1995very} gave the first significant results and showed that when $k > 2\Delta$ the Glauber dynamics mixes in time $O(n \log n)$. The modern proof of this result is easier and proceeds through {\em path coupling} \cite{bubley1997path}, whereby it is enough to couple the updates between two colorings $\sigma$ and $\tau$ that differ only at a single node $v$ and show that the expected distance between them is strictly decreasing. Then Jerrum's bound follows by comparing how often the distance between the colorings decreases (when $v$ is selected and after the update has the same color in both) vs. how often it increases (when a neighbor of $v$ is selected and recolored in one but not the other). This result is closely related to work in the statistical physics community by Salas and Sokal \cite{salas1997absence} on the {\em Dobrushin uniqueness condition}. 

In a breakthrough work, Vigoda \cite{vigoda2000improved} gave the first algorithm for sampling random colorings that crossed the natural barrier of $2\Delta$. His approach was through a different Markov chain that in addition to recoloring single nodes also swaps the colors in larger Kempe components (which are also called alternating components). His chain was a variant of the Wang-Swendsen-Koteck\'{y} (WSK) algorithm \cite{wang1989antiferromagnetic}. The key insight is that the bottleneck in Jerrum's approach \---- when the neighbors of $v$ all have distinct colors \---- can be circumvented by flipping larger components. More precisely, when a neighbor of $v$ is recolored in one chain in a way that would have increased the distance, one can instead match it with the flip of a Kempe component of size two in the other chain that keeps the distance the same. But now one needs to couple the flips of larger Kempe components in some manner. Vigoda devised a coupling and a choice of flip parameters that works for any $k > \frac{11}{6} \Delta$. His Markov chain mixes in time $O(n \log n)$ and one can also connect it to Glauber dynamics and prove an $O(n^2)$ mixing time under the same conditions. This is still the best known bound for general graphs. 

Subsequently, there was a flurry of work on getting better bounds for restricted families of graphs. Dyer and Frieze \cite{dyer2003randomly} considered graphs of maximum degree $\Omega(\log n)$ and girth $\Omega(\log\log n)$ and proved that the Glauber dynamics mixes rapidly whenever $k > \alpha \Delta$ where $\alpha$ is the solution to $\alpha = e^{1/\alpha}$ and numerically $\alpha = 1.763\cdots$. Their approach was to show that under the uniform distribution on proper colorings, the number of colors missing from the neighborhood of $v$ is roughly $k(1-\frac{1}{k})^\Delta$ with high probability. Results like these were later termed {\em local uniformity properties}. They were improved in many directions in terms of reducing the degree and/or the girth requirements to be independent of $n$ \cite{hayes2003randomly, molloy2004glauber, hayes2005coupling, frieze2006randomly, lau2006randomly}, culminating in two incomparable results. Dyer et al. \cite{dyer2013randomly} showed that Glauber dynamics mixes rapidly whenever $k > \beta \Delta$ where $\beta$ is the solution to $(1-e^{-1/\beta})^2 + \beta e^{-1/\beta} = 1$ and numerically $\beta = 1.489\cdots$ for girth $g \geq 6$ and the degree $\Delta$ being a sufficiently large constant. Hayes and Vigoda \cite{hayes2003non} showed rapid mixing for any $k > (1+\epsilon)\Delta$ for any $\epsilon > 0$ provided that the girth $g \geq 11$ and the degree is logarithmic, using an intriguing non-Markovian coupling.  

%It is important to emphasize that the types of local uniformity properties being exploited by the works above do {\em not} hold for general graphs \---- e.g. ones with triangles. In fact, even the chromatic number is asymptotically different: Johansson \cite{johansson1996asymptotic} proved that the chromatic number of a triangle free graph is at most $\frac{(9 + o(1)) \Delta}{\log \Delta}$ which was later improved by Molloy \cite{molloy2017list} to $\frac{(1 + o(1)) \Delta}{\log \Delta}$, compared to $\Delta + 1$ for general graphs which is tight. 

On the hardness side, Galanis et al. \cite{galanis2014inapproximability} showed that for triangle-free graphs, unless $\NP = \RP$, it is $\NP$-hard to approximately sample $k$-colorings for $d$-regular graphs $G$ when $k < d$, even when $G$ is triangle-free.

There have also been many other algorithmic improvements, but all for special graph families. Through an eigenvalue generalization of the Dobrushin condition, Hayes \cite{hayes2006simple} showed that Glauber dynamics mixes rapidly for $k > \Delta + c \sqrt{\Delta}$ on planar graphs and graphs of constant treewidth. Berger et al. \cite{berger2005glauber} showed rapid mixing on graphs of logarithmic cutwidth, which was strengthened by Vardi \cite{vardi2017randomly} to graphs of logarithmic pathwidth. Some recent papers have studied settings such as bipartite or random graphs \cite{dyer2006randomly, mossel2010gibbs, efthymiou2018sampling}, where it is possible to mix with fewer colors than the maximum degree. Hayes et al. \cite{hayes2015randomly} notably improved the abovementioned result of \cite{hayes2006simple} to show that Glauber dynamics in fact mixes rapidly for planar graphs when $k = \Omega(\Delta/\log\Delta)$. \cite{tetali2012phase} established that the mixing time of Glauber dynamics for sampling colorings of trees undergoes a phase transition at the reconstruction threshold (up to first order). %These works all leverage structural properties that hold in restricted settings.
{\em Given that there has been progress in so many restricted graph families, it is natural to wonder why there hasn't been any further progress on the general case since Vigoda's results. }

\subsection{Our Results}

Our main result is the first improvement on randomly sampling colorings on general bounded degree graphs since the FOCS $1999$ paper of Vigoda \cite{vigoda2000improved}. Specifically, we prove:

\begin{thm}\label{thm:main}
The flip dynamics for sampling $k$-colorings is rapidly mixing with mixing time $O(n \log n)$, for any $k \geq (\frac{11}{6} - \epsilon_0)\Delta$ where $\epsilon_0>0$ is an absolute constant that is independent of $\Delta$.
\end{thm}

As in Vigoda's paper \cite{vigoda2000improved}, we obtain the following as implications of our main result\footnote{In order to prove rapid mixing of Glauber dynamics, one can use the comparison technique of Diaconis and Saloff-Coste~\cite{diaconis1993comparison}. 
Vigoda~\cite{vigoda2000improved} directly used the results in~\cite{diaconis1993comparison} to show that the mixing time of the Glauber dynamics is $O(n^2\log{n})$. It has been observed (see e.g.~\cite{frieze2007survey}) that $\tau_{mix}(\epsilon)=O(n(\log{n}+\log{\epsilon^{-1}}))$ for flip dynamics implies mixing time $O(n^2)$ for Glauber dynamics. This can be shown using spectral bounds on the mixing time for $\epsilon=1/n$~\cite{sinclair1992improved} and observing that the spectral gaps of Glauber dynamics and flip dynamics are the same up to a constant factor. % (see the discussion in~\cite{DJV}). 
The same argument applies to our case, so Theorem~\ref{thm:main} implies Theorem~\ref{thm:compareglauber}.

Regarding Theorem~\ref{thm:phasetransition}, it is known (see e.g. Lemma 7 of \cite{vigoda2000improved}) that it holds provided that \emph{under any boundary configuration}, the flip dynamics mixes in time $O(n\log n)$ on $Q_L$ the $d$-dimensional cube of $\Z^d$ with side length $2L + 1$. Theorem~\ref{thm:listcolorings} below implies this is the case for $k \ge (11/6-\epsilon_0)\Delta$ (note that the degree of $Q_L$ is $\Delta = 2d$).
}: 

\begin{thm}
The Glauber dynamics for sampling $k$-colorings is rapidly mixing with mixing time $O(n^2)$, for any $k \geq (\frac{11}{6} - \epsilon_0)\Delta$ where $\epsilon_0>0$ is the same constant from Theorem~\ref{thm:main}.\label{thm:compareglauber}
\end{thm}

% Finally, we obtain new bounds on the location of the related phase transitions on the integer lattice. As usual, the result extends straightforwardly to other common lattices. 

\begin{thm}
The $k$-state zero temperature anti-ferromagnetic Potts model on $\Z^d$ lies in the disordered phase when $k \geq (\frac{11}{3} - 2 \epsilon_0)d$ where $\epsilon_0>0$ is the same constant from Theorem~\ref{thm:main}.\label{thm:phasetransition}
\end{thm}

Our proof is guided by linear programming and duality arguments. %It is not really a computer assisted proof in the sense that we obtain the improvement over $\frac{11}{6}$ by solving a single linear program that works for all degrees, that leverages various structural tools we prove about the space of proper colorings in general graphs. 
The starting point is the observation that choosing the best flip parameters in the flip dynamics (i.e. the probability of flipping Kempe components of each possible size), when utilizing Vigoda's greedy coupling \cite{vigoda2000improved}, can be cast as a linear program. In this manner, Vigoda's analysis provides a feasible solution and it is natural to wonder if choosing the flip parameters differently or flipping larger size components would lead to an improvement. As it turns out, the answer is no in a strong sense. Not only are Vigoda's choice of parameters an optimal solution to the linear programming relaxation\footnote{Several approximations are made along the way in reaching this linear program, such as restricting to flipping components of size at most $6$ and replacing certain infinite sets of constraints with a finite set of stronger constraints.}, \sitan{but we use \emph{extremal configurations}, which correspond to tight constraints in the linear program, to build a small family of hard instances showing that $\frac{11}{6}$ is a natural barrier for a class of analyses. Specifically, for all $n$ and infinitely many $\Delta$ we can exhibit two graphs $G_1$ and  $G_2$, together with a pair of proper colorings $(\sigma_i,\tau_i)$ for each $G_i$ that differ at a single node, such that the following holds: for any choice of flip parameters and one-step coupling of the flip dynamics, $(\sigma_1,\tau_1)$ and $(\sigma_2,\tau_2)$ cannot both contract in Hamming distance under the coupling provided $k < \frac{11}{6} \Delta$ (see Construction~\ref{constr} and Lemma~\ref{lem:tight}).}

%We leverage our linear programming formulation for finding a one-step coupling to pass to the dual and exactly characterize the family of extremal configurations of colors around $v$ that cause such analyses to get stuck at $\frac{11}{6}$. There are already some specific configurations along these lines mentioned in \cite{vigoda2000improved}. For our purposes, it is crucial that we have a full characterization: the intuition is that if we can defeat all of these configurations simultaneously, then by complementary slackness we ought to be able to break the $\frac{11}{6}$ barrier. 

At this juncture, there are two potential approaches for circumventing the $\frac{11}{6}$ barrier:
\begin{enumerate}

\item[(1)] \sitan{Instead of using a one-step coupling, use a multi-step coupling, still with respect to the Hamming metric.}

\item[(2)] \sitan{Use a one-step coupling but change the metric.}

\end{enumerate}

\noindent As it turns out, {\em both} approaches work and we are able to give two separate proofs of Theorem~\ref{thm:main}. The present work is a merger of \cite{chen2018linear} and \cite{delcourt2018rapid}, and we chose to present the results together because there are parallels between the two proofs.\footnote{The $\epsilon_0$ obtained in Theorem~\ref{thm:main} is slightly different under the two proofs, but both roughly on the order of $10^{-5}$.} \sitan{The main idea behind both approaches is that any extremal configuration is quite {\em brittle}, and if there are few enough extremal configurations present, there is a choice of flip parameters for which we can go below the 11/6 barrier. To get an idea for just how brittle these configurations are, suppose $v$ is the unique node of disagreement between two colorings $\sigma$ and $\tau$. For every color $c$, there is an extremal configuration which if present in $\sigma,\tau$ would require that in $\sigma$, the maximal connected induced subgraph colored only with $c$ and $\sigma(v)$ and containing $v$ is a tree of size 7 and moreover that $v$ has degree exactly 2 in this tree. There are many possible transitions in the flip dynamics that would break this rigid pattern, for instance, flipping the color of a descendant of $v$ in the tree.}

\sitan{In our first proof of Theorem~\ref{thm:main}, we run a multi-step coupling which terminates when the Hamming distance between the two colorings changes. We argue that by the time the coupling terminates, extremal configurations like the one above are likely to break apart. In the above example, regardless of the choice of flip parameters, at any point there is an $\Omega(1/n)$ chance of breaking it (Lemma~\ref{lem:badtogood}) and a $\Theta(1/n)$ chance the coupling terminates (Lemma~\ref{lem:probend}). What remains, and this is the most delicate part of the proof, is to upper bound the probability that configurations that are not extremal transition to ones that are, and the key point is that this is also $O(1/n)$ (Lemma~\ref{lem:goodtobadend}). These three bounds together imply that by the end of the coupling, the number of non-extremal configurations around $v$ will be at least some constant times the number of extremal configurations in expectation. This is enough to show that for some tuning of the flip dynamics, even if the expected change in Hamming distance after one step of coupling is still positive when $k$ is slightly below $\frac{11}{6}\Delta$, the expected change by the time the coupling terminates will be negative.}

\sitan{In our second proof of Theorem~\ref{thm:main}, the main idea is essentially the same: we want to win on colorings like $G_1$ and $G_2$ at the expense of losing on the other possible types of colorings. However we accomplish this by changing the metric. More precisely, we choose $d$ to be the Hamming metric $d_H$ minus a small correction factor $d_B$ that counts the number of extremal configurations around the node $v$ of disagreement. We show that even when $k$ is slightly less than $\frac{11}{6}\Delta$, there is a choice of flip parameters for which the following win-win analysis shows $d$ contracts in expectation in one step. For pairs of colorings with few non-extremal configurations around $v$, $d_H$ will stay the same in expectation in one step (Corollary~\ref{cor:improvement}). When most of the configurations around $v$ are extremal, we show that even though $d_H$ does not contract in expectation in one step, $d_B$ does (Corollary~\ref{cor:improvement2}). It is here that our proofs bear a sharp resemblance: The proof that $d_B$ is contractive revolves around the exact same types of reasoning about lower bounding the probability that an extremal configuration breaks apart, and upper bounding the probability that a non-extremal configuration becomes extremal (Lemma~\ref{lem:complement}).}

\sitan{While we have chosen to present the technical pieces in both proofs separately because of the subtle differences in their structure and relevant definitions\footnote{See Remark~\ref{remark:subtle} for a discussion of some of these differences.}, we believe that presenting both of them and drawing analogies when possible gives a more complete picture about the different ways to circumvent bottlenecks in one-step coupling with respect to the Hamming metric, and what the conceptual relationship is between these different techniques.}

%We believe that the idea of leveraging insights from linear programming and duality may be more generally useful in optimizing the design of Markov chains for sampling problems. Such techniques have been used in approximation algorithms \cite{goemans1998improved, jain2003greedy} under the name {\em factor revealing LPs}. Here they play a somewhat different role in utilizing duality and complementary slackness to identify and characterize in a principled manner the obstacles to a single step coupling under the Hamming metric. We remark that before Vigoda's work \cite{vigoda2000improved}, Bubley et al. \cite{bubley1998beating} used linear programming to show that Glauber dynamics is rapidly mixing with five colors on graphs with maximum degree three. Their approach required solving several hundred linear programs, and was subject to ``combinatorial explosion" as a function of the degree. This also points to a main challenge going forward: to find appropriate linear programming relaxations for constructing families of couplings that avoid such an explosion and to understand what these relaxations do and do not give up. This can be quite subtle, but we have many geometric tools to build our understanding upon.

Lastly, we are able to extend our techniques to the problem of sampling \emph{list colorings}, a natural and well-studied generalization of coloring. Jerrum's proof for $k > 2\Delta$ carries over immediately for list-coloring, and while there have been subsequent works studying this problem in the case of \emph{triangle-free} graphs when $k > 1.763\Delta$ \cite{goldberg2005strong,gamarnik2007correlation,gamarnik2015strong}, there had been no known improvement on $2 \Delta$ for general bounded degree graphs. Given that Vigoda's algorithm for sampling random colorings with $\frac{11}{6} \Delta$ colors has been known for many years, it is somewhat surprising that it (and our subsequent improvements) can be extended to the list-coloring setting in a fairly straightforward manner:

\begin{thm}
The Glauber dynamics for sampling $k$-list-colorings is rapidly mixing with mixing time $O(n^2)$, for any $k \geq (\frac{11}{6} - \epsilon_0)\Delta$ where $\epsilon_0>0$ is the same constant from Theorem~\ref{thm:main}.\label{thm:listcolorings}
\end{thm}

The key step in our proof is to introduce a notion of flip dynamics for list colorings where a Kempe component is flipped only if both colors appear in all lists of vertices of the component; we call such components \emph{flippable}. While such a chain seems natural in hindsight, we are not aware of it appearing anywhere in the literature. Roughly speaking, introducing this distinction between flippable and un-flippable Kempe components requires introducing additional constraints to Vigoda's linear program. The point is that none of these extra configurations are extremal under the flip parameters used to prove Theorem~\ref{thm:main}, so both proofs of Theorem~\ref{thm:main} carry over to prove Theorem~\ref{thm:listcolorings}.

\subsection{Further Discussion}

Here we survey some previous uses of multi-step coupling, alternative metrics, and linear programming in approximate sampling. In terms of multi-step coupling for sampling colorings, the first improvements to Vigoda's bound for graphs of large degree and girth by Dyer and Frieze \cite{dyer2003randomly} used a burn-in method. Hayes and Vigoda \cite{hayes2003non} gave a sophisticated method based on looking into the future to prevent singly blocked updates in Jerrum's maximal coupling. Non-Markovian couplings have also been used to get $O(n \log n)$ mixing time \---- rather than $O(n^2)$ \---- for Glauber dynamics with $k = (2 - \epsilon)\Delta$ colors \cite{dyer2001extension, hayes2007variable}. The crux of these last results is to terminate the coupling at a random stopping time corresponding to the first point at which something interesting happens in the coupled evolution, e.g. the Hamming distance changing. This \emph{variable-length coupling} approach is also one of the approaches we take.

On the other hand, to our knowledge alternative metrics have found only one application in previous works on sampling colorings, namely in the analysis of the ``scan'' chain for sampling colorings of bipartite graphs in \cite{bordewich2006stopping}. That said, path coupling using alternative metrics has found success in other problems in the approximate sampling literature \cite{luby1999fast,bubley1998faster,bordewich2006stopping}. Bordewich et al. \cite{bordewich2006stopping} gave evidence that the multi-step stopping time-based approach can be captured by one-step coupling with an appropriate metric, though the metric they use to establish this connection is itself based on stopping times. In an orthogonal direction, another take on the question of designing metrics for one-step coupling is given in~\cite{hayes2015randomly} via the spectral radius of the adjacency matrix and in~\cite{efthymiou2016convergence} via the Jacobian of the belief propagation operator.

Lastly, we remark that before Vigoda's work \cite{vigoda2000improved}, Bubley et al. \cite{bubley1998beating} used linear programming to show that Glauber dynamics is rapidly mixing with five colors on graphs with maximum degree three. Their approach however required solving several hundred linear programs, and was subject to ``combinatorial explosion" as a function of the degree.

\subsection*{Organization}

\sitan{In Section~\ref{sec:prelim} we review the basics of path coupling, briefly summarize Vigoda's Markov chain and coupling analysis, and recall some basic notions in variable-length coupling. In Section~\ref{sec:LP}, we interpret Vigoda's one-step coupling analysis as implicitly solving a linear program and identify a family of worst-case neighborhoods which is tight at $k > \frac{11}{6}\Delta$ for one-step coupling of the flip dynamics with respect to the Hamming metric. In Section~\ref{sec:CM}, we give a proof of Theorem~\ref{thm:main} via variable-length coupling. In Section~\ref{sec:DPP}, we give a proof via one-step coupling with an alternative metric. In Section~\ref{sec:list}, we overview how to extend our techniques to prove Theorem~\ref{thm:listcolorings}. }

\sitan{In Appendix~\ref{app:vigodareview}, we give a self-contained exposition of the details of Vigoda's coupling analysis. In Appendix~\ref{app:obsproof}, we provide the set of configurations in Vigoda's analysis that are extremal under his choice of flip parameters. In Appendix~\ref{app:CM} and \ref{app:DPP}, we supply proofs of technical lemmas that appeared in Sections~\ref{sec:CM} and \ref{sec:DPP} respectively. In Appendix~\ref{app:list}, we complete the proof of Theorem~\ref{thm:listcolorings} that was sketched in Section~\ref{sec:list}.}

\section{Preliminaries}
\label{sec:prelim}

Let $\N_0$ denote the set of non-negative integers. In a graph $G = (V,E)$, for vertex $v\in V$ define $N(v)$ to be the set of neighbors of $v$ and $\Delta(v)$ to be the degree of vertex $v$. Given a (proper or improper) coloring $\sigma: V\to[k]$, define $A_{\sigma}(v)$ to be the set of colors \emph{available to $v$}, i.e. the set of colors $c$ for which no neighbor of $v$ is colored $c$. Given a Markov chain with transition probability matrix $P$ on finite state space $\Omega$ and initial state $\sigma^{(0)}$, denote the distribution of state $\sigma^{(t)}$ at time $t$ by $P^t(\sigma^{(0)},\cdot)$. Denote the stationary distribution of an ergodic Markov chain by $\pi$, let 
$$\tau_{mix} (\epsilon):= \max_{\sigma^{(0)}\in\Omega}\min\{t: \tvd(P^{t'}(\sigma^{(0)},\cdot),\pi)\le\epsilon \ \forall t'\ge t\},$$ where $\tvd(\cdot,\cdot)$ is the total variation distance, and define the \emph{mixing time} $\tau_{mix}$ to be $\tau_{mix}(1/2e)$.

The state spaces we will be interested in are $\Omega = [k]^V$ the set of all colorings (or \emph{labellings}) of $G$, and $\Omega^*$ the set of \emph{proper} colorings of $G$.

\subsection{The Flip Dynamics}
\label{subsec:introflip}

The Markov chain we use is a variant of the Wang-Swendsen-Koteck\'{y} (WSK) algorithm \cite{wang1989antiferromagnetic} studied in \cite{vigoda2000improved}, which we define below. In a (proper or improper) coloring $\sigma$ of a graph $G$, for vertex $v$ and color $c$ let the \emph{Kempe component} $S_{\sigma}(v,c)$ denote the set of vertices $w$ for which there exists an alternating path between $v$ and $w$ using only the colors $c$ and $\sigma(v)$. Under this definition, $S_{\sigma}(v,\sigma(v)) = \emptyset$. The motivation for this definition is that if $\sigma$ is proper, then if one \emph{flips} $S_{\sigma}(v,c)$, i.e. changes the color of all $\sigma(v)$-colored vertices in $S_{\sigma}(v,c)$ to $c$ and that of all $c$-colored vertices in $S_{\sigma}(v,c)$ to $\sigma(v)$, the resulting coloring is still proper.

\begin{defn}
	Let $\vec{p} = \{p_\alpha\}_{\alpha\in \N_0}$ be a collection of \emph{flip parameters}. The \emph{flip dynamics} is a random process generating a sequence of colorings $\sigma^{(0)},\sigma^{(1)}, \sigma^{(2)},\dots$ of $G$ where $\sigma^{(0)}$ is an arbitrary coloring in $\Omega$ and $\sigma^{(t)}$ is generated from $\sigma:=\sigma^{(t-1)}$ as follows:\begin{enumerate}
		\item Select a random vertex $v^{(t)}$ and a random color $c^{(t)}$.
		\item Let $\alpha = |S_{\sigma}(v^{(t)},c^{(t)})|$ and flip $S_{\sigma}(v^{(t)},c^{(t)})$ with probability $p_{\alpha}/\alpha$.
	\end{enumerate}

	The reason for the $p_{\alpha}/\alpha$ term is that we have a nice equivalent way of formulating the flip dynamics. Let $\mathcal{S}_{\sigma}$ denote the family of all Kempe components in $G$ under the coloring $\sigma$, i.e. all $S\subset V$ for which there exist $v,c$ such that $S = S_{\sigma}(v,c)$. Here we emphasize that $\mathcal{S}_{\sigma}$ is a multiset.\footnote{For each color $c$ not in the neighborhood of $v$, there exists a distinct component $S_{\sigma}(v,c) = \{v\}$ in $\mathcal{S}_{\sigma}$.}

Then, $\sigma^{(t)}$ is generated from $\sigma:=\sigma^{(t-1)}$ as follows:\begin{enumerate}
		\item Pick any component $S^{(t)}\in\mathcal{S}_\sigma$, each with probability $1/nk$.
		\item Let $\alpha = |S^{(t)}|$ and flip $S^{(t)}$ with probability $p_{\alpha}$. 
	\end{enumerate}
\end{defn}

Because the flip dynamics embeds Glauber dynamics, it is ergodic on the space of proper colorings for every $k\ge\Delta + 2$. As every improper coloring has a positive probability of eventually being transformed into a proper one, the space of proper colorings is the only closed subset of the space of all colorings. The following holds as a consequence:

\begin{lem}[\cite{vigoda2000improved}]
	The stationary distribution of the flip dynamics is the uniform distribution over proper colorings of $G$.
\end{lem}

The WSK algorithm corresponds to the choice of $p_{\alpha} = \alpha$ for all $\alpha\in\N_0$. For the purposes of path coupling, it turns out one only needs to flip Kempe components of size at most some absolute constant $N_{max}$ (in Vigoda's chain, $N_{max} = 6$), and this ``local'' nature of the flip dynamics will simplify the analysis.
Henceforth, we will take $p_0 = 0, p_1 = 1$\footnote{Note that one must set $p_1 = 1$ because otherwise, by rescaling all flip parameters by a factor of $1/p_1$, the mixing time simply scales by a factor of $1/p_1$.}, and $0\le p_{\alpha+1}\le p_{\alpha}\le 1$ for all $\alpha\ge 1$.

\subsection{Path Coupling}

Coupling is a common way to bound the mixing time of Markov chains. A $T$-step coupling for a Markov chain with transition matrix $P$ and state space $\Omega$ defines for every initial $(\sigma^{(0)},\tau^{(0)})\in\Omega^2$ a stochastic process $(\sigma^{(t)},\tau^{(t)})$ such that the distribution of $\sigma^{(T)}$ (resp. $\tau^{(T)}$) is the same as $P^T(\sigma^{(0)},\cdot)$. (resp. $P^T(\tau^{(0)},\cdot)$). The \emph{coupling inequality} states that for any starting point $\sigma^{(0)}$ for the Markov chain, \begin{equation}\tvd(\sigma^{(t)},\pi)\le\max_{\tau^{(0)}}\Pr[\sigma^{(T)}\neq\tau^{(T)}].\label{eq:coupineq}\end{equation}

We will think of $T$-step couplings as random functions $\Omega^2\to\Omega^2$, so we will denote them by $(\sigma^{(0)},\tau^{(0)})\mapsto(\sigma^{(T)},\tau^{(T)})$, or more succinctly, $(\sigma,\tau)\mapsto(\sigma',\tau')$. For an initial pair of colorings $\sigma,\tau$, we say that a coupling \emph{$\gamma$-contracts for $(\sigma,\tau)$} for some $\gamma \in (0,1)$ and metric $d$ on $\Omega$ if it satisfies \begin{equation}\E[d(\sigma',\tau')]\le \gamma\,d(\sigma,\tau).\label{eq:contract}\end{equation} If there exist $\alpha>0$ and a coupling that $(1-\alpha)$-contracts for all $(\sigma,\tau)$, then one can show that $\tau_{mix} = O(T\log(D)/\alpha)$, where $D$ is the diameter of $\Omega$ under $d$.

For the rest of this subsection, we specialize our discussion of coupling to the setting of sampling colorings. 
%Recall that $\Omega^*$ and $\Omega$ denote the spaces of proper and improper colorings of $G$ respectively. 
Fix any Markov chain over $\Omega$ whose stationary distribution is the uniform distribution over $\Omega^*$, e.g. the Glauber or flip dynamics, and denote it by $\sigma\mapsto\sigma'$.

In complicated state spaces like the space of all proper colorings of a graph, it is often tricky to construct couplings that give good bounds on mixing time. Path coupling, introduced in \cite{bubley1997path}, is a useful tool for simplifying this process: rather than define $(\sigma,\tau)\mapsto(\sigma',\tau')$ for all $(\sigma,\tau)\in\Omega^2$, it is enough to do so for a small subset of initial pairs in $\Omega^2$. This subset is specified by a \emph{pre-metric}.

\begin{defn}\label{defn:premetric}
	A \emph{pre-metric} on $\Omega$ is a pair $(\Gamma,\omega)$ where $\Gamma$ is a connected, undirected graph with vertex set $\Omega$, and $\omega$ is a function that assigns positive, real-valued weights to edges $\sigma\tau$ of $\Gamma$ such that for every edge $\sigma\tau$, $\omega(\sigma\tau)$ is the minimum weight among all paths between $\sigma$ and $\tau$.

	We will often refer to a pair of adjacent colorings $\sigma,\tau$ in $\Gamma$ as a \emph{neighboring coloring pair}, denoted by $(G,\sigma,\tau)$. Where the context is clear, we omit $G$ and refer to neighboring coloring pairs as $(\sigma,\tau)$.

	For any $\tilde{\sigma},\tilde{\tau}\in\Omega$, let $P_{\tilde{\sigma},\tilde{\tau}}$ denote the set of simple paths $\vec{\phi} = (\phi_0,\dots,\phi_s)$ in $\Gamma$ where $\phi_0 = \tilde{\sigma}$ and $\phi_s = \tilde{\tau}$. The metric $d$ induced by the pre-metric $(\Gamma,\omega)$ is defined by $d(\tilde{\sigma},\tilde{\tau}) := \min_{\vec{\phi}\in P_{\tilde{\sigma},\tilde{\tau}}}\sum^s_{i=1}\omega(\phi_{i-1}\phi_i).$%\label{defn:premetric}
\end{defn}
\begin{example}
	If $\Gamma$ is the graph with vertex set $\Omega$ and edges between colorings $\sigma,\tau$ which differ on exactly one vertex, and $\omega$ assigns weight 1 to all edges of $\Gamma$, then the metric induced by $(\Gamma,\omega)$ is simply the Hamming distance $d_H$, i.e. the total number of vertices on which two colorings differ.
\end{example}

\begin{thm}[\cite{bubley1997path}]\label{thm:path_coupling}
	Let $(\Gamma,\omega)$ be a pre-metric on $\Omega$ where $\omega$ takes on values in $(0,1]$, and let $d$ be the metric it induces. If the Markov chain $\sigma\mapsto\sigma'$ has a coupling $(\sigma,\tau)\mapsto(\sigma',\tau')$ defined for all neighboring coloring pairs $(\sigma,\tau)$ that $(1-\alpha)$-contracts for some $\alpha > 0$, then there exists a coupling defined for \emph{all} pairs of colorings $(\sigma,\tau)\in\Omega^2$ which $(1-\alpha)$-contracts.
\end{thm}

\begin{remark}\label{remark:big_space}
	The reason we need to extend the state space from $\Omega^*$ to $\Omega$ is already apparent in the context of path coupling with the Hamming metric. Given two colorings $\sigma,\tau$ for which $d_H(\sigma,\tau) = \ell$, there does not necessarily exist a sequence of \emph{proper} colorings  $\sigma = \sigma_0, \sigma_1, \dots, \sigma_{\ell} = \tau$ of length $\ell$ for which $(\sigma_i,\sigma_{i+1})$ are neighboring coloring pairs for all $0\le i < \ell$. However, there certainly exist such sequences if we allow the colorings to be improper.

\trim{	This is a standard fix that comes up in typical applications of path coupling to sampling colorings. As noted by the authors in these applications, the stationary distribution for this Markov chain is still the uniform distribution over proper colorings. This is because if we start at a proper coloring, we only ever visit proper colorings, and if start at an improper coloring, we eventually reach a proper coloring, so the support of the stationary distribution consists only of proper colorings.}
\end{remark}

\subsection{Variable-Length Coupling}

In this section we review the basics of variable-length coupling. This will only be relevant to Section~\ref{sec:CM} which gives the first of our two proofs of Theorem~\ref{thm:main}.

Jerrum's $k> 2\Delta$ bound \cite{jerrum1995very} and Vigoda's $k > \frac{11}{6}\Delta$ bound \cite{vigoda2000improved} can both be proved via \emph{one-step} path couplings. Yet there is evidence that multi-step couplings can sometimes be stronger than one-step couplings. As shown in \cite{kumar2001coupling}, there exist Markov chains for some sampling problems where one-step coupling analysis is insufficient. \cite{czumaj1999delayed} used multi-step coupling for a tighter analysis of a Markov chain for sampling random permutations, and the celebrated $k\ge(1+\epsilon)\Delta$ result of \cite{hayes2003non} for $\Omega(\log n)$-degree graphs uses a multi-step coupling which is constructed by looking into future time steps.

There are also several other works that carried out a multi-step coupling analysis by looking at one-step coupling over multiple time steps \cite{dyer2001extension,dyer2002very,hayes2007variable} and obtained slight improvements over Jerrum's $k> 2\Delta$ bound by terminating path coupling of the Glauber dynamics at a random stopping time. In the literature, this is known as \emph{variable-length coupling}, and this is the approach we take in Section~\ref{sec:CM}, but for the flip dynamics.

\begin{defn}[Definition 1 in \cite{hayes2007variable}]
	For every initial neighboring coloring pair $(\sigma^{(0)},\tau^{(0)})$, let $(\overbar{\sigma},\overbar{\tau},\Tstop)$ be a random variable where $\Tstop$ is a nonnegative integer, and $\overbar{\sigma},\overbar{\tau}$ are sequences of colorings $(\sigma^{(0)},\dots,\sigma^{(\Tstop)})\in\Omega^{\Tstop}$ and $(\tau^{(0)},\dots,\tau^{(\Tstop)})\in\Omega^{\Tstop}$ respectively. We say that $(\overbar{\sigma},\overbar{\tau},\Tstop)$ is a \emph{variable-length path coupling} if $\overbar{\sigma},\overbar{\tau}$ are faithful copies of the Markov chain in the following sense.

	For $(\sigma^{(0)},\tau^{(0)})$ and $t\ge 0$, define the random variables $\sigma_t, \tau_t$ via the following experiment: 1) sample $(\overbar{\sigma},\overbar{\tau},\Tstop)$, 2) if $t\le T$, define $\sigma_t = \sigma^{(t)}, \tau_t = \tau^{(t)}$, 3) if $t > \Tstop$, then sample $\sigma_t$ and $\tau_t$ from $P^{t-\Tstop}(\sigma^{(\Tstop)},\cdot)$ and $P^{t-\Tstop}(\tau^{(\Tstop)},\cdot)$ respectively.

	We say that $\overbar{\sigma}$ (resp. $\overbar{\tau}$) is a \emph{faithful copy} if for every neighboring coloring pair $(\sigma^{(0)},\tau^{(0)})$ and $t\ge 0$, $\sigma_t$ and $\tau_t$ defined above are distributed according to $P^t(\sigma^{(0)},\cdot)$ and $P^t(\tau^{(0)},\cdot)$ respectively.\label{def:varlength}
\end{defn}

When $\Tstop$ is always equal to some fixed $T$, this is just the usual notion of $T$-step path coupling. The following extends the path coupling theorem of \cite{bubley1997path} to variable-length path coupling.

\begin{thm}[Corollary 4 of \cite{hayes2007variable}]
	Let $\epsilon > 0$. For a variable-length path coupling $(\overbar{\sigma},\overbar{\tau},\Tstop)$, let $$\alpha:= 1 - \max_{\sigma^{(0)},\tau^{(0)}}\E[d_H(\sigma^{(\Tstop)},\tau^{(\Tstop)})],\quad W :=\max_{\sigma^{(0)},\tau^{(0)},t\le\Tstop} d_H(\sigma^{(t)},\tau^{(t)}), \quad\beta:= \max_{\sigma^{(0)},\tau^{(0)}}\E[\Tstop].$$ If $\alpha > 0$, then $\tau_{mix}(\epsilon)\le 2\left\lceil 2\beta W/\alpha\right\rceil\cdot\left\lceil\ln(n/\epsilon)/\alpha\right\rceil.$\label{thm:hayesvigoda}
\end{thm}

\subsection{Vigoda's Greedy Coupling}
\label{subsec:introonestep}

In Appendix~\ref{app:vigodareview}, we give a self-contained review of the one-step path coupling analysis from \cite{vigoda2000improved}. We encourage readers unfamiliar with his analysis to simply read the entirety of Appendix~\ref{app:vigodareview} in place of this subsection, as here we only introduce relevant notation and state the parts of his analysis that are essential to our proofs.

Fix a neighboring coloring pair $(G,\sigma,\tau)$. For $c\in[k]$, let $U_c$ denote the set of neighbors of $v$ (in either coloring) that are colored $c$, and let $\delta_c = |U_c|$. We will sometimes denote the vertices in $U_c$ by $\{u^c_1,\dots,u^c_{\delta_c}\}$; where $c$ is clear from context, we will simply denote these by $\{u_1,\dots,u_{\delta_c}\}$.

Note that the symmetric difference $\mathcal{D} = \calS_{\sigma}\Delta\calS_{\tau}$ is precisely the Kempe components $S_{\sigma}(u^c_i,\tau(v))$ and $S_{\sigma}(v,c)$ in $\sigma$ and the Kempe components $S_{\tau}(u^c_i,\sigma(v))$ and $S_{\tau}(v,c)$ in $\tau$, for all colors $c$ appearing in the neighborhood of $v$ and all $i\in[\delta_c]$. All other Kempe components are shared between $\sigma$ and $\tau$, so for those, it is enough to use the identity coupling. Note that for colors $c\neq\sigma(v),\tau(v)$ not appearing in $N(v)$, the identity coupling then matches the flip of $S_{\sigma}(v,c)$ to that of $S_{\tau}(v,c)$ so that the two colorings of $G$ become identical. 

So the main concern is how to couple the flips of components in $\mathcal{D}$. It is easy to see that for $c\neq \sigma(v),\tau(v)$, \begin{equation}S_{\sigma}(v,c) = \left(\bigcup^{\delta_c}_{i=1}S_{\tau}(u^c_i,\sigma(v))\right)\cup\{v\} \ \ \ \ S_{\tau}(v,c) = \left(\bigcup^{\delta_c}_{i=1}S_{\sigma}(u^c_i,\tau(v))\right)\cup\{v\},\label{eq:sigmadecomp}\end{equation} Purely for simplicity of exposition, we will assume that the sets $S_{\sigma}(u^c_i,\tau(v))$ are distinct as $i$ varies, and likewise for $S_{\tau}(u^c_i,\sigma(v))$, and we will only consider $c\neq\sigma(v),\tau(v)$, referring the reader respectively to Remarks \ref{remark:multiplicity} and \ref{remark:specialcase} in Appendix~\ref{app:vigodareview} for the missing details. We remark that the extra cases of $c = \sigma(v),\tau(v)$ are the primary place where one needs to be careful about the fact that neighboring coloring pairs $\sigma,\tau$ need not be proper.

For $c$ such that $\delta_c > 0$, define $A_c := |S_{\sigma}(v,c)|$, $B_c := |S_{\tau}(v,c)|$, $a^c_i = |S_{\tau}(u^c_i,\sigma(v))|$, and $b^c_i = |S_{\sigma}(u^c_i,\tau(v))|$. Define the vectors $\vec{a}^c := (a^c_i: i\in[\delta_c])$ and $\vec{b}^c := (b^c_i: i\in[\delta_c])$. 
We say that a neighboring coloring pair $(G,\sigma,\tau)$ has \emph{configuration $(A_c,B_c;\vec{a}^c,\vec{b}^c)$ of size $\delta_c$} for color $c$.
Also define $a^c_{max} = \max_i a^c_i$ and denote a maximizing $i$ by $i^c_{max}$. Likewise define $b^c_{max} = \max_j b^c_j$ and denote a maximizing $j$ by $j^c_{max}$. When the color $c$ is clear from context, we refer to these as $A,B,a_i,b_i,\vec{a},\vec{b}, a_{max}, i_{max},b_{max},j_{max}$. 

When $c\neq\sigma(v),\tau(v)$, note that \begin{equation}A_c = 1+ \sum_i a^c_i, \ \ \ \ \ B_c= 1 + \sum_i b^c_i.\label{eq:trivABbound}\end{equation}
While $S_{\sigma}(v,c)$ and $S_{\tau}(v,c)$ can be quite different, we do know that $S_{\sigma}(v,c) \supset S_{\tau}(u_i,\sigma(v))$. The idea of Vigoda's coupling is thus to greedily couple the flips of the largest components, i.e. $S_{\sigma}(v,c), S_{\tau}(v,c)$, to the flips of the next largest components, i.e. $S_{\tau}(u_{i_{max}},\sigma(v)), S_{\sigma}(u_{j_{max}},\tau(v))$, and then to couple together as closely as possible the flips of $S_{\sigma}(u_i,\tau(v))$ and $S_{\tau}(u_i,\sigma(v))$ for each $i\in[\delta_c]$. Henceforth, we will refer to this coupling as the \emph{greedy coupling}.

For any configuration $(A,B;\vec{a},\vec{b})$, define
\begin{equation}
	H(A,B;\vec{a},\vec{b}) := (A - a_{max} - 1)p_A + (B - b_{max} - 1)p_B + \sum_i a_iq_i + b_iq'_i - \min(q_i,q'_i), \label{eq:H}
\end{equation}
where $q_i = p_{a_i} - p_A\cdot \mathbbm{1}_{i = i_{max}}$ and $q'_i = p_{b_i} - p_B\cdot\mathbbm{1}_{i=j_{max}}$. In \cite{vigoda2000improved} it is shown that under this greedy coupling, for $c\neq\sigma(v),\tau(v)$ appearing in the neighborhood of $v$, \begin{equation}kn\cdot\E[\mathbbm{1}_{X_c}\cdot(d_H(\sigma',\tau') - 1)]\le H(A_c,B_c;\vec{a}^c,\vec{b}^c),\label{eq:mainvigodaineqcomponent}\end{equation} where $X_c$ is the event that the coupling flips components in $\mathcal{D}_c$ in both colorings. Therefore:

\begin{lem}[\cite{vigoda2000improved}]\label{lem:vigoda}
	Let $(\sigma,\tau)\mapsto(\sigma',\tau')$ be the greedy coupling. Then \begin{equation}\E[d_H(\sigma',\tau') - 1] \le\frac{1}{nk}\left(-|\{c:\delta_c = 0\}| + \sum_{c: \delta_c\neq 0}H(A_c,B_c;\vec{a}^c,\vec{b}^c)\right).\label{eq:mainvigodaineq}\end{equation}
\end{lem}
The function $H$ implicitly depends on the choice of flip parameters $\{p_{\alpha}\}$, while $(A_c,B_c;\vec{a}^c,\vec{b}^c)$ depends on $(G,\sigma,\tau)$. The remaining analysis in \cite{vigoda2000improved} once \eqref{eq:mainvigodaineq} has been deduced essentially boils down to picking a good set of flip parameters.

\section{Linear Programming and Choice of Flip Parameters}
\label{sec:LP}

The key idea to choose the flip parameters is to cast the problem as an instance of linear programming with variables $\{p_\alpha\}_{\alpha\in\N_0}$ to minimize the right-hand side of \eqref{eq:mainvigodaineq} over all neighboring coloring pairs $(G,\sigma,\tau)$ where $G$ has maximum degree $\Delta$ and $\sigma,\tau$ are $k$-colorings. 
The following gives terminology for quantifying over all such $(G,\sigma,\tau)$.

\begin{defn}\label{def:realizable}
	A configuration $(A,B;\vec{a},\vec{b})$ is \emph{realizable} if there exists a neighboring coloring pair $(G,\sigma,\tau)$ and color $c$ such that $(A,B;\vec{a},\vec{b}) = (A_c,B_c;\vec{a}^c,\vec{b}^c)$.
\end{defn}

Vigoda's remaining analysis can thus be interpreted as solving the following linear program.

\begin{lp}
For variables $\{p_{\alpha}\}_{\alpha\in\N_0}$ and $\lambda$, minimize $\lambda$ subject to: $p_0 = 0\le p_{\alpha}\le p_{\alpha-1}\le p_1 = 1$ for all $\alpha\ge 2$, and $H(A,B;\vec{a},\vec{b})\le -1 + \lambda\cdot m$ for all realizable $(A,B;\vec{a},\vec{b})$ of size $m$.
 \label{def:prelp}
 \end{lp}

There are three minor issues with this linear program: $(a)$ the linear program has an infinite number of variables, $(b)$ it has an infinite number of constraints, and $(c)$ given $\vec{a},\vec{b}$, it is not immediately obvious how to enumerate all $A,B$ for which $(A,B;\vec{a},\vec{b})$ is realizable.

Vigoda handles $(a)$ by restricting to flips of components of size at most $N_{max}$, i.e. by fixing some small constant $N_{max}$ and insisting that \begin{equation}p_{\alpha} = 0 \ \forall \alpha > N_{max}.\label{eq:NMax}\end{equation} We emphasize that this still leaves an infinite number of constraints as $m$ can be unbounded.

He handles $(b)$ by shrinking the feasible region via the following observation.

\begin{lem}[\cite{vigoda2000improved}]
$H(A,B;\vec{a},\vec{b})\le (A - 2)p_A + (B - 2)p_B + \sum_i (a_ip_{a_i} + b_ip_{b_i} - \min(p_{a_i},p_{b_i}))$.\label{lem:crude}
\end{lem}

Whereas $f(u_i)$ are linear functions of $p_{a_i},p_{b_i},p_A,p_B$, the summands in the upper bound of Lemma~\ref{lem:crude} are simply linear functions of $p_{a_i},p_{b_i}$. So we can pick some $m^*$ (\cite{vigoda2000improved} picks $m^* = 3$) and replace the infinitely many constraints for which $m\ge m^*$ in Linear Program~\ref{def:prelp} with finitely many constraints to optimize the right-hand side of Lemma~\ref{lem:crude}.

Finally, Vigoda implicitly handles $(c)$ as follows. To cover all constraints corresponding to realizable $(A_c,B_c;\vec{a}^c,\vec{b}^c)$ with $c\neq\sigma(v),\tau(v)$  include \begin{equation}H(A,B;\vec{a},\vec{b})\le -1 + \lambda\cdot m,\label{eq:mainconstraint}\end{equation} for all $1\leq m < m^*$ and $(A,B;\vec{a},\vec{b})$ for which $\vec{a},\vec{b}\in\{0,1,\dots,N_{max}\}^m\backslash\{(0,0,\dots 0)\}$ and $A,B$ satisfy \eqref{eq:trivABbound}. As we will discuss in Appendix~\ref{app:realizable}, to cover all constraints corresponding to $c = \sigma(v)$ and $c = \tau(v)$, it is enough to include \begin{equation}(B - b_m)p_B + \sum^{m-1}_{i=1}b_ip_{b_i}\le -1 + \lambda\cdot m,\label{eq:Hforsigmav}\end{equation} for all $2\le m < m^*$, $0\le b_1\le\cdots\le b_m\le N_{max}$ where $b_m > 0$, and $B = \sum_i b_i$, as well as \begin{equation}\alpha\cdot p_{\alpha}\le 1, \ \text{ for all }\alpha\in\N_0.\label{eq:pap}\end{equation}

Concretely, we have the following relaxation of Linear Program~\ref{def:prelp}.

\begin{lp}
	Fix some $N_{max}\ge 1$ and $m^*\ge 2$. For variables $\{p_{\alpha}\}_{\alpha\in\N_0}$ and $\lambda$, and dummy variables $x, y$, minimize $\lambda$ subject to the following constraints: $p_0 =  0\le p_{\alpha}\le p_{\alpha-1}\le p_1 = 1$ for all $\alpha\ge 2$, constraint \eqref{eq:NMax}, constraint~\eqref{eq:mainconstraint} for all $\vec{a},\vec{b}\in\{0,1,\dots,N_{max}\}^m\backslash \{(0,\dots,0)\}$ with $1\leq m < m^*$ and $A,B$ satisfying \eqref{eq:trivABbound}, constraint \eqref{eq:Hforsigmav} for all $2\le m < m^*$ and $0\le b_1\le\cdots\le b_m\le N_{max}$ where $b_m > 0$ and $B = \sum_i b_i$, constraint \eqref{eq:pap}, and constraints \begin{align}x&\ge (A - 2)p_A\nonumber \\
	y&\ge a\cdot p_a + b\cdot p_b - \min(p_a,p_b)\nonumber \\
	-1 + \lambda\cdot m^* &\ge 2x + m^*\cdot y\label{eq:approxconstraint}\end{align} for every $A,a,b$ satisfying $0\le A\le 1+N_{max}$ and $0\le a < b\le N_{max}$.
	\label{def:lp}
\end{lp}

\begin{remark}
	Note that we only add in constraints for the upper bound of Lemma~\ref{lem:crude} in the case of $m = m^*$ (constraint \eqref{eq:approxconstraint}) because the constraints for $m = m^*$ implies the constraints for $m > m^*$.
\end{remark}

\begin{lem}\label{lem:contractive}
Let $\lambda_2^*$ be the objective value of Linear Program~\ref{def:lp}.
If $\lambda_2^*\ge 1$ and $k > \lambda_2^*d$, then there exist flip parameters $\{p_{\alpha}\}_{\alpha\in\N_0}$ for which $\E[d_H(\sigma',\tau') - 1] < 0$ for all neighboring coloring pairs $(G,\sigma,\tau)$, where $(\sigma,\tau)\mapsto(\sigma',\tau')$ is the greedy coupling.\label{lem:vigodaoptimal}
\end{lem}

\begin{proof}
	Let $\{p_{\alpha}\}_{\alpha\in\N_0}$ be flip parameters achieving objective value $\lambda_2^*$. By \eqref{eq:mainvigodaineq} and Lemma~\ref{lem:crude}, we have that \[nk\cdot \E[d_H(\sigma',\tau') - 1]\le -|\{c: \delta_c = 0\}| + \sum_{c:\delta_c\neq 0} H(A_c,B_c;\vec{a}^c,\vec{b}^c)\le -k + \lambda^*_2\cdot d < 0.\qedhere\]
\end{proof}

\trim{We will show the converse: if $k < \lambda_2^*d$, then for any flip parameters $\{p_{\alpha}\}_{\alpha\in\N_0}$ there exists a neighboring coloring pair $(G,\sigma,\tau)$, such that $\E[d_H(\sigma',\tau')] > 1$ under \emph{any} one-step coupling. We remark that it is \emph{a priori} unclear how to conclude this even about Linear Program~\ref{def:prelp} and the greedy coupling because its constraint set is infinite. It is even less clear how to conclude this about Linear Program~\ref{def:lp} because it is a relaxation of Linear Program~\ref{def:prelp}, so even their objective values need not agree.}

In \cite{vigoda2000improved}, Vigoda shows that for $m^* = 3$, $N_{max} = 6$, and the following flip parameters, Linear Program~\ref{def:lp} attains a value $\lambda_2^*=11/6$: \begin{equation}
	p_1 = 1,\, p_2 = 13/42,\, p_3 = 1/6,\, p_4 = 2/21,\, p_5 = 1/21,\, p_6 = 1/84 \text{ and } p_{\alpha} = 0 \ \forall \alpha\ge 7.\label{eq:vigodasprobs}
\end{equation}
%\sitan{In \cite{vigoda2000improved}, Vigoda exhibts for $m^* = 3$ and $N_{max} = 6$ an assignment of flip parameters for which Linear Program~\ref{def:lp} attains a value of 11/6.} 
Not only is this assignment a feasible solution to Linear Program~\ref{def:lp}, but it happens to be an optimal solution of Linear Programs~\ref{def:prelp} and~\ref{def:lp}. We will be interested in which constraints are tight under such optimal solutions, as they will guide us to the configurations key to our proofs. 

\begin{defn}
Given a feasible solution $\mathbf{p}$ of Linear Program~\ref{def:prelp} with objective value $\lambda$, a configuration $(A,B;\vec{a},\vec{b})$ of size $m$ is \emph{$\mathbf{p}$-extremal} (or simply \emph{extremal} if the flip parameters are clear from the context) if  $H(A,B;\vec{a},\vec{b})= -1 + \lambda\cdot m$ under the assignation $\mathbf{p}$.
\end{defn}

Vigoda's proof in \cite{vigoda2000improved} already implicitly gives a collection of six extremal configuration under the assignment~\eqref{eq:vigodasprobs}~(see Observation~\ref{obs:slack} in Appendix~\ref{app:obsproof}). It turns out that among these tight constraints, two of them are already enough to force the objective value of Linear Program~\ref{def:lp} to be $11/6$.

Consider the following linear program obtained by restricting to constraints~\eqref{eq:mainconstraint} associated with configurations $(3,2;(2),(1))$ and $(7,3;(3,3),(1,1))$, which are both realizable.
\begin{lp}
For variables $\{p_{\alpha}\}_{\alpha\in\N_0}$ and $\lambda$, minimize $\lambda$ subject to: $p_0 = 0\le p_{\alpha}\le p_{\alpha-1}\le p_1 = 1$ and
	\begin{align*} 
	p_1 + p_2 - 2p_3 - \min(p_1-p_2,p_2-p_3) &\le -1 +\lambda,\\
%	$$p_1 - p_2 + 3p_3 - 3p_4 - \min(p_1 - p_2,p_3 - p_4) \le -1 + \lambda$$ 
%	$$p_1 - p_2 + 4p_4 - 4p_5 - \min(p_1-p_2,p_4-p_5)\le -1 + \lambda$$ 
%	$$2p_1 + 5p_3 - \min(p_1 - p_3,p_3 - p_6)\le -1 + 2\lambda.$$ 
	2p_1 + 5p_3 - \min(p_1 - p_3,p_3 - p_7)&\le -1 + 2\lambda.
	\end{align*}
\label{def:mostbasictight}
\end{lp}
It is easy to check that Linear Program~\ref{def:mostbasictight} also has objective value 11/6, and its constraints are a strict subset of those of Linear Programs~\ref{def:prelp} and \ref{def:lp}, from which we conclude that

\begin{cor}
	The objective values of Linear Program~\ref{def:prelp}, Linear Program~\ref{def:lp} with $N_{max}\ge \new{6}$ and $m^* = 3$, and Linear Program~\ref{def:mostbasictight} are all equal to 11/6.\label{cor:mostbasictightcor}
\end{cor}

Corollary~\ref{cor:mostbasictightcor} allows us to exhibit a family $\mathcal{C}$ of just \emph{two} neighboring coloring pairs $(G,\sigma,\tau)$ for which no one-step coupling, greedy or otherwise, simultaneously contracts with respect to the Hamming metric for $k < (11/6)\Delta$ for all $(G,\sigma,\tau)\in\mathcal{C}$. To clarify, this lemma is not used in the proofs of our main result, but provides intuition for the limitations of one-step coupling with respect to the Hamming metric and motivates our two approaches for circumventing them.

\begin{figure}[ht]
	\centering
	\subcaptionbox{$(G_1,\sigma_1,\tau_1)$}{
\begin{tikzpicture}
    \node[draw,circle,fill=white,text width=0.8cm, scale=0.6] (center) at (0,0) {$\sigma_1(v)$\\ $\tau_1(v)$};
\foreach \phi in {1,...,8}{
    \node[draw,circle,fill=white, scale=1.2] (v_\phi) at (360/8 *\phi +360/8:1.5cm) {$c_\phi$};
         \draw (v_\phi) -- (center);
      }
\foreach \phi in {1,...,8}{
    \node[draw,circle,scale=0.7, fill=white] (u_\phi) at (360/8 * \phi+ 360/8:3cm) {$\sigma_1(v)$};
         \draw (u_\phi) -- (v_\phi);
      }
   \end{tikzpicture}
	}
	\subcaptionbox{$(G_2,\sigma_2,\tau_2)$}{
	   \begin{tikzpicture}
    \node[draw,circle,fill=white,text width=0.8cm, scale=0.6] (center) at (0,0) {$\sigma_2(v)$\\ $\tau_2(v)$};
\foreach \phi in {1,...,4}{
    \node[draw,circle,fill=white, scale=1.2] (v_\phi) at (360/4 * \phi:1.5cm) {$c_\phi$};
         \draw (v_\phi) -- (center);
      }
\foreach \phi in {1,...,4}{
    \node[draw,circle,fill=white, scale=1.2] (v2_\phi) at (360/4 * \phi+360/8:1.5cm) {$c_\phi$};
         \draw (v2_\phi) -- (center);
      }
\foreach \phi in {1,...,16}{
    \node[draw,circle,scale=0.7, fill=white] (u_\phi) at (360/16 * \phi+5*360/32:3cm) {$\sigma_2(v)$};
      }
      \draw (v_1) -- (u_1);
      \draw (v_1) -- (u_2);
      \draw (v2_1) -- (u_3);
      \draw (v2_1) -- (u_4);
      \draw (v_2) -- (u_5);
      \draw (v_2) -- (u_6);
      \draw (v2_2) -- (u_7);
      \draw (v2_2) -- (u_8);
      \draw (v_3) -- (u_9);
      \draw (v_3) -- (u_10);
      \draw (v2_3) -- (u_11);
      \draw (v2_3) -- (u_12);
      \draw (v_4) -- (u_13);
      \draw (v_4) -- (u_14);
      \draw (v2_4) -- (u_15);
      \draw (v2_4) -- (u_16);
   \end{tikzpicture}
	}
	\caption{Neighboring coloring pairs defined in Construction~\ref{constr} for $\Delta=8$.}
	\label{fig:allGs}
\end{figure}
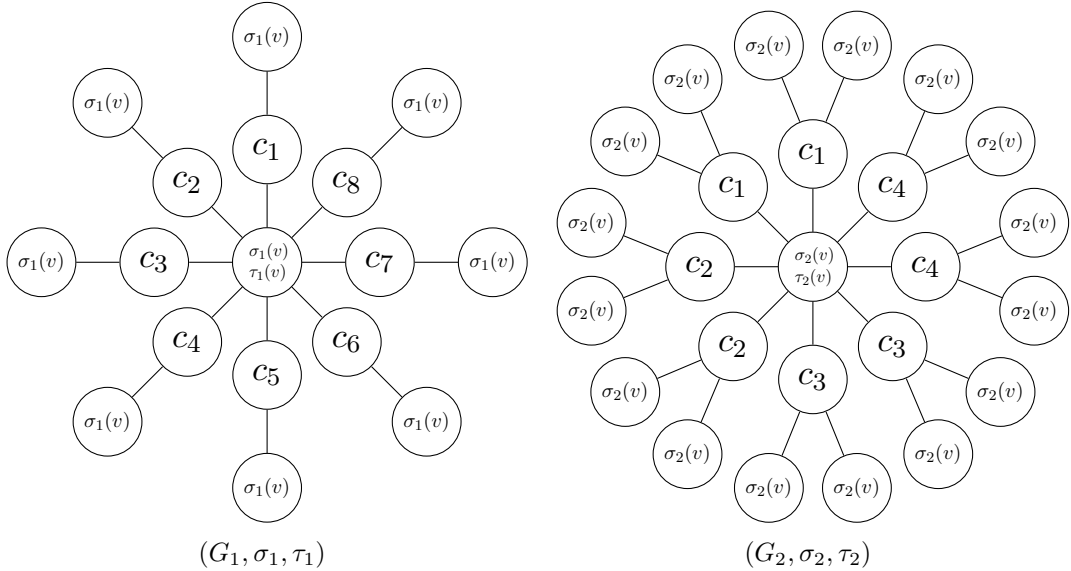
%
%\begin{figure}[ht]
%	\centering
%	\subcaptionbox{$(G_1,\sigma_1,\tau_1)$}{\includegraphics[width=0.45\textwidth]{G2-bw.pdf}}
%	\subcaptionbox{$(G_2,\sigma_2,\tau_2)$}{\includegraphics[width=0.45\textwidth]{G1-bw.pdf}}
%	\caption{Examples of neighboring coloring pairs defined in Construction~\ref{constr}}
%	\label{fig:allGs}
%\end{figure}

\begin{constr}
	Let $G_1$ be the tree of height $2$ rooted at vertex $v$ with $\Delta$ path graphs, each consisting of $2$ other vertices, attached to $v$. Let $u_1,\dots,u_\Delta$ be the neighbors of $v$. In colorings $\sigma_1,\tau_1$, assign two different colors to $v$, $\sigma_1(v)$ and $\tau_1(v)$, assign each $u_i$ the color $c_i\neq \sigma_1(v),\tau_1(v)$, and assign the child of $u_i$ the color $\sigma_1(v)$.
	
	For $\Delta$ even, let $G_2$ be the tree of height two rooted at a vertex $v$ with exactly $\Delta$ children $u_1,\dots,u_\Delta$ such that each $u_i$ has exactly two children $w^i_1$ and $w^i_2$. In colorings $\sigma_{2},\tau_{2}$, assign two different colors to $v$, $\sigma_2(v)$ and $\tau_2(v)$, assign $u_{2j-1}$ and $u_{2j}$ the color $c_j\neq \sigma_2(v),\tau_2(v)$ for $j \in \{1,\dots, \Delta/2\}$, and assign all $w^i_{\ell}$ the color $\sigma_2(v)$.

	Let $\mathcal{C}^* = \{(G_1,\sigma_1,\tau_1),(G_2,\sigma_2,\tau_2)\}$ (see Figure~\ref{fig:allGs}).\label{constr}
\end{constr}

\begin{lem}\label{lem:tight}
	If $k < (11/6)\Delta$, there exists no choice of flip parameters $\{p_{\alpha}\}$ and one-step coupling $(\sigma,\tau)\mapsto(\sigma',\tau')$ for which $\E[d_H(\sigma',\tau') - 1] < 0$ for all $(G,\sigma,\tau)\in\mathcal{C}^*$, where $\mathcal{C}^*$ is defined in Construction~\ref{constr}.\label{lem:Cfamily}
\end{lem}
We defer the proof of this to Appendix~\ref{app:obsproof}. Lemma~\ref{lem:tight} states that it is impossible to design a one-step coupling analysis of the flip dynamics with the Hamming metric that crosses Vigoda's $11/6$ barrier for general graphs (or even for the family of two trees defined by $\mathcal{C}^*$). There are two natural strategies to overcome this problem: use a \emph{multi-step coupling} analysis with the Hamming metric, or use a one-step coupling with an alternative metric. In the remainder of the paper we present two independent proofs of Theorem~\ref{thm:main}. In Section~\ref{sec:CM}, we present the multi-step coupling approach to the proof as devised by Chen and Moitra~\cite{chen2018linear}. In Section~\ref{sec:DPP}, we use the alternative metric approach to the proof as devised by Delcourt, Perarnau and Postle~\cite{delcourt2018rapid}.
For the sake of clarity, we defer the proof of some technical lemmas to Appendices~\ref{app:CM} and~\ref{app:DPP}.
% While the two approaches are different in nature, we make an effort to highlight parallelisms between them.

%\input{LPchoice}

%!TEX root = fullpaper.tex

\section{Proof of Theorem~\ref{thm:main} using Variable-Length Coupling}
\label{sec:CM}

In this section, we prove Theorem~\ref{thm:main} by arguing that a suitably defined variable-length coupling contracts with respect to the Hamming metric $d_H$. Recall from Definition~\ref{defn:premetric} that under the pre-metric inducing the Hamming metric, $(G,\sigma,\tau)$ is a \emph{neighboring coloring pair} if $\sigma,\tau$ are colorings of $G$ that differ on exactly one vertex $v$.

\subsection{Modifying the LP}
\label{subsec:modifyinglp}

For a color $c$, the condition that $(A_c,B_c;\vec{a}^c,\vec{b}^c)$ for a neighboring coloring pair $(G,\sigma,\tau)$ be extremal is a very stringent condition on $(G,\sigma,\tau)$. The hope is that for a suitable notion of ``typical,'' this condition holds for few colors $c$ for a ``typical'' neighboring coloring pair.

\begin{defn}
Given a neighboring coloring pair $(G,\sigma,\tau)$ and a color $c$ appearing in the neighborhood of $v$, then the pair $\sigma,\tau$ is in the state \begin{enumerate}
	\item \Sing{c} if $\delta_c = 1$ and $c\neq\sigma(v),\tau(v)$\footnote{For $c\neq\sigma(v),\tau(v)$, \Sing{c} corresponds to the extremal configuration of size 1, as well as all other configurations which satisfy $\delta_c = 1$. We include these latter configurations just for simplicity of analysis; if we did not do this, it would yield additional improvements upon our main result. Also note that we exclude the cases of $c = \sigma(v),\tau(v)$ from \Sing{c} because in those cases, we have that $A_c = 0$ in which case $(A_c,B_c;\vec{a}^c,\vec{b}^c)$ is not extremal.}
	\item \Bad{c} if $(A_c,B_c;\vec{a}^c,\vec{b}^c)$ is either $(7,3;(3,3),(1,1))$ or $(3,7;(1,1),(3,3))$
	\item \Good{c} otherwise.
\end{enumerate} Moreover, define $N_{sing}(\sigma,\tau)$, $N_{bad}(\sigma,\tau)$, and $N_{good}(\sigma,\tau)$ to be the number of $c$ for which $(G,\sigma,\tau)$ is in state \Sing{c}, \Bad{c}, \Good{c} respectively.\label{def:badgood}
\end{defn}

\begin{obs}
Let $(G,\sigma,\tau)$ be any neighboring coloring pair. For $c = \sigma(v),\tau(v)$, if $\delta_c > 0$, then $\sigma,\tau$ are in state \Good{c}.\label{obs:always}
\end{obs}

\begin{proof}
	If $\delta_c = 1$, then by definition $\sigma,\tau$ are in state \Good{c}. If $\delta_c \ge 2$, then because $A_c = 0$ for $c = \sigma(v),\tau(v)$ by Remark~\ref{remark:specialcase}, $\sigma,\tau$ must be in state \Good{c}.
\end{proof}

$N_{sing}(\sigma,\tau)$ can be large even for a ``typical'' neighboring coloring: consider any $(G,\sigma,\tau)$ where the neighbors of $v$ form a clique. Indeed, this example is the reason that all existing results on sampling colorings that proceeded \cite{vigoda2000improved} needed to at least assume triangle-freeness of $G$, otherwise the local uniformity properties they leverage simply do not hold. Instead of avoiding state \Sing{c}, we want ``typical'' neighboring coloring pairs to avoid \Bad{c} for many $c$. Specifically, we want $N_{bad}(\sigma,\tau)$ to be at most a constant times $N_{good}(\sigma,\tau)$. 
\sitan{\begin{remark}
At this point the reader may be wondering: why can we get away with only analyzing what happens to the extremal configurations of size 2 and not those of size 1? The reason is that neighboring coloring pairs where $v$ is surrounded by many configurations of size 1 are precisely the kinds of examples on which Vigoda's analysis does particularly well: under Jerrum's maximal coupling, it is impossible to go beneath $k > 2\Delta$ for any one-step coupling of the Glauber dynamics with respect to the Hamming metric, but under Vigoda's greedy coupling of the flip dynamics, one can \emph{perfectly} couple the flips of Kempe components in $\mathcal{D}_c$ for any $c$ with $\delta_c = 1$ if $N_{max}$ is big enough. The reason Vigoda's analysis doesn't get all the way down to $k > (1+\epsilon)\Delta$ is simply that if $N_{max}$ is too big and the flip parameters too tuned to configurations of size 1, one cannot closely couple flips corresponding to configurations of size at least 2. For this reason, what we really care about is actually the fraction of configurations of size at least 2 around $v$ that are extremal.
\label{remark:whynotavoidsing}
\end{remark}}

Consider the following thought experiment. Let $\mathcal{C}$ consist of all neighboring coloring pairs such that for every $(G,\sigma,\tau)\in\mathcal{C}$, \begin{equation}N_{bad}(\sigma,\tau)\le\gamma\cdot N_{good}(\sigma,\tau)\label{eq:comparable}\end{equation} for some absolute constant $\gamma > 0$. Suppose $k = (11/6-\epsilon)\Delta$ for some small absolute constant $\epsilon > 0$, and our goal is just to pick flip parameters so every pair in $\mathcal{C}$ contracts. Observation~\ref{obs:slack} and complementary slackness intuitively suggest that this should be possible for small enough $\epsilon$ depending only on $\gamma$. To get an effective estimate for $\epsilon$, we encode \eqref{eq:comparable} into Linear Program~\ref{def:lp}:

\begin{lp}
Introduce into Linear Program~\ref{def:lp} the additional variables $\lambda_{sing},\lambda_{bad},\lambda_{good}$. In constraints~\eqref{eq:approxconstraint} and \eqref{eq:Hforsigmav}, replace $\lambda$ with $\lambda_{good}$. In the constraints~\eqref{eq:mainconstraint} corresponding to configuration $(A,B;\vec{a},\vec{b})$, replace $\lambda$ with $\lambda_{sing}$ if $m = 1$, $\lambda_{bad}$ if $(A,B;\vec{a},\vec{b})$ is either $(7,3;(3,3),(1,1))$ or $(3,7;(1,1),(3,3))$, or $\lambda_{good}$ otherwise. Lastly, introduce the constraints $$\lambda \ge \lambda_{sing}, \ \ \ \lambda\ge\lambda_{good}, \ \ \ \lambda\ge\frac{\gamma}{\gamma+1}\cdot\lambda_{bad} + \frac{1}{\gamma+1}\cdot\lambda_{good}.$$ Call this the \emph{$\gamma$-mixed coupling LP} and denote its objective value by $\lambda^*_{\gamma}$.
\label{def:mixedlp}
\end{lp}

The following is straightforward to prove; see Appendix~\ref{app:thoughtexperiment} for a formal proof.

\begin{lem}
	If $k > \lambda^*_{\gamma}\Delta$, then under the greedy coupling $(\sigma,\tau)\mapsto(\sigma',\tau')$, $\E[d_H(\sigma',\tau') - 1] < 0$ for any $(G,\sigma,\tau)\in\mathcal{C}$.\label{lem:thought}
\end{lem}

Next, we go from the intuition of this thought experiment to a rigorous notion of ``typical'' neighboring coloring pairs avoiding the state \Bad{c}. Having already reduced finding a coupling for all of $\mathcal{C}$ to analyzing the $\gamma$-mixed coupling LP, in the sequel we will reduce finding a coupling for \emph{all} neighboring coloring pairs to analyzing the $\gamma$-mixed coupling LP.

\subsection{Variable-Length Coupling}
\label{sec:coupling}

The key idea is that regardless of what neighboring coloring pair $(G,\sigma,\tau)$ one starts with, the probability that $\sigma',\tau'$ derived from one step of greedy coupling has changed in distance is $\Theta(1/n)$ (see Lemma~\ref{lem:probend} below). So in expectation, one can run $\Theta(n)$ steps of greedy coupling before the two colorings either coalesce to the same coloring or have Hamming distance greater than 1, but by that time the set of colors around $v$ will have changed substantially. This is the main insight of \cite{dyer2001extension,hayes2007variable}, who leverage it to analyze the Glauber dynamics and slightly improve upon Jerrum's $k\ge 2d$ bound under extra girth and degree assumptions. We leverage this insight to analyze the flip dynamics under no extra assumptions.

Our variable-length coupling simply runs greedy coupling until the distance between the colorings changes: start with neighboring colorings $\sigma^{(0)},\tau^{(0)}$, initialize $t = 1$, and repeat the following.

\begin{enumerate}
 	\item Run the greedy one-step coupling of Section~\ref{subsec:introonestep} to flip components $S_t$ in $\sigma^{(t-1)}$ and $S'_t$ in $\tau^{(t-1)}$, producing $\sigma^{(t)},\tau^{(t)}$ (note that $S_t$ or $S'_t$ might be empty, e.g with probability $1 - p_{\alpha}$, a component of size $\alpha$ that is chosen to be flipped is not actually flipped).
 	\item If $d_H(\sigma^{(t)},\tau^{(t)})\neq d_H(\sigma^{(t-1)},\tau^{(t-1)})$, terminate and define $\Tstop = t$. Else, increment $t$.
\end{enumerate}

We call any subsequence of pairs of flips $(S_i,S'_i),\dots,(S_j,S'_j)$ a \emph{coupling schedule} starting from the neighboring coloring pair $(G,\sigma^{(i-1)},\tau^{(i-1)})$.

It is easy to see that this satisfies the conditions of being a variable-length coupling as in Definition~\ref{def:varlength}. Indeed it is the same coupling as in \cite{hayes2007variable}, except for the flip dynamics instead of the Glauber dynamics. Note that we can characterize which pairs of flips $(S_t,S'_t)$ terminate the coupling: at least one of them must belong to the symmetric difference $\mathcal{D}$ defined in Section~\ref{subsec:introonestep}.

\begin{defn}
	Given a neighboring coloring pair $(G,\sigma,\tau)$, a pair of components $S$ in $\sigma$ and $S'$ in $\tau$ is \emph{terminating} if $S = S_{\sigma}(v,c)$ or $S' = S_{\tau}(v,c)$, or there exists $u\in N(v)$ for which $S = S_{\sigma}(u,\tau(v))$ or $S' = S_{\tau}(u,\sigma(v))$.
\end{defn}

Note that for any $t$, $S_t$ and/or $S'_t$ may be the empty set. Moreover, because flips of components outside of $\mathcal{D}$ are matched via the identity coupling, if $(S_t,S'_t)$ is not terminating, then $S_t = S'_t$.

\begin{lem}
Let components $S$ in $\sigma$ and $S'$ in $\tau$ be chosen according to the greedy coupling. Then $$\frac{k - \Delta - 2}{nk}\le \P[(S,S') \ \text{terminating}] \le \frac{k + 2p_2\Delta}{nk},$$ where $p_2$ is the flip parameter for components of size 2.\label{lem:probend}
\end{lem}

\begin{proof}
	For the lower bound, note that the pair $(S_{\sigma}(v,c),S_{\tau}(v,c))$ is terminating for any $c$. In particular, for $c\neq\{\sigma(v),\tau(v)\}$ such that $\delta_c = 0$, $S_{\sigma}(v,c) = S_{\tau}(v,c) = \{v\}$ \---- note that while the vertex sets for these components are all $\{v\}$, the flips are all distinct as they vary with $c$. Each such pair of flips has probability mass $(1/nk)\cdot p_1 = (1/nk)$, and there are at least $k - \Delta - 2$ such choices of $c$, giving the lower bound. We defer the proof of the upper bound to Appendix~\ref{app:probend}.
\end{proof}

\begin{cor}
	$\max_{\sigma^{(0)},\tau^{(0)}}\E[\Tstop]\le\frac{nk}{k - \Delta - 2}$.\label{eq:maxET}
\end{cor}

We now give a reduction from analyzing the expected change in distance under our variable-length coupling to proving that the relation \eqref{eq:comparable} from our thought experiment holds in expectation by the end of the coupling.

% reduction/linearity of expectation lemma
\begin{lem}
	Suppose there exists a constant $\gamma > 0$ for which \begin{equation}\E[N_{bad}(\sigma^{(\Tstop - 1)},\tau^{(\Tstop - 1)})]\le\gamma\cdot\E[N_{good}(\sigma^{(\Tstop - 1)},\tau^{(\Tstop - 1)})]\label{eq:desired}\end{equation} for any initial neighboring coloring pair $(G,\sigma^{(0)},\tau^{(0)})$.

	Let $\lambda^*_{C\gamma}$ be the objective value of the $C\gamma$-mixed coupling LP, where $C := \frac{k + 2p_2\Delta}{k - \Delta - 2}$. Then $$\E[d_H(\sigma^{(\Tstop)},\tau^{(\Tstop)})-1]\le \frac{-k+\lambda^*_{C\gamma}\Delta}{k - \Delta - 2}.$$\label{lem:reduction}
\end{lem}

This mainly just follows by linearity of expectation and the calculation done in the proof of Lemma~\ref{lem:thought}. The only complication is that the probability that the coupling terminates at any given point is not fixed, but this is fine because it is still the same \emph{up to constant factors}, which only leads to the loss of a factor of $C$ as defined in the lemma. We defer the full proof to Appendix~\ref{app:reduction}.
\sitan{\begin{remark}
\eqref{eq:desired} is one way to say that a typical coloring has few extremal configurations around $v$. This is slightly different from the analogous notion in the second proof of Theorem~\ref{thm:main} in Section~\ref{sec:DPP}. There, the goal is to show that the number of extremal configurations around $v$ of size 1 or 2 goes down in one step of Vigoda's greedy coupling. In contrast, we only choose to upper bound the fraction of configurations of size at least 2 at the end of our variable-length coupling which are extremal. That we don't attempt to analyze how extremal components of size 1 break apart is purely out of technical convenience, and as we discuss in Remark~\ref{remark:whynotavoidsing}, still enough to break the $\frac{11}{6}$ barrier.
\label{remark:subtle}
\end{remark}}

\subsection{Few Extremal Configurations When Coupling Terminates}
\label{subsec:burnin}

We are left with proving the main technical lemma that \eqref{eq:desired} holds. Throughout this subsection, assume that $k\ge 1.833\Delta$.

\begin{lem}
Suppose flip parameters $\{p_{\alpha}\}_{\alpha\in\N_0}$ satisfy $p_0=0\le p_{\alpha}\le p_{\alpha-1}\le p_1 = 1$ for all $\alpha\ge 2$, constraint \eqref{eq:pap}, and additionally $\alpha p_{\alpha - 2}\le 3$ for all $\alpha\ge 3$. Then for $\gamma:= \frac{(6k - \Delta - 2)(k + 2p_2\Delta)}{4(k - \Delta - 2)(k - \Delta - 1)},$ we have that \eqref{eq:desired} holds for any initial neighboring coloring pair $(G,\sigma^{(0)},\tau^{(0)})$.\label{lem:main}
\end{lem}

\begin{remark}
The additional constraint that $\alpha p_{\alpha - 2}\le 3$ already holds for the solutions to the $\gamma$-mixed coupling LP, so we assume it just to obtain better constant factors in our analysis.
\end{remark}

We first make a simple reduction. Fix any initial neighboring coloring pair $(G,\sigma^{(0)},\tau^{(0)})$, and for every color $c$ denote by $p_{bad}(c)$ and $p_{good}(c)$ the probability that $\sigma^{(\Tstop)},\tau^{(\Tstop)}$ is in state \Bad{c} and \Good{c}, respectively.

By linearity of expectation we have that $$\E[N_{bad}(\sigma^{(\Tstop)},\tau^{(\Tstop)})] = \sum_{c}p_{bad}(c), \ \ \ \ \E[N_{good}(\sigma^{(\Tstop)},\tau^{(\Tstop)})] = \sum_{c}p_{good}(c).$$ Therefore to show Lemma~\ref{lem:main}, it is enough to show the following.

\begin{lem}
 	$p_{bad}(c)\le\gamma\cdot p_{good}(c)$ for every color $c$.\label{lem:realmain}
\end{lem} 

This is certainly true for $c = \sigma^{(0)}(v),\tau^{(0)}(v)$, in which case $p_{bad}(c) = 0$ and $p_{good}(c) = 1$ by Observation~\ref{obs:always}. We point out that while the case of $c = \sigma^{(0)}(v),\tau^{(0)}(v)$ is the one for which the fact that our state space includes all colorings, improper and proper, introduces complications in the definition of the greedy coupling (see Remark~\ref{remark:specialcase}), it happens to be the easiest case of Lemma~\ref{lem:realmain}.

So henceforth assume $c\neq\sigma^{(0)}(v),\tau^{(0)}(v)$. We proceed via a fractional matching argument. Take any coupling schedule $\Sigma_{pre} = (S_1,S_1), (S_2,S_2), \dots, (S_{T-1},S_{T-1})$ consisting of pairs of identical flips, and define $\mathcal{W}$ to be the set of all coupling schedules of the form $(S_1,S_1), (S_2,S_2), \dots, (S_{T-1},S_{T-1})$, $(S_T, S'_T)$ for $(S_T,S'_T)$ terminating. In other words, $\mathcal{W}$ consists of all $T$-step coupling schedules whose first $T - 1$ steps are fixed to $\Sigma_{pre}$ and which only changes the distance between the colorings in the last step $(S_T,S'_T)$. We will match to the collection of schedules $\mathcal{W}$ an (infinite) collection of schedules of the following form.

\begin{defn}
	Fix $S_1,\dots,S_{T-1}$. A coupling schedule $\Sigma_*$ starting from the neighboring coloring pair $(G,\sigma^{(0)},\tau^{(0)})$ is \emph{satisfying} if it is of the form \begin{equation}\Sigma_* = (S_1,S_1),\dots,(S_{T-1},S_{T-1}),(S^*_T,S^*_T),\dots, (S^*_{T^*-1}, S^*_{T^*-1}), (S^*_{T^*},S'^*_{T^*})\label{eq:defsatisfying}\end{equation} for $(S^*_{T^*},S'^*_{T^*})$ terminating, and gives rise to a sequence of colorings $$(\sigma^{(0)},\tau^{(0)}), (\sigma^{(1)},\tau^{(1)}),\dots,(\sigma^{(T-1)},\tau^{(T-1)}),(\sigma^{(T)}_*,\tau^{(T)}_*),\dots,(\sigma^{(T^*)}_*,\tau^{(T^*)}_*)$$ for which \begin{enumerate}
 	\item $\sigma^{(T^*-1)}_*,\tau^{(T^*-1)}_*$ are in state \Good{c}
 	\item $\sigma^{(t)}_*,\tau^{(t)}_*$ is not in state \Bad{c} for any $T\le t<T^*$.
	\end{enumerate}\label{def:satisfying}
\end{defn}

In Definition~\ref{def:satisfying}, Property 2) ensures that from any satisfying $\Sigma_*$, we can uniquely decode the collection $\mathcal{W}$ to which it is being fractionally matched: in $\Sigma_*$, take the last pair of colorings in state \Bad{c}, and $\Sigma_{pre}$ is the subsequence of $\Sigma_*$ starting from $(\sigma^{(1)},\tau^{(1)})$ and ending at that pair.

If we can show for any $\Sigma_{pre}$ that $$\sum_{\Sigma_* \ \text{satisfying}}\Pr[(S^*_T,S^*_T),\dots, (S^*_{T^*-1}, S^*_{T^*-1}), (S^*_{T^*},S'^*_{T^*})\vert \Sigma_{pre}]\ge \frac{1}{\gamma}\cdot\frac{k + 2p_2\Delta}{nk},$$ then this will imply that $p_{bad}(c)\le\gamma\cdot p_{good}(c)$ because the upper bound of Lemma~\ref{lem:probend} tells us that $\Pr[(S_T,S'_T) \ \text{terminating}\vert\Sigma_{pre}]\le\frac{k+2p_2\Delta}{nk}$.

To exhibit such a collection of satisfying coupling schedules $\Sigma_*$, we first define a coarsening of the state space as follows. Starting from the neighboring coloring pair $(G,\sigma^{(T-1)},\tau^{(T-1)})$ which is in state \Bad{c}, take any subsequent coupling schedule \begin{equation}(S^*_T,S^*_T),\dots,(S^*_{T^*},S'^*_{T^*})\label{eq:extraschedule}\end{equation} with $(S^*_{T^*},S'^*_{T^*})$ terminating which gives rise to a sequence of pairs of colorings \begin{equation}(\sigma^{(T)}_*,\tau^{(T)}_*),\dots,(\sigma^{(T^*)}_*,\tau^{(T^*)}_*),\label{eq:extra}\end{equation} where $(G,\sigma^{(t)}_*,\tau^{(t)}_*)$ is a neighboring coloring pair for all $t$ except $t = T^*$. Define the following auxiliary states. To avoid confusion with the states in Definition~\ref{def:badgood}, we will refer to the auxiliary states defined below as \emph{stages}.

\begin{defn}
	Let $c$ be any color, not necessarily one appearing in the neighborhood of $v$. We say that $\sigma^{(t)}_*,\tau^{(t)}_*$ is in \emph{stage \Goodend{c}} if $\sigma^{(t-1)}_*,\tau^{(t-1)}_*$ is in state \Good{c} and the pair of flips $(S,S')$ giving rise to $\sigma^{(t)}_*,\sigma^{(t)}_*$ from $\sigma^{(t-1)}_*,\tau^{(t-1)}_*$ is terminating.

	We say $\sigma^{(t)}_*,\tau^{(t)}_*$ is in \emph{stage \Badend{c}} if, intuitively, we choose to quit looking for satisfying coupling schedules among those of which $(\sigma^{(0)}_*,\tau^{(0)}_*),\dots,(\sigma^{(t)}_*,\tau^{(t)}_*)$ is a prefix. Formally, $\sigma^{(t)}_*,\tau^{(t)}_*$ is in stage \Badend{c} if least one of the following conditions holds (note that these conditions aren't necessarily mutually exclusive):
	\begin{enumerate}
		\item[(i)] $t = T $ and the pair of flips $(S,S')$ giving rise to $\sigma^{(T)}_*,\sigma^{(T)}_*$ from the initial pair $\sigma^{(T-1)},\tau^{(T-1)}$ is terminating (i.e. if $(S_1,S_1),\dots,(S_{T-1},S_{T-1}),(S,S')\in\mathcal{W}$).
		\item[(ii)] $t = T$ and $\sigma^{(t)}_*,\tau^{(t)}_*$ is not in state \Good{c}.
		\item[(iii)] $\sigma^{(t-1)}_*,\tau^{(t-1)}_*$ is in state \Good{c} but $\sigma^{(t)}_*,\tau^{(t)}_*$ is not in state \Good{c} or stage \Goodend{c} (this includes the case that $c$ does not appear in the neighborhood of $v$ in $\sigma^{(t)}_*,\tau^{(t)}_*$).
		\item[(iv)] $t > T$ and $\sigma^{(t-1)}_*,\tau^{(t-1)}_*$ is in stage \Badend{c}.
	\end{enumerate}

	If $\sigma^{(t)}_*,\tau^{(t)}_*$ is not in stage \Badend{c} or \Goodend{c} and is in state \Bad{c} (resp. \Good{c}), then we say it is also in \emph{stage \Bad{c}} (resp. \emph{stage \Good{c}}).\label{def:stages}
\end{defn}

Note that if a sequence of the form \eqref{eq:extra} contains a pair of colorings in stage \Goodend{c}, that pair must be $\sigma^{(T^*)}_*,\tau^{(T^*)}_*$. Furthermore, given any sequence \eqref{eq:extra} for which $\sigma^{(T^*)}_*,\tau^{(T^*)}_*$ is in stage \Goodend{c} with associated coupling schedule \eqref{eq:extraschedule}, note that the corresponding coupling schedule $\Sigma_*$ defined in \eqref{eq:defsatisfying} is satisfying, by definition of stage \Badend{c}.

So it is enough to show that if we start from a neighboring coloring pair $(G,\sigma^{(T-1)},\tau^{(T-1)})$ which is in state \Bad{c} and evolve a sequence of pairs of colorings \eqref{eq:extra} according to the greedy coupling at each step, then \begin{equation}\Pr[\sigma^{(T^*)}_*,\tau^{(T^*)}_* \ \text{are in stage \Goodend{c}}\vert \sigma^{(T-1)},\tau^{(T-1)}]\ge\frac{1}{\gamma}\cdot\frac{k + 2p_2\Delta}{nk}.\label{eq:mainmarkov}\end{equation} It remains to bound the probabilities of the transitions between the different stages of Definition~\ref{def:stages} under the flip dynamics and the greedy coupling (see Figure~\ref{fig:transitions} for a depiction of the transitions that can occur). A key point is that these bounds will be independent of the specific colorings or structure of $G$.

%!TEX root = fullpaper.tex

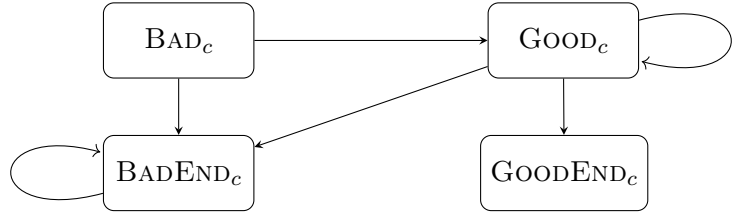
\begin{figure}[ht]
\centering
\begin{tikzpicture}[align=center]
\usetikzlibrary{positioning,shapes.geometric,arrows}
\tikzstyle{box} = [rectangle, rounded corners, minimum width=2cm, minimum height=1cm,text centered, draw=black, fill=white]
\tikzstyle{arrow} = [thin,->,>=stealth]
\tikzstyle{line} = [draw, -latex']

	\node (bad) [box] {\Bad{c}};
	\node (badend) [box, below=0.75cm of bad] {\Badend{c}};
	\node (goodend) [box,right=3cm of badend] {\Goodend{c}};
	\node (good) [box,right=3.1cm of bad] {\Good{c}};
	\draw [arrow] (bad) -- (badend);
	\draw [arrow] (bad) -- (good);
	\draw [arrow] (good) -- (badend);
	\draw [arrow] (good) -- (goodend);
	\path [line] (good) edge[loop right]();
	\path [line] (badend) edge[loop left]();
\end{tikzpicture}
\caption{Possible transitions among stages of Definition~\ref{def:stages}}
\label{fig:transitions}
\end{figure}

For $\sigma^{(T^*)}_*,\tau^{(T^*)}_*$ to be in stage \Goodend{c}, $\sigma^{(t)}_*,\tau^{(t)}_*$ cannot be in stage \Bad{c} for any $t\ge T$. In other words, because $\sigma^{(T-1)}_*,\tau^{(T-1)}_*$ is in state \Bad{c}, the pair of colorings must escape from \Bad{c} in the very first step of \eqref{eq:extraschedule} and never return. The following says that this probability of escape is comparable to the total probability mass of $\mathcal{W}$. We defer its proof to Appendix~\ref{app:burnin}.

% Given states $S_1$ and $S_2$ among $\{$\Good{c},\Bad{c},\Goodend{c},\Badend{c}$\}$, denote by $P(S_1\to S_2)$ the quantity $$\Pr[\sigma',\tau' \ \text{are in state} \ S_2\vert \sigma,\tau \ \ \text{are in state} \ \ S_1],$$ where $(G,\sigma,\tau)$ is a neighboring coloring and $\sigma',\tau'$ is a pair of colorings obtained from one step of Vigoda's coupling.

\begin{lem}
Let $\sigma,\tau$ be any neighboring coloring pair in state \Bad{c}, and let $\sigma',\tau'$ be derived from one step of greedy coupling. Then $P[\sigma',\tau' \ \text{in state \Good{c}}]\ge\frac{4(k - \Delta - 1)}{nk}.$\label{lem:badtogood}
\end{lem}

Once a pair of colorings has escaped from state \Bad{c} into state \Good{c}, at every step it can only stay at \Good{c}, end at stage \Goodend{c}, or get absorbed into stage \Badend{c}. The next two lemmas say that the last two events have probability $\Omega(1/n)$ and $O(1/n)$ respectively.

\begin{lem}
Let $\sigma,\tau$ be any neighboring coloring pair in stage \Good{c}, and let $\sigma',\tau'$ be derived from one step of greedy coupling. Then $P[\sigma',\tau' \ \text{in stage \Goodend{c}}]\ge\frac{k - \Delta - 2}{nk}.$\label{lem:goodtogoodend}
\end{lem}

\begin{lem}
Let $\sigma,\tau$ be any neighboring coloring pair in stage \Good{c}, and let $\sigma',\tau'$ be derived from one step of greedy coupling. Then $P[\sigma',\tau' \ \text{in stage \Badend{c}}]\le\frac{5}{n}.$\label{lem:goodtobadend}
\end{lem}

Lemma~\ref{lem:goodtogoodend} follows immediately from Lemma~\ref{lem:probend}. Lemma~\ref{lem:goodtobadend} is the most technically involved step, and we defer its proof to Appendix~\ref{app:burnin}.

We can now complete the proofs of Lemma~\ref{lem:main} and Theorem~\ref{thm:main}.

\begin{proof}[Proof of Lemma~\ref{lem:main}]
Starting from $\sigma^{(T-1)},\tau^{(T-1)}$ in stage \Bad{c}, by Lemma~\ref{lem:badtogood}, the probability of transitioning to stage \Good{c} in the very next step is at least $\frac{4(k - \Delta - 1)}{nk}$. As shown in Figure~\ref{fig:transitions}, once we leave stage \Bad{c} we never return. From stage \Good{c}, it is at most $\frac{5}{n}\cdot\frac{nk}{k - \Delta - 2} = \frac{5k}{k - \Delta - 2}$ times as likely to eventually end up at stage \Badend{c} as it is to end up at stage \Goodend{c}, by Lemmas~\ref{lem:goodtogoodend} and Lemma~\ref{lem:goodtobadend}. So the probability of ending in stage \Goodend{c} is at least $\frac{k - \Delta - 2}{5k + (k - \Delta - 2)}\cdot\frac{4(k - \Delta - 1)}{nk}$, and we conclude that \eqref{eq:mainmarkov} and consequently Lemma~\ref{lem:main} hold for $\gamma$ as defined.
\end{proof}

\begin{proof}[Proof of Theorem~\ref{thm:main}]
Note that for $k > 1.833\Delta$ and $p_2 < 0.3$, $\gamma = \frac{(6k - \Delta - 2)(k + 2p_2\Delta)}{4(k - \Delta - 2)(k - \Delta - 1)}< 7.683410$, while $C = \frac{k + 2p_2\Delta}{k - \Delta - 2} < 2.920764$, so $C\gamma < 25.597784$ as defined in Lemma~\ref{lem:reduction}. Thus, substituting $25.597784$ into the $\gamma$ parameter for Linear Program~\ref{def:mixedlp} and solving numerically\footnote{Code for solving Linear Program~\ref{def:mixedlp} can be found at \url{https://github.com/sitanc/mixedlp}.}, we find that for \begin{equation}\hat{p}_1 = 1, \hat{p}_2\approx 0.296706, \hat{p}_3\approx 0.166762, \hat{p}_4\approx 0.101790, \hat{p}_5\approx 0.058475, \hat{p}_6\approx 0.025989, \hat{p}_{\alpha} = 0 \ \forall \alpha\ge 7,\label{eq:CMflipprobs}\end{equation} Linear Program~\ref{def:mixedlp} attains value $\lambda^* < 1.833239$. So provided $k\ge1.833239\Delta$, Lemma~\ref{lem:reduction} implies that the variable-length coupling is $(1-\alpha)$-contractive for absolute constant $\alpha:= \frac{k-\lambda^*\Delta}{k - \Delta - 2}$.

For $k\ge 1.833239\Delta$, Corollary~\ref{eq:maxET} implies that $\beta$ in the definition of Theorem~\ref{thm:hayesvigoda} is at most $\frac{nk}{k - \Delta - 2}\le 2.21n$, so applying Theorem~\ref{thm:hayesvigoda} with $\beta = 2.21n$, and $W = 2N_{max} + 1 = 13$ gives that $\tau_{mix}(\epsilon) = O(n\log(n/\epsilon))$. In particular, $\tau_{mix} = O(n\log n)$.
\end{proof}

% \begin{proof}[Proof of Theorem~\ref{thm:phasetransition}]
% 	In \cite{vigoda2000improved}, Vigoda proves that if the flip dynamics mix in time $O(n\log n)$ on the induced subgraph of $\Z^d$ on vertex set $\{-L,\cdots,L\}^d$ for all fixed boundary configurations, 
% \end{proof}

% DPP approach
%!TEX root = fullpaper.tex

\newcommand{\const}{84000}

\newcommand{\cA}{\mathcal{A}} 
\newcommand{\cB}{\mathcal{B}}
\newcommand{\cD}{\mathcal{D}}
\newcommand{\cE}{\mathcal{E}}
\newcommand{\cK}{\mathcal{K}}
\newcommand{\cS}{\mathcal{S}}
\newcommand{\cR}{\mathcal{R}}
\newcommand{\cT}{\mathcal{T}}

\newcommand{\bP}{\Pr}%{\mathbb{P}}
\newcommand{\bE}{\mathbb{E}}

%
%\newcommand{\overbar}[1]{\mkern 1.5mu\overline{\mkern-1.5mu#1\mkern-1.5mu}\mkern 1.5mu}
%\newcommand{\comm}[1]{\colorbox{yellow}{\textbf{#1}}}
%
%\title{Rapid mixing of Glauber dynamics for colorings below Vigoda's $11/6$ threshold}
%%: beyond Vigoda's bound}
%% Improving VIgoda's bound for sampling colorings
%\author{\
%Michelle Delcourt
%\thanks{School of Mathematics,
%University of Birmingham, Birmingham, UK  {\tt m.delcourt@bham.ac.uk,}. Research supported by supported by EPSRC grant EP/P009913/1.}
%\and
%Guillem Perarnau
%\thanks{School of Mathematics,
%University of Birmingham, Birmingham, UK.  {\tt g.perarnau@bham.ac.uk}.}
%\and
%Luke Postle
%\thanks{Combinatorics and Optimization Department,
%University of Waterloo, Waterloo, Ontario N2L 3G1, Canada {\tt lpostle@uwaterloo.ca}. Partially supported by NSERC
%under Discovery Grant No. 2014-06162.}}
%\date{\today}

\section{Proof of Theorem~\ref{thm:main} Using an Alternative Metric}\label{sec:DPP}

%To our knowledge, alternative metrics have found only one application in previous works on sampling colorings, namely in the analysis of the ``scan'' chain for sampling colorings of bipartite graphs in \cite{bordewich2006stopping}. On the other hand, path coupling using alternative metrics has found success in other problems in the approximate sampling literature \cite{luby1999fast,bubley1998faster,bordewich2006stopping}. In this framework, the hope is to get away with a one-step analysis by choosing a metric that is, roughly speaking, less sensitive than the Hamming metric to worst-case neighborhoods. Bordewich et al. \cite{bordewich2006stopping} gave evidence that the multi-step stopping time-based approach can be captured by one-step coupling with an appropriate metric, though the metric they use to establish this connection is itself based on stopping times. In an orthogonal direction, another take on the question of designing metrics for one-step coupling is given in~\cite{hayes2015randomly} by using the spectral radius of the adjacency matrix and in~\cite{efthymiou2016convergence} by analyzing the Jacobian of the belief propagation operator. 

This section contains the proof of Theorem~\ref{thm:main} using an alternative metric as presented in~\cite{delcourt2018rapid}. Lemma~\ref{lem:tight} shows that a one-step coupling analysis using the Hamming metric cannot yield any improvement over $11/6$. The first step in our proof is to find a set of flip parameters that only has two extremal configurations, up to symmetry, namely the ones used to define Linear Program~\ref{def:mostbasictight}. We then introduce a new metric that depends on the number of colors in non-extremal configurations. Analyzing the one-step coupling defined in Appendix~\ref{app:vigodareview} (greedy coupling) with this metric, we obtain a constant improvement over the $11/6$ bound.

\subsection{Choice of Flip Parameters and Expected Change in Hamming metric}\label{sec:nablaH}

%Vigoda's choice of parameters are optimal and imply
%\begin{align}\label{eq:optimal}
%H(A,B;\vec{a},\vec{b})\leq -1+(11/6)\cdot m,
%\end{align}
%for every realizable configuration $(A,B;\vec{a},\vec{b})$ of size $m$. As noted in Observation~\ref{obs:slack}, up to symmetry, there are six extremal realizable configurations under Vigoda's assignment. In this section we find a new assignation $\mathbf{\hat{p}}$ of the flip parameters that form an optimal solution for Linear Programming~\ref{def:lp} while minimizing the number of $\mathbf{\hat{p}}$-extremal configurations. We will also quantify the improvement over~\eqref{eq:optimal} that one obtains for configurations that are not $\mathbf{\hat{p}}$-extremal.

%Recall that $p_0=0\le p_{\alpha}\le p_{\alpha-1}\le p_1 = 1$.  
By Corollary~\ref{cor:mostbasictightcor}, the objective value of Linear Program~\ref{def:mostbasictight} is $11/6$. It is straightforward to check that $\lambda=11/6$ only if the solution satisfies $p_3=1/6$ and $p_{\alpha}=0$ for all $\alpha\geq 7$. Since for any such assignment, the constraints corresponding to $(3,2;(1),(1))$ and $(7,3;(3,3),(1,1))$ are tight, we introduce the following variant of Linear Program~\ref{def:prelp}:
\begin{lp}\label{def:prelp_DPP}
For variables $\{p_\alpha\}_{\alpha\in \mathbb{N}_0}$ and $\lambda$, minimize $\lambda$ subject to the following constraints: $p_0=0\le p_{\alpha}\le p_{\alpha-1}\le p_1 = 1$ for all $\alpha\ge 2$, $p_3=1/6$, $p_{\alpha}=0$ for $\alpha\geq 7$ and for all realizable configurations $(A,B;\mathbf{a},\mathbf{b})$ of size $m$ different from $(3,2;(2),(1))$, $(2,3;(1),(2))$, $(7,3;(3,3),(1,1))$ and $(3,7;(1,1),(3,3))$, define a constraint
\begin{align*}
H(A,B;\mathbf{a},\mathbf{b})\leq -1+\lambda m .
\end{align*}
\end{lp}

Consider the following reduced linear program with a finite set of variables and constraints. %that captures the optimal solutions of Linear Program~\ref{def:prelp_DPP}.

\begin{lp}\label{def:lp_DPP}
For variables $\{p_1,\dots, p_6\}$ and $\lambda$, minimize $\lambda$ subject to the following constraints: $0\le p_{\alpha}\le p_{\alpha-1}\le p_1 = 1$ for all $\alpha\ge 2$, $p_3=1/6$,  constraints
\begin{align}\label{eq:constr_1}
i(p_i - p_{i+1}) + j(p_j - p_{j+1})- \min\{p_i - p_{i+1},p_j - p_{j+1} \} &\leq -1+ \lambda ,
\end{align}
 for $i\in \{1,\dots,6\}$, $ j\in \{2,\dots, 6\}$ with $ (i,j)\neq (1,2)$,  and  
$$
2p_1+3p_2  -\min\{p_2- p_{5},p_3-p_1\} \leq -1+ 2\lambda. 
$$
%\alpha p_{\alpha}\leq 1 & \alpha\in \{1,\dots ,6 \}\\
%(\alpha-1) p_{\alpha}\leq 2p_3 & \alpha\in \{2,\dots ,6 \}\\
%(\alpha-2)p_\alpha <  (3\lambda-5)/2 & \alpha\in \{3,\dots ,6 \}\\ 
%\end{align*}
\end{lp}
One can solve the program using a computer.
\begin{obs}\label{obs:DPP_assig}
The Linear Program~\ref{def:lp_DPP} has objective value $\hat{\lambda}=\frac{161}{88}=1.8295\dots$ and an optimal solution is given by
$$
\hat{p}_1 = 1,\, \hat{p}_2 = \frac{185}{616},\, \hat{p}_3=\frac{1}{6},\, \hat{p}_4=\frac{47}{462},\, \hat{p}_5=\frac{9}{154},\, \hat{p}_6 =\frac{2}{77}.
$$
\end{obs}
\medskip

\begin{lem}\label{lem:LP_equiv}
The assignment $\mathbf{\hat{p}}= \{\hat{p}_{\alpha}\}_{\alpha\in \N_0}$ where $\hat{p}_{\alpha}$ is given by Observation~\ref{obs:DPP_assig} for $\alpha\in [6]$  and  $\hat{p}_{\alpha}=0$ otherwise, forms a feasible solution of Linear Program~\ref{def:prelp_DPP} with objective value $\hat{\lambda}$.
\end{lem}
We defer the proof of the lemma to Appendix~\ref{app:lemLP}.
%As the constraints in Linear Program~\ref{def:prelp} that are not contained  in Linear Program~\ref{def:prelp_DPP} are directly implied by the conditions $p_1=1$, $p_3=1/6$ and $p_7=0$, any feasible solution of the Linear Program~\ref{def:prelp_DPP} is also feasible in Linear Program~\ref{def:prelp}. Moreover, since $\mathbf{\hat{p}}$ has objective value in Linear Program~\ref{def:prelp_DPP} strictly smaller than $11/6$,

\begin{obs}
Consider the solution $\mathbf{\hat{p}}$ of Linear Program~\ref{def:prelp_DPP} given in Observation~\ref{obs:DPP_assig}. The constraints in Linear Program~\ref{def:prelp} that are not contained in Linear Program~\ref{def:prelp_DPP} are implied by the conditions $p_1=1$, $p_3=1/6$ and $p_7=0$. Thus, $\mathbf{\hat{p}}$ is a feasible solution of Linear Program~\ref{def:prelp} with objective value $11/6$. Since the objective value of $\mathbf{\hat{p}}$ in Linear Program~\ref{def:prelp_DPP} is strictly smaller than $11/6$, up to symmetry, there are only two $\mathbf{\hat{p}}$-extremal configurations: $(3,2;(1),(1))$ and $(7,3;(3,3),(1,1))$.
\end{obs} 
%\begin{itemize}
%\item[i)] $m=1$, $A-1=a_1=3$ and $B-1=b_1=1$;
%\item[ii)] $m=2$, $A=7$, $a_1=a_2=3$, $B=3$ and $b_1=b_2=1$.
%\end{itemize}

We conclude that
\begin{align}\label{eq:bound}
H(A,B;\mathbf{a},\mathbf{b}) &\leq 
\begin{cases}
\frac{11}{6} & \text{ for every  $\mathbf{\hat{p}}$-extremal configuration } (A,B;\mathbf{a},\mathbf{b}),\\ 
\hat{\lambda}=\frac{161}{88}  & \text{ otherwise}.
\end{cases}
\end{align}

%We conclude that for every $\mathbf{\hat{p}}$-extremal configuration $(A,B;\mathbf{a},\mathbf{b})$, we have
%\begin{align}\label{eq:bound_1}
%H(A,B;\mathbf{a},\mathbf{b}) = \frac{11}{6}\;,
%\end{align}
%and otherwise, we obtain an improvement over~\eqref{eq:optimal}
%\begin{align}\label{eq:bound_2}
%H(A,B;\mathbf{a},\mathbf{b}) \leq \hat{\lambda}=\frac{161}{88}\;.
%\end{align}

\subsection{Definition of the Alternative Metric}\label{sec:metric}

In this section we introduce the alternative metric we will use for the analysis of the one-step coupling, defined using a pre-metric $(\Gamma,\omega)$. Let $\Gamma$ be the graph with vertex set $\Omega$ where two colorings are adjacent if and only if they differ at exactly one vertex. We proceed to define $\omega$.

%In order to define $\omega$ and motivated by we introduce the notion of an extremal configuration. The configurations $(A,B;\mathbf{a},\mathbf{b})= (3,2,(2), (1)), (2,3,(1),(2))$ are called \emph{extremal $1$-configurations}, and the configurations $(A,B;\mathbf{a},\mathbf{b})= (7,3,(3,3), (1,1)), (3,7,(1,1),(3,3))$ are called \emph{extremal $2$-configurations}. All the other realizable configuration are called \emph{non-extremal}.
%We will see in Section~\ref{sec:change} why these configurations are of particular interest.

Fix the assignation of flip parameters $\mathbf{\hat{p}}$ given in Observation~\ref{obs:DPP_assig}. Our definition of the pre-metric is driven by the fact that extremal configurations are an obstacle to show contraction of the metric when $k<11\Delta/6$. Define the following sets of colors, 
\begin{align*}
C^1_{\sigma, \tau}(v) &:= \left\{c :\, (A_c,B_c,\vec{a}^c,\vec{b}^c) \text{ is an extremal configuration for }(\sigma,\tau)\text{ of size $1$}   \right\},\\
C^2_{\sigma, \tau}(v) &:= \left\{c :\,  (A_c,B_c,\vec{a}^c,\vec{b}^c) \text{ is an extremal configuration for }(\sigma,\tau)\text{ of size $2$}  \right\}.
\end{align*}
Let $C_{\sigma,\tau}(v)= C^1_{\sigma, \tau}(v) \cup C^2_{\sigma, \tau}(v)$ be the set of colors $c$ such that $(\sigma,\tau)$ has an extremal configuration for $c$. Note that for each color $c\in C^2_{\sigma, \tau}(v)$, there are two neighbors of $v$ with color $c$. Let $\gamma_{\sigma,\tau}(v)= (|C^1_{\sigma, \tau}(v)| +2 |C^2_{\sigma, \tau}(v)|)/\Delta$, that is, the number of neighbors of $v$ that participate in extremal configurations of $(\sigma,\tau)$ normalised by a factor $\Delta$; thus $\gamma_{\sigma,\tau}(v)\leq 1$.

Let $\eta\in \left(0,\frac{1}{2}\right)$ be a sufficiently small constant to be fixed later. The \emph{weight function} that we will use for our pre-metric is defined as
\begin{align}\label{eq:omega}
\omega(\sigma, \tau) := 1 - \eta(1- \gamma_{\sigma,\tau}(v))\;.
\end{align}

Note that $\omega(\sigma,\tau)\in [1-\eta,1]$. Since $\eta<\frac{1}{2}$ and $\gamma_{\sigma,\tau}(v)\leq 1$, every path containing at least two edges has weight greater than one. So every edge is a minimum weight path, implying that $(\Gamma,\omega)$ is a pre-metric. Let $d$ be the metric on $\Omega$ obtained from $(\Gamma,\omega)$ using minimum weight paths in $\Gamma$. By Remark~\ref{remark:big_space}, for every $(\tilde{\sigma},\tilde{\tau})\in \Omega^2$ there exists a path between $\tilde{\sigma}$ and $\tilde{\tau}$ in $\Gamma$ with $d_H(\tilde{\sigma},\tilde{\tau})$ edges in which every edge has weight at most $1$. It follows that $d(\tilde{\sigma},\tilde{\tau})\leq d_H(\tilde{\sigma},\tilde{\tau})$. 

Define
\begin{align}\label{eq:d_B}
d_B(\tilde{\sigma},\tilde{\tau}) := d_H (\tilde{\sigma},\tilde{\tau})- d(\tilde{\sigma},\tilde{\tau})\;.
\end{align}
In general, $d_B$ is not a metric, here we will only use that it is non-negative. The contribution of $d_B$ will be crucial for the constant improvement over $\frac{11}{6}$ in this approach.

Given the greedy coupling $(\sigma,\tau)\to (\sigma',\tau')$ for neighboring coloring pairs $(\sigma,\tau)$, define
\begin{align}\label{eq:sum}
\nabla(\sigma, \tau) &:= nk\, \bE\left[d(\sigma',\tau')-d(\sigma, \tau)\right]\;. 
\end{align}
%The crux of the argument to prove Theorem~\ref{thm:main} lies in showing that $\nabla(\sigma, \tau) $ is negative. 
%The rescaling factor $nk$ is natural as the probability an alternating component of size $\ell$ is flipped is exactly $p_\ell/nk$.
Define the rescaled contributions to the expected change of $d_H$ and $d_B$ as
\begin{align*}
\nabla_H(\sigma, \tau) &:= nk\,\bE\left[d_H(\sigma',\tau')-1\right] \;, \\
\nabla_B(\sigma, \tau) &:= -nk \,\bE\left[d_B(\sigma', \tau') - d_B(\sigma, \tau)\right]\;,
\end{align*}
and note that $\nabla(\sigma, \tau) = \nabla_H(\sigma, \tau)+\nabla_B(\sigma, \tau)$.

We first bound $\nabla_H$. Recall that $X_c$ is the event that the coupling flips  Kempe components in $\cD_c$ in both chains. Denote by $\overline{X} $ the complement of the event $ \cup_{c:\delta_c>0} X_c$. Define
\begin{align}\label{eq:nabla_dH}
\nabla_H(\sigma, \tau,c) &:= nk\,\bE\left[\mathbbm{1}_{X_c} \cdot (d_H(\sigma',\tau')-1)\right] \;.
\end{align}
Note that $\bE\left[\mathbbm{1}_{\overline{X}} \cdot (d_H(\sigma',\tau')-1)\right]=0$ as we use the identity coupling if $\overline{X}$ holds, so $\sigma'$ and $\tau'$ only differ at $v$.
Using~\eqref{eq:mainvigodaineqcomponent}, Lemma~\ref{lem:vigoda},  
%we obtain
%$$
%\nabla_H(\sigma,\tau)\leq -|\{c: \delta_c=0\}|+ \sum_{c: \delta_c\neq 0} H(A_c,B_c,\mathbf{a}^c,\mathbf{b}^c)\;.
%$$
and~\eqref{eq:bound}, we obtain a bound on the expected change of the Hamming part of the metric in terms of the number on non-extremal configurations.
\begin{cor}\label{cor:improvement}
Let $\delta=\frac{11}{6}-\frac{161}{88}$. For every neighboring coloring pair $(\sigma, \tau)$, we have
$$
\nabla_H(\sigma, \tau) \leq \left(\frac{11}{6}-\delta\left(1-\gamma_{\sigma,\tau}\left(v\right)\right)\right) \Delta-k\;.
$$
\end{cor}

\subsection{Contribution of the Extremal Part of the Metric}\label{sec:nablaB}

In this section we bound $\nabla_B(\sigma,\tau)$ from above. Define the contributions
\begin{align*}
\nabla_B(\sigma, \tau, c) &:=- nk \, \bE[ \mathbbm{1}_{X_c} \cdot (d_B(\sigma', \tau')-d_B(\sigma, \tau)) ] \;,\\
\overline{\nabla_B}(\sigma, \tau) &:=- nk\,  \bE[\mathbbm{1}_{\overline{X}} \cdot  (d_B(\sigma', \tau')-d_B(\sigma, \tau))] \;.
\end{align*}
By equations~\eqref{eq:omega} and~\eqref{eq:d_B}, since $d_H(\sigma,\tau)=1$, we have $d_B(\sigma,\tau)= 1-\omega(\sigma,\tau)=\eta (1-\gamma_{\sigma,\tau}(v))$. Moreover, $d_B(\sigma',\tau')\geq 0$. By the properties of the greedy coupling,
\begin{align*}
\nabla_B(\sigma, \tau, c) &\leq \eta(1-\gamma_{\sigma,\tau}(v)) nk  \bP[X_c=1]
  \leq 2\eta(1-\gamma_{\sigma,\tau}(v)) (\delta_c+1) \;,
\end{align*}
where we have used that the probability of flipping a given Kempe component is at most $1/nk$.

We can bound the expected change of $\nabla_B$ as follows
\begin{align}\label{eq:bound nabla_B}
\nabla_B(\sigma, \tau) &= \overline{\nabla_B}(\sigma, \tau) + \sum_{c\in [k]} \nabla_B(\sigma,\tau,c)\leq \overline{\nabla_B}(\sigma, \tau)+2\eta (k+\Delta) (1-\gamma_{\sigma,\tau}(v))  \;.
\end{align}
Let $\overline{\cD}= (\cS_\sigma \cup \cS_\tau)\setminus \cD$ denote the set of Kempe components of $\sigma$ and $\tau$ that do not involve vertex $v$. Note that each component in $\overline{\cD}$ is a Kempe component of both $\sigma$ and $\tau$. An important difference here as opposed to the analysis of the contribution of $\nabla_H$, is that the components in $\overline{\cD}$ have an effect on the expected change of $\nabla_B$. 
For a coloring $\sigma$ and $S\in \mathcal{S}_\sigma$, let $\sigma_S$ denote the coloring obtained by flipping $S$ in $\sigma$. 
For $S\in \overline{\cD}$, since $d_H(\sigma_S,\tau_S)=1$, we have $d_B(\sigma_S,\tau_S)= \eta(1-\gamma_{\sigma_S,\tau_S}(v))$.
It follows that, $\overline{\nabla_B}(\sigma, \tau) =\eta \sum_{S\in\overline{\cD}} p_{|S|} (\gamma_{\sigma_S, \tau_S}(v)-\gamma_{\sigma, \tau}(v)) .$

For each $c\in [k]$ and $i\in\{1,2\}$ and $S\in\overline{\cD}$, let 
$$
\xi_{\sigma,\tau}(v,c,S) := 
%\gamma_{\sigma_S, \tau_S}(v)-\gamma_{\sigma, \tau}(v)
\begin{cases}
-i & \text{if }c\in C^i_{\sigma,\tau}(v) \text{ and } c\notin C_{\sigma_S,\tau_S}(v),\\
i & \text{if }c\notin C_{\sigma,\tau}(v)  \text{ and }c\in C^i_{\sigma_S,\tau_S}(v),\\
-1 & \text{if }c\in C^2_{\sigma,\tau}(v) \text{ and } c\in C^1_{\sigma_S,\tau_S}(v),\\
1 & \text{if }c\in C^1_{\sigma,\tau}(v)  \text{ and }c\in C^2_{\sigma_S,\tau_S}(v), \\
0 & \text{otherwise}.
\end{cases}
$$
The variable $\xi_{\sigma,\tau}(v,c,S)$ can be understood as the contribution of color $c$ to $\gamma_{\sigma_S, \tau_S}(v)-\gamma_{\sigma, \tau}(v)$.
For every $\cS'\subseteq \overline{\cD}$, we define
$$
\overline{\nabla_B}(\sigma, \tau, c, \cS') := \frac{\eta}{\Delta}  \sum_{S\in \cS'} p_{|S|} \xi_{\sigma,\tau}(v,c,S)\;,
$$
and note that 
$
\overline{\nabla_B}(\sigma, \tau) = \sum_{c\in [k]} \overline{\nabla_B}(\sigma, \tau, c, \overline{\cD})  \;.
$

Next lemma bounds from above the contribution $\overline{\nabla_B}(\sigma, \tau, c, \overline{\cD})$ for each $c$. This is the most technical part of our approach and we defer its proof to Appendix~\ref{app:lemcomplement}.
\begin{lem}\label{lem:complement}
For every neighboring coloring pair $(\sigma,\tau)$ and color $c$, we have:
\begin{itemize}
\item[i)] For $i\in\{1,2\}$, if $c\in C^i_{\sigma,\tau}(v)$, then
$\overline{\nabla_B}(\sigma, \tau, c,\overline{\cD}) \leq -i \eta\left( \frac{k}{\Delta}-\frac{3}{2}\right)$;
\item[ii)] If $c\notin C_{\sigma,\tau}(v)$, then
$\overline{\nabla_B}(\sigma, \tau, c,\overline{\cD}) \leq  2\eta\left(9+\frac{15k}{\Delta}\right)$.
\end{itemize}
\end{lem}

The following bound on $\nabla_B$ follows directly from~\eqref{eq:bound nabla_B} and Lemma~\ref{lem:complement}.
\begin{cor}\label{cor:improvement2}
For every neighboring coloring pair $(\sigma,\tau)$, we have
$$ 
\nabla_B(\sigma, \tau) \leq - \eta\left(\frac{k}{\Delta}-\frac{3}{2}\right)\gamma_{\sigma,\tau}(v)+2\eta\left(10+\frac{16k}{\Delta}\right)(1-\gamma_{\sigma,\tau}(v))\;.
$$
\end{cor}

\subsection{Contraction of the Metric and Proof of Theorems~\ref{thm:main} and~\ref{thm:compareglauber}}\label{sec:proof_thm_DPP}

We now can show that the metric $d$ contracts in expectation.
\begin{thm}\label{thm:improvement}
For the flip parameters $\mathbf{\hat{p}}$ given in Observation~\ref{obs:DPP_assig} there exists $\epsilon_0,\mu>0$ such that for every $k\geq \left(\frac{11}{6}-\epsilon_0\right)\Delta$ and every neighboring coloring pair $(\sigma,\tau)$, the greedy coupling satisfies
$$
\nabla(\sigma,\tau)\leq -\mu k\;.
$$
\end{thm}
%The choice of $-1$ in the theorem is arbitrary, and proving that $\nabla(\sigma,\tau)\leq c$ for any $c<0$ would be enough to show that the mixing time of flip dynamics is $O(n\log{n})$.

%We conclude this section with the proof of Theorem~\ref{thm:improvement}.
\begin{proof}
Recall that $\delta =\frac{11}{6}-\frac{161}{88} = \frac{1}{264}$, and set $\eta = \frac{\delta\Delta}{53 k}$. Fix $\epsilon_0=\frac{1}{84000}$ and note that
$
 \left(\frac{11}{6}-\frac{\delta}{318}\right)\Delta \leq  k -\mu k
$ 
%$$
% k\geq  \left(\frac{11}{6}-\epsilon \right)\Delta = \left(\frac{11}{6}-\frac{\delta}{318}\right)\Delta +\mu\Delta\;,
%$$ 
for some small constant $\mu>0$. Note that $\frac{k}{\Delta}\geq \frac{9}{5}$.
Using~\eqref{eq:sum} and Corollaries~\ref{cor:improvement} and~\ref{cor:improvement2}, it follows that  
\begin{align*}
\nabla(\sigma,\tau)&\leq \left(\frac{11}{6} - \left(\delta-2\eta\left(10+\frac{16k}{\Delta}\right)\right)(1 - \gamma_{\sigma,\tau}(v)) - \eta \left(\frac{k}{\Delta}-\frac{3}{2}\right) \gamma_{\sigma,\tau}(v)\right)\Delta  -k \\
%&\leq \left(\frac{11}{6} - \left(\delta-\frac{52\eta k}{\Delta}\right)(1 - \gamma_{\sigma,\tau}(v)) - \frac{\eta k}{6 \Delta}  \gamma_{\sigma,\tau}(v)\right)\Delta  -k \\
%&= \left(\frac{11}{6} - \frac{\delta}{53}(1 - \gamma_{\sigma,\tau}(v)) - \frac{\delta}{318} \gamma_{\sigma,\tau}(v)\right)\Delta - k\\
&\leq \left(\frac{11}{6} - \frac{\delta}{318} \right)\Delta  - k \leq -\mu k\;. \qedhere
\end{align*}
\end{proof}

We now proceed with the proof of our main result. 
\begin{proof}[Proof of Theorem~\ref{thm:main}]
%Let $\Omega^* \subset \Omega = [k]^n$ be the set of proper $k$-colorings. 
Consider the metric $d$ on $\Omega$ defined in Section~\ref{sec:metric}. 
By Theorem~\ref{thm:improvement}, for the flip probabilities $\mathbf{\hat{p}}$ there exist $\epsilon_0, \mu>0$ such that if $k\geq \left(\frac{11}{6}-\epsilon_0\right)\Delta$, then the greedy coupling $(\sigma,\tau)\to (\sigma',\tau')$ defined on neighbouring coloring pairs $(\sigma,\tau)$ satisfies
 \begin{align}\label{eq:contract_2}
\bE\left[d(\sigma', \tau')\right] \leq  d(\sigma, \tau) - \frac{\mu}{n} \leq \left(1-\frac{\mu}{n}\right) d(\sigma, \tau)\;.
\end{align}

By Theorem~\ref{thm:path_coupling} with $\alpha=\mu/n$, we can extend the coupling over all $(\sigma,\tau)\in \Omega^2$ so~\eqref{eq:contract_2} is still satisfied. As $\eta< 1/2$, for $\sigma\neq \tau$ one has $ d(\sigma,\tau)\in (1/2,n]$. We use the coupling bound in~\eqref{eq:coupineq} together with Markov's inequality, to obtain for $\sigma^{(0)}\in \Omega$
\begin{align*}
\tvd(P^t(\sigma^{(0)},\cdot),\pi)\
&\leq \max_{ \tau^{(0)}\in \Omega} \bP[\sigma^{(t)}\neq \tau^{(t)}] 
= \max_{ \tau^{(0)}\in \Omega} \bP[d(\sigma^{(t)}, \tau^{(t)})\geq 1/2]
\leq \max_{ \tau^{(0)}\in \Omega} 2\,\bE[d(\sigma^{(t)}, \tau^{(t)})]\\
&\leq  2(1-\mu/n)^t n\;.
\end{align*}
It follows that $\tau_{\text{mix}}(\epsilon) \leq C n(\log{n}+\log{\epsilon^{-1}})$, for some absolute constant $C>0$.
%that the mixing time of the chain satisfies
%$$
%t_{\flip (\mathbf{p})}\leq nk\log{(8n)}\;.
%$$
%We can apply Theorem~\ref{thm:path coupling} with $\alpha=\frac{1}{nk}$ to conclude that the mixing time of the chain satisfies
%$$
%t_{\flip (\mathbf{p})}\leq nk\log{(4n)}\;.
%$$
\end{proof}
%
%\sitan{SITAN: Move this to a footnote in Section 1.2?}
%\trim{In order to proof Theorem~\ref{thm:compareglauber}, one can use the comparison technique of Diaconis and Saloff-Coste~\cite{diaconis1993comparison}. 
%Vigoda~\cite{vigoda2000improved} directly used the results in~\cite{diaconis1993comparison} to show that the mixing time of the Glauber dynamics is $O(n^2\log{n})$. It has been observed (see e.g.~\cite{frieze2007survey}) that Vigoda's result on flip dynamics implies $\tau_{\text{mix}} = O(n^2)$ for Glauber dynamics. This can be shown using spectral bounds on the mixing time for $\epsilon=1/n$~\cite{sinclair1992improved} and observing that the spectral gaps of Glauber dynamics and flip dynamics are the same up to a constant factor. % (see the discussion in~\cite{DJV}). 
%The same argument applies to our case, so Theorem~\ref{thm:main} implies Theorem~\ref{thm:compareglauber}.}

% list colorings
%!TEX root = fullpaper.tex

\section{List Colorings}\label{sec:list}

%Let $\delta_c=2$.  In the case when $c \not\in L(v)$, the big components do not flip and $11/3$ suffices. %(arg +1) 

%In their survey paper~\cite{FV07}, Frieze and Vigoda ask if the results obtained for colorings can be extended to list colorings. 
In this section we introduce the notation for list-colorings and give an overview of the proof of Theorem~\ref{thm:listcolorings}, deferring the details to Appendix~\ref{app:list}.

A \emph{list assignment} of $G$ is a function $L: V(G) \rightarrow 2^\mathbb{N}$. % that assigns to each $u\in V(G)$ a list of colors $L(u)$.  
An \emph{$L$-coloring} is a function $\sigma:V(G)\rightarrow \mathbb{N}$ such that $\sigma(u) \in L(u)$ for all $u \in V(G)$.  Usually in the literature list-colorings are assumed to be proper, here we will not require this but distinguish between proper and not necessarily proper list-colorings.  We denote by $\Omega^{L}$ the set of all $L$-colorings of $G$.
If $|L(u)| = k$ for all $u \in V(G)$, then we say that $L$ is a \emph{$k$-list-assignment} and that an $L$-coloring is a \emph{$k$-list-coloring}. 

The \emph{Glauber dynamics for $L$-colorings} is a discrete-time Markov chain $(\sigma^{(t)})$ with state space $\Omega^{L}$ where $\sigma^{(t)}$ is generated from $\sigma:=\sigma^{(t-1)}$ as follows:
%\guillem{is it $\Omega$ or $\Omega_0$ here?}
\begin{enumerate}
\setlength\itemsep{0em}
\item Choose $v^{(t)}$ uniformly at random from $V(G)$.
\item For all vertices $v \neq v^{(t)}$, let $\sigma^{(t)}(v) = \sigma(v)$.
\item Choose $c^{(t)}$ uniformly at random from $L(v^{(t)})$, if $c$ does not appear among the colors in the neighborhood of $v^{(t)}$ then let $\sigma^{(t)}(v^{(t)}) = c$, otherwise let $\sigma^{(t)}(v^{(t)}) = \sigma(v^{(t)})$.
\end{enumerate}
Although we define the chain over $\Omega^L$, $\sigma^{(t)}$ will converge to the uniform distribution on proper $L$-colorings, as in the case of non-list-colorings. %(we refer to the discussion at the end of Section~\ref{sec:flip} for further details).

%The proof strategy to show that Glauber dynamics for $k$-list-colorings is rapidly mixing provided that $k$ is large enough will be analogous to the non-list-coloring case.  Before describing the version of flip dynamics for list-colorings that we will analyze, we introduce some definitions.  

Given $\sigma\in \Omega^{L}$, one can define Kempe components of $\sigma$ as for colorings and we denote by $\cS^L_\sigma$ the multiset of Kempe components $S=S_\sigma(u,c)$ with $u\in V(G)$ and $c\in L(u)$. 
Let $\sigma_S$ be the coloring obtained from $\sigma$ by swapping the colors in $S$ and note that $\sigma_S$ is not necessarily an $L$-coloring as the new color of a vertex might not be in its list. 
%\begin{defn}
Given $S=S_\sigma(v,c)\in \cS^L_\sigma$, we say that $S$ is \emph{flippable} if for every $u\in S$ we have $\{\sigma(v),c\}\subseteq L(u)$. If $S$ is flippable, then $\sigma_S\in \Omega^L$.
%\end{defn}

Let $\vec{p} = \{p_\alpha\}_{\alpha\in \N_0}$ be a collection of \emph{flip parameters}. The \emph{flip dynamics for $L$-colorings} is a random process generating a sequence of colorings $\sigma^{(0)},\sigma^{(1)}, \sigma^{(2)},\cdots$ of $G$ where $\sigma^{(0)}$ is an arbitrary coloring in $\Omega^L$ and $\sigma^{(t)}$ is generated from $\sigma:=\sigma^{(t-1)}$ as follows:
\begin{enumerate}
\setlength\itemsep{0em}
\item Choose $v^{(t)}$ uniformly at random from $V(G)$.
\item Choose $c^{(t)}$ uniformly at random from $L(v^{(t)})$.
\item Let $S=S_{\sigma}(v^{(t)},c^{(t)})$ and $\alpha= |S|$. If $S$ is flippable, with probability $p_\alpha/\alpha$ let $\sigma^{(t)}= \sigma_{S}$, otherwise %be the colouring obtained from $\sigma$ by flipping $S_{\sigma}(u,c)$.  
let $\sigma^{(t)}= \sigma$.
\end{enumerate}

We prove the analogue of Theorem~\ref{thm:main} for list-colorings.
\begin{thm}\label{thm:list_flip}
The flip dynamics for sampling  $k$-list-colorings is rapidly mixing with mixing time $O(n \log n)$, for any $k \geq (\frac{11}{6} - \epsilon_0)\Delta$ where $\epsilon_0>0$ is the same constant from Theorem~\ref{thm:main}.
%There exists a choice of flip probabilities $\mathbf{p}$ such that for every $k\geq \left(\frac{11}{6}-\eta\right)\Delta$, with $\eta = \frac{1}{\const}$, and every $k$-list-assignment $L$, flip dynamics for $L$-colorings on a graph on $n$ vertices with maximum degree $\Delta$ has mixing time 
\end{thm}

The proof of this theorem follows the same strategy as the proof of Theorem~\ref{thm:main}. 
%All the notation introduced in Section~\ref{sec:prelim} can be extended to list-colorings in a natural way.
The main difference in the analysis of the flip dynamics for list-coloring is that some of the moves that are valid in the non-list case, are not legal here. An important observation is that~\eqref{eq:trivABbound} no longer holds for list-colorings. For instance, it might be the case that $A_c=0$ since $S_\sigma(v,c)$ is not flippable, while some of the $a^c_i\neq 0$. This leads to a weaker definition of realizable configuration and produces a linear program whose set of contraints is a superset of the constraints in Linear Program~\ref{def:prelp}. An analysis similar to the one given for $c\in\{\sigma(v),\tau(v)\}$ implies that the new constraints have positive slack for the optimal solutions of the linear program we use, so no new extremal configuration arises. Thus, the analysis of the expected change of the one-step greedy coupling with respect to the Hamming metric is the same as for non-list-colorings in Section~\ref{sec:LP}.

\sitan{The other key observation common to extending both approaches of this work to list colorings is that if $S=S_\sigma(v^{(t)},c^{(t)})$ has size $1$, then $S$ is always flippable since $c^{(t)}\in L(v^{(t)})$. This is enough to show that the same flip parameters used to prove Theorem~\ref{thm:main} in either approach also work in the list coloring setting.}

\sitan{In Section~\ref{sec:CM}, only size 1 flips are used to lower bound the probabilities of breaking apart extremal configurations and terminating the coupling, so Lemma~\ref{lem:badtogood}, Lemma~\ref{lem:goodtogoodend}, and the lower bound in Lemma~\ref{lem:probend} still hold. It is also obvious that upper bounds on these events (Lemma~\ref{lem:goodtobadend} and the upper bound in Lemma~\ref{lem:probend}) still hold, so Section~\ref{sec:CM} carries over to the list setting.}

In Section~\ref{sec:DPP}, the list version of Corollary~\ref{cor:improvement} holds for the flip parameters $\mathbf{\hat{p}}$ given by Observation~\ref{obs:DPP_assig}. As components of size one are the ones giving the negative contribution in Lemma~\ref{lem:complement}, the list version of Corollary~\ref{cor:improvement2} also holds. These two corollaries imply Theorem~\ref{thm:list_flip} in a similar way as in Section~\ref{sec:proof_thm_DPP}.

\sitan{We provide a full proof of Theorem~\ref{thm:list_flip} in Appendix~\ref{app:list} using the approach of Section~\ref{sec:DPP}. This implies Theorem~\ref{thm:listcolorings} in the same way that Theorem~\ref{thm:main} implies Theorem~\ref{thm:compareglauber}.}

%\section{Conclusion}

\section{Acknowledgments}

The authors would like to thank Eric Vigoda for his helpful suggestion regarding the mixing time in Theorem~\ref{thm:compareglauber}, as well as the anonymous reviewers for their feedback.

\bibliographystyle{alpha}
\bibliography{biblio}

\appendix

%!TEX root = fullpaper.tex

\section{Review of Vigoda's Greedy Coupling}
\label{app:vigodareview}

\subsection{The Coupling}

In this section we give a self-contained overview of Vigoda's coupling analysis. For the reader's convenience, we restate some notation that was already introduced in Section~\ref{subsec:introonestep}.

Fix a neighboring coloring pair $(G,\sigma,\tau)$. For $c\in[k]$, let $U_c$ denote the set of neighbors of $v$ (in either coloring) that are colored $c$, and let $\delta_c = |U_c|$. We will sometimes denote the vertices in $U_c$ by $\{u^c_1,...,u^c_{\delta_c}\}$; where $c$ is clear from context, we will simply denote these by $\{u_1,...,u_{\delta_c}\}$.

Note that the symmetric difference $\mathcal{D} = \calS_{\sigma}\Delta\calS_{\tau}$ is precisely the Kempe components $S_{\sigma}(u^c_i,\tau(v))$ and $S_{\sigma}(v,c)$ in $\sigma$ and the Kempe components $S_{\tau}(u^c_i,\sigma(v))$ and $S_{\tau}(v,c)$ in $\tau$, for all colors $c$ appearing in the neighborhood of $v$ and all $i\in[\delta_c]$. All other Kempe components are shared between $\sigma$ and $\tau$, so for those, it is enough to use the identity coupling. Note that for colors $c\neq\sigma(v),\tau(v)$ not appearing in $N(v)$, the identity coupling then matches the flip of $S_{\sigma}(v,c)$ to that of $S_{\tau}(v,c)$ so that the two colorings of $G$ become identical.

So the main concern is how to couple the flips of components in $\mathcal{D}$.  We can decompose $\mathcal{D}$ into $\cup_{c:\delta_c > 0} \mathcal{D}_c$, where the sets $\mathcal{D}_c$ are defined as follows:

\begin{defn}
Let	$\mathcal{D}_c$ be the set of Kempe components consisting of $S_{\sigma}(v,c),S_{\tau}(v,c)$, and all $S_{\sigma}(u,\tau(v))$ and $S_{\tau}(u,\sigma(v))$ for all $u\in U_c$.\label{appdef:Dc}
\end{defn}

Informally, $\mathcal{D}_c$ is the subset of $\mathcal{D}$ that involves the color $c$. It is easy to see that for $c\neq \sigma(v)$, \begin{equation}S_{\sigma}(v,c) = \left(\bigcup^{\delta_c}_{i=1}S_{\tau}(u^c_i,\sigma(v))\right)\cup\{v\},\label{appeq:sigmadecomp}\end{equation} and when $c = \sigma(v)$, $S_{\sigma}(v,c), S_{\tau}(u,\sigma(v))= \emptyset$ for any $u\in U_c$. Likewise we have that for $c\neq\tau(v)$, \begin{equation}S_{\tau}(v,c) = \left(\bigcup^{\delta_c}_{i=1}S_{\sigma}(u^c_i,\tau(v))\right)\cup\{v\},\label{eq:taudecomp}\end{equation} and when $c = \tau(v)$, $S_{\tau}(v,c), S_{\sigma}(u,\tau(v)) = \emptyset$ for $u\in U_c$. 

The sets $\mathcal{D}_c$ are disjoint except possibly the pair $\mathcal{D}_{\sigma(v)},D_{\tau(v)}$, as these both contain $(\sigma(v),\tau(v))$-colored Kempe components, though we defer this point to later.

\begin{remark}
One subtlety is that there may exist multiple neighbors $u'_1, \dots, u'_m\in N(v)$ which are colored $c$ but which satisfy $S_{\tau}(u'_1,\sigma(v)) = \cdots = S_{\tau}(u'_m,\sigma(v))$; to guarantee that the flip of each component is considered exactly once, redefine $S_{\tau}(u'_i,\sigma(v)) = \emptyset$ for all $1 < i\le m$. Handle the components $S_{\sigma}(u'_i,\tau(v))$ analogously.\label{remark:multiplicity}
\end{remark}

In \cite{vigoda2000improved}, Vigoda couples flips of Kempe components within $\mathcal{D}_c$ as follows. First we require some notation. For $c$ such that $\delta_c > 0$, define $A_c := |S_{\sigma}(v,c)|$, $B_c := |S_{\tau}(v,c)|$, $a^c_i = |S_{\tau}(u^c_i,\sigma(v))|$, and $b^c_i = |S_{\sigma}(u^c_i,\tau(v))|$. Define the vectors $\vec{a}^c := (a^c_i: i\in[\delta_c])$ and $\vec{b}^c := (b^c_i: i\in[\delta_c])$. Also define $a^c_{max} = \max_i a^c_i$ and denote the maximizing $i$ by $i^c_{max}$. Likewise define $b^c_{max} = \max_j b^c_j$ and denote the maximizing $j$ by $j^c_{max}$. When it is clear from context that we are just focusing on a generic color $c$, we will refer to these as $A,B,a_i,b_i,\vec{a},\vec{b}, a_{max}, i_{max},b_{max},j_{max}$. 
We say that neighboring coloring pair $(G,\sigma,\tau)$ has a \emph{configuration} $(A_c,B_c;\vec{a}^c,\vec{b}^c)$ of size $\delta_c$.

Naively we have the bounds \begin{equation}1 + a_{max}\le A \le 1+ \sum_i a_i, \ \ \ \ \ 1 + b_{max}\le B\le 1 + \sum_i b_i,\label{appeq:trivABbound}\end{equation} and moreover the upper bounds in \eqref{eq:trivABbound} are equalities when $c\neq\sigma(v),\tau(v)$.

Note that $S_{\sigma}(v,c)$ and $S_{\tau}(v,c)$ can be quite different but $S_{\sigma}(v,c) \supset S_{\tau}(u_i,\sigma(v))$ so it is easier to understand the overlap between these two components. Among all choices of $i$, this overlap is maximized for $i = i_{max}$, and the idea of Vigoda's coupling is thus to greedily couple the flips of the biggest components, i.e. $S_{\sigma}(v,c), S_{\tau}(v,c)$, to the flips of the next biggest components, i.e. $S_{\tau}(u_{i_{max}},\sigma(v)), S_{\sigma}(u_{j_{max}},\tau(v))$, and then to couple together as closely as possible the flips of $S_{\sigma}(u_i,\tau(v))$ and $S_{\tau}(u_i,\sigma(v))$ for each $i\in[\delta_c]$. Formally, assuming $p_1 \geq p_2 \geq \cdots $ we have:

\begin{enumerate}
	\item Flip $S_{\sigma}(v,c)$ and $S_{\tau}(u_{i_{max}},\sigma(v))$ together with probability $p_A$.
	\item Flip $S_{\tau}(v,c)$ and $S_{\sigma}(u_{j_{max}},\tau(v))$ together with probability $p_B$.
	\item For $i\in[\delta_c]$, define \begin{equation}
		q_i = \begin{cases}
			p_{a_i} - p_A & \text{if $i = i_{max}$} \\
			p_{a_i} & \text{otherwise}
		\end{cases}
		\label{appeq:qi}
	\end{equation}

	\begin{equation}
		q'_i = \begin{cases}
			p_{b_i} - p_B & \text{if $i = j_{max}$} \\
			p_{b_i} & \text{otherwise}
		\end{cases}
		\label{appeq:qprimei} 
	\end{equation}

	Note that $q_i$ and $q'_i$ are the remaining probability associated to flips $S_{\tau}(u_i,\sigma(v))$ and $S_{\sigma}(u_i,\tau(v))$ respectively.
	
	\begin{enumerate}
		\item Flip $S_{\tau}(u_i,\sigma(v))$ and $S_{\sigma}(u_i,\tau(v))$ together with probability $\min(q_i,q'_i)$
		\item Flip only $S_{\tau}(u_i,\sigma(v))$ together with probability $q_i - \min(q_i,q'_i)$
		\item Flip only $S_{\sigma}(u_i,\tau(v))$ together with probability $q'_i - \min(q_i,q'_i)$
	\end{enumerate}
\end{enumerate}

Coupled moves 1) and 2) change the Hamming distance by at most $A - a_{max} - 1$ and $B - b_{max} - 1$ respectively (with equality, for example, if $G$ is a tree rooted at $v$). For any given $i\in[\delta_c]$, coupled move 3a) changes the Hamming distance by $a_i + b_i - 1$, where the extra -1 term comes from the fact that $S_{\tau}(u_i,\sigma(v))$ and $S_{\sigma}(u_i,\tau(v))$ are of size $a_i$ and $b_i$ respectively but share vertex $u_i$. On the other hand, coupled moves 3b) and 3c) obviously change the Hamming distance by $a_i$ and $b_i$ respectively. For a configuration $(A,B;\vec{a},\vec{b})$, define

\begin{equation}
	H(A,B;\vec{a},\vec{b}) = (A - a_{max} - 1)p_A + (B - b_{max} - 1)p_B + \sum_i f(u_i), \label{appeq:H}
\end{equation}

where \begin{equation}
	f(u_i) = a_iq_i + b_iq'_i - \min(q_i,q'_i)\label{appeq:fi}
\end{equation}

The above discussion implies that for $c\neq\sigma(v),\tau(v)$ appearing in the neighborhood of $v$, \begin{equation}kn\cdot\E[\mathbbm{1}_{X_c}\cdot(d(\sigma',\tau') - 1)]\le H(A_c,B_c;\vec{a}^c,\vec{b}^c),\label{appeq:mainvigodaineqcomponent}\end{equation} where $X_c$ is the random event that the coupling flips components in $\mathcal{D}_c$ in both colorings. For $c$ not appearing in the neighborhood of $v$, the Hamming distance will not change if Kempe components containing the color $c$ are flipped in both colorings, as the coupling is the identity on these components, except if $v$ is flipped to $c$ in both colorings, in which case the Hamming distance decreases by 1.

Lastly, we review how the case of $c = \sigma(v),\tau(v)$ and $\mathcal{D}_{\sigma(v)}\cup D_{\tau(v)}\neq\emptyset$ is handled in \cite{vigoda2000improved}. This is the main place where one needs to be careful about the fact that neighboring coloring pairs $\sigma,\tau$ need not be proper.

\begin{remark}
	When $c =\sigma(v),\tau(v)$, we must make sure not to double count flips, as it is possible that $\mathcal{D}_{\sigma(v)}$ and $\mathcal{D}_{\tau(v)}$ share Kempe components. In this remark, suppose $\mathcal{D}_{\sigma(v)}\cap D_{\tau(v)}\neq\emptyset$. This can only happen if there exist $x_i,y_j\in N(v)$ colored $\sigma(v),\tau(v)$ respectively for which $S_{\sigma}(v,\tau(v)) = S_{\sigma}(x_i,\tau(v))$ and $S_{\tau}(v,\sigma(v)) = S_{\tau}(x_i,\sigma(v))$. To avoid double counting, Vigoda sets $S_{\sigma}(v,\tau(v)) = S_{\tau}(y_j,\sigma(v)) = \emptyset$ in this case. The bound \eqref{eq:mainvigodaineqcomponent} then holds for both $c = \sigma(v),\tau(v)$. The only difference is that some values among $A_c,B_c$ and the entries of $\vec{a}^c,\vec{b}^c$ will be zero, in which case we take $p_0 = 0$.

	Specifically, for $c = \tau(v)$, we have $A_c = 0$, $B_c = b^c_{max} = 0$, and at least one $a^c_j$ is zero, namely the one corresponding to the component $S_{\tau}(y_j,\sigma(v)) = S_{\tau}(v,\sigma(v))$. In this case one can check that $H(A_c,B_c;\vec{a}^c,\vec{b}^c) = \sum a^c_jp_{a^c_j}$, and provided $\alpha p_{\alpha}\le 1$ for all $\alpha$, this is at most $\delta_c - 1$.

	For $c = \sigma(v)$, we have $A_c = 0$, $a^c_i = 0$ for all $i$, and $B_c = \sum_j b^c_j$. Let $j^*$ be the index of the unique neighbor $u_{j^*}$ of $v$ for which $S_{\tau}(v,\sigma(v)) = S_{\sigma}(u_{j^*},\sigma(v))$. Then because $S_{\sigma}(u_{j^*},\sigma(v))$ contains $v$, we need to modify the definition of $b^c_{max}$. Let $b^c_{max} = \max_j(b^c_j - \mathbb{I}_{j = j^*})$ and denote the maximizing $j$ by $j^c_{max}$. Then the lower bound on $B_c$ in \eqref{eq:trivABbound} still holds, and \eqref{eq:mainvigodaineqcomponent} still holds. Moreover, if $\delta_c = 1$, then $\E[d(\sigma',\tau') - 1|X_{\sigma(v)}] = H(A_c,B_c;\vec{a}^c,\vec{b}^c) = 0$.

	\label{remark:specialcase}
	% Then $$\E[d(\sigma',\tau') - 1\vert X_{\tau(v)}]\le(\delta_{\tau(v)} - 1)\cdot(\max_{\alpha}\alpha p_{\alpha}),$$ so as long as $$\alpha\cdot p_{\alpha}\le 1 \ \forall \alpha,$$ we have that $\E[d(\sigma',\tau') - 1\vert X_{\tau(v)}] \le -1 + \delta_{\tau(v)}$.
\end{remark}

Henceforth, we will refer to the coupling defined above as the \emph{greedy coupling}. We can conclude the following, implicit in \cite{vigoda2000improved}:

\begin{lem}\label{applem:vigoda}
	Let $(\sigma,\tau)\mapsto(\sigma',\tau')$ be the greedy coupling. Then \begin{equation}\E[d(\sigma',\tau') - 1] \le\frac{1}{nk}\left(-|\{c:\delta_c = 0\}| + \sum_{c: \delta_c\neq 0}H(A_c,B_c;\vec{a}^c,\vec{b}^c)\right).\label{appeq:mainvigodaineq}\end{equation}
\end{lem}

The function $H$ implicitly depends on the choice of flip parameters $\{p_{\alpha}\}$, while $A_c,B_c,\vec{a}^c,\vec{b}^c$ depend on $(G,\sigma,\tau)$. The remaining analysis in \cite{vigoda2000improved} once \eqref{eq:mainvigodaineq} has been deduced essentially boils down to picking a good set of flip parameters.

\subsection{Enumerating Realizable Configurations}
\label{app:realizable}

In Section~\ref{sec:LP} we raised the issue of enumerating all realizable configurations in Linear Program~\ref{def:prelp}. In particular, while it was easy to enumerate realizable configurations $(A_c,B_c;\vec{a}^c,\vec{b}^c)$ for which $c\neq\sigma(v),\tau(v)$, we provided without proof two types of constraints (\eqref{eq:Hforsigmav} and \eqref{eq:pap}) that we claimed would handle realizable configurations for which $c = \sigma(v),\tau(v)$. In this subsection we fill in the details for why these two constraints suffice for configurations with $c = \sigma(v),\tau(v)$.

For $c = \sigma(v)$, by Remark~\ref{remark:specialcase}, any realizable $(A_c,B_c;\vec{a}^c,\vec{b}^c)$ satisfies $H(A_c,B_c;\vec{a}^c,\vec{b}^c) = 0$ if $\delta_c = 1$, and satisfies $A_c = 0$, $\vec{a}^c = (0,...,0)$, $B_c = \sum_i b^c_i$, and $$H(A_c,B_c,\vec{a}^c,\vec{b}^c) = (B_c - b^c_{max} - 1)p_{B_c} + \sum_{i\neq j_{max}} b^c_ip_{b^c_i} + b_{j_{max}}(p_{b^c_{j_{max}}} - p_{B_c}) \le (B_c - b^c_m)p_{B_c} + \sum b^c_ip_{b^c_i}$$ if $\delta_c > 1$. So the relaxed constraint \eqref{eq:Hforsigmav} covers all constraints corresponding to $c = \sigma(v)$.

For $c = \tau(v)$, we know by Remark~\ref{remark:specialcase} that when $c = \tau(v)$, \eqref{eq:pap} ensures that $H(A_c,B_c;\vec{a}^c,\vec{b}^c)\le\delta_c - 1$, and for any $\lambda > 1$ (corresponding to $k > d$, which is the regime we are interested in to begin with), we automatically have that $\delta_c - 1 < -1 + \lambda\cdot\delta_c$.

%!TEX root = fullpaper.tex

\section{Extremal Configurations for Vigoda's Choice of Flip Parameters and Missing Proofs from Section~\ref{sec:LP}}
\label{app:obsproof}

\begin{obs}
Consider the assignment~\eqref{eq:vigodasprobs} in Linear Program~\ref{def:lp} for $N_{max}=6$ and $m^*=3$. Constraint \eqref{eq:pap} is tight under the assignment \eqref{eq:vigodasprobs} only for $\alpha = 1$. Among the constraints of the form \eqref{eq:mainconstraint} associated to a realizable configuration $(A,B;\vec{a},\vec{b})$, up to symmetry, there are six tight constraints: 
\begin{itemize}
\item[i)] $m=1$, $A-1=a_1\in \{2,3,4,5\}$ and $B-1=b_1=1$;
\item[ii)]  $m=2$, $A=a_1+a_2+1$, $a_1=a_2\in\{2,3\}$, $B=1$ and $b_1=b_2=1$.
\end{itemize}	
Any other constraints of the form \eqref{eq:mainconstraint} that do not meet these conditions, and all constraints of the form \eqref{eq:Hforsigmav} and \eqref{eq:approxconstraint}, are not tight under the assignment \eqref{eq:vigodasprobs}. This can be verified numerically.
It follows that Vigoda's solution has six extremal realizable configurations, up to symmetries.
	\label{obs:slack}
\end{obs}

\begin{proof}
	The tightness of \eqref{eq:pap} only for $\alpha = 1$ is obvious. That the other constraints mentioned in the observation have zero slack can be checked by hand. We verify that all other constraints have nonzero slack.

	\setcounter{case}{0}
	\begin{case}
		Constraint \eqref{eq:mainconstraint} for $m = 1$
	\end{case}

	We first consider realizable $(A,B;\vec{a}^c,\vec{b}^c)$. It is easy to see that $(i - 1)(p_i - p_{i+1})\le 1/7$ with equality if and only if $i\in \{2,3,4,5\}$, and that $i(p_i - p_{i+1})\le 29/42$ with equality if and only if $i = 1$. Note that for $m = 1$, \begin{align*}H(A,B;\vec{a},\vec{b}) &= \max\left(a_1(p_{a_1} - p_{a_1+1}) + (b_1 - 1)(p_{b_1} - p_{b_1 + 1}), (a_1 - 1)(p_{a_1} - p_{a_1+1}) + b_1(p_{b_1} - p_{b_1 + 1})\right) \\ &\le\frac{29}{42} + \frac{1}{7} = \frac{5}{6},\end{align*} with equality if and only if $a_1 = 1$ and $b_1 \in  \{2,3,4,5\}$ or $b_1 = 1$ and $a_1\in  \{2,3,4,5\}$.

	\begin{case}
		Constraint \eqref{eq:mainconstraint} for $m = 2$
	\end{case}

	We analyze this case in the same way that Claim 6 of \cite{vigoda2000improved} is proved. Assume without loss of generality that $p_{a_{max}} - p_A \le p_{b_{max}} - p_B$ and $a_1\ge a_2$. In \cite{vigoda2000improved} it is noted that one may assume that $b_2\ge b_1$ so that $$H(A,B;\vec{a},\vec{b}) = (A - 2a_1)p_A + (B-2b_2 - 1) + (a_1 - 1)p_{a_1} + a_2p_{a_2} + b_1p_{b_1} + b_2p_{b_2} - \min(p_{a_2},p_{b_2} - p_B).$$ Now we proceed by casework on $\min(p_{a_2},p_{b_2} - p_B)$:

	\begin{itemize}
	 	\item $p_{a_2}\le p_{b_2} - p_B$: In this case we have \begin{equation}H(A,B;\vec{a},\vec{b}) = (a_1 - 1)p_{a_1} + (a_2 - 1)p_{a_2} + (A-2a_1)p_A + b_1p_{b_1} + b_2p_{b_2} + (B - 2b_2 - 1)p_B.\label{eq:case1}\end{equation}

One can check that $(a - 1)p_{a}\le 1/3$ with equality if and only if $a = 3$. If $a_1 = 3$, $(A-2a_1)p_A \le 0$ with equality if and only if $a_2=3$ and $6\le A\le 7$. However $(6,3;(3,3);(b_1,b_2))$ is not realizable. If $a_1\neq 3$, $(A-2a_1)p_A > 0$ if and only if $a_1 = a_2$ and $A = a_1+a_2 + 1$. It turns out that $p_3< 2p_2+p_5=4p_3+p_7=2/3$ and thus \eqref{eq:case1} is only maximized when $a_1=a_2\in \{2,3\}$ and $A=a_1+a_2+1$.
In a similar manner, we can verify that for any fixed $A, a_1,a_2$, \eqref{eq:case1} is only maximized when $b_1 = b_2 = 1$ and $B = 3$ and, in such a case, the contribution to the part involving $B,b_1,b_2$ is $2$. So, with the given assumptions, the only two configurations with $H(A,B;\vec{a},\vec{b}) = 8/3$ are $(7,3;(3,3);(1,1))$ and $(5,3;(2,2);(1,1))$, up to  .

	 	\item $p_{a_2}> p_{b_2} - p_B$: In this case we have \begin{equation*}
	 		H(A,B;\vec{a},\vec{b}) = (a_1 - 1)p_{a_1} + a_2p_{a_2} + (A - 2a_1)p_A + b_1p_{b_1} + (b_2 - 1)p_{b_2} + (B-2b_2)p_B.\end{equation*} This is symmetric with respect to flipping the roles of $(a_1,a_2)$ and $(b_2,b_1)$, and it can be verified that $(a_1 - 1)p_{a_1} + a_2p_{a_2} + (A-2a_1)p_A< 4/3$, so $H(A,B;\vec{a},\vec{b}) < 8/3$. 
	 		%On the other hand, as we have seen in above, for $(A,B; (a_1, a_2),(b_1,b_2)) = (7,3;(3,3),(1,1))$, we have that $H(A,B;\vec{a},\vec{b}) = 8/3$, concluding the proof for $m = 2$.
	\end{itemize}

	\begin{case}
		Constraint \eqref{eq:Hforsigmav}
	\end{case}

	For $m = 2$, the left-hand side of \eqref{eq:Hforsigmav} is $b_1p_{b_1 + b_2} + b_1p_{b_1} + b_2p_{b_2}$, which attains its maximum value of $p_2 + 2 < -1 + 2\lambda$ at $b_1 = b_2 = 1$. For $m > 2$, note that \eqref{eq:pap} implies that the left-hand side of \eqref{eq:Hforsigmav} is at most $m + 1 < -1 + \lambda\cdot m$ provided $\lambda > 5/3$, which is certainly the case.

	\begin{case}
	 	Constraint \eqref{eq:approxconstraint}
	\end{case}

	One can check that $(A-2)p_A\leq 4/21$. And if $p_a\le p_b$, then $a\cdot p_a + b\cdot p_b - \min(p_a,p_b) = (a-1)p_a + b\cdot p_b$. But $(a-1)p_a\le 1/3$ and $b\cdot p_b\le 1$. So $x^* = 4/21$ and $y^* = 4/3$, and it is clear that $-1 + (11/6)\cdot 3 > 2\cdot x^* + m^*\cdot y^*$ for $m^* = 3$, so \eqref{eq:approxconstraint} has nonzero slack.
\end{proof}

\begin{proof}[Proof of Lemma~\ref{lem:tight}]
	We first show there exists no choice of flip parameters for which greedy coupling contracts for all of $\mathcal{C}^*$. Let $\lambda = k/d$, and suppose to the contrary that $1\le \lambda < 11/6$ and yet there exists a set of flip parameters $\{p_{\alpha}\}$ for which all pairs of colorings in $\mathcal{C}^*$ contracted in distance. 
The expected change in distance for $G_1$ is 
	$$
	\frac{d}{nk}\cdot H(3,2;(2),(1)) = d\left(p_1 + p_2 -2p_3 - \min(p_1 - p_2,p_2 - p_{3})\right) < 0 \le \frac{d(-1 + \lambda)}{nk}.
	$$ 
	The expected change in distance for $G_2$ is 
	$$
	\frac{d}{2nk}\cdot H(7,3;(3,3),(1,1)) = (d/2)\left(2p_1 + 5p_3 - \min(p_1 - p_3,p_3-p_7)\right) < 0 \le \frac{d(-1 + 2\lambda)}{2nk}.
$$ 
As this implies $p_1-p_3\leq -1+\lambda$ and $2p_1+4p_3+p_7\leq -1+2\lambda$, using $p_1=1$ and $p_7\geq 0$ it yields $\lambda\geq 11/6$, a contradiction.

	Finally, it is straightforward to check that no one-step coupling can do better than the greedy coupling. This is clear for $G_1$. Indeed, certainly for any component not in $
	\mathcal{D}$, the coupling should just be the identity. Now for any neighbor $u$ of $v$ with color $c$, suppose a nonzero amount of probability mass $p$ for the flip of $S_{\tau}(u,\sigma(v))$ is matched in the optimal one-step coupling to the flip of a component other than $S_{\sigma}(v,c)$. The expected change in distance conditioned on this pair of components being chosen in the coupling is strictly greater than the expected change if that mass $p$ were instead reallocated to the empty flip in $\sigma$, contradicting optimality. By symmetry we can show that the flip of $S_{\sigma}(u_i,\tau(v))$ is coupled only to the empty flip in $\tau$ and the flip of $S_{\tau}(u_i,c)$. Finally, if not all of the probability mass for the flip of $S_{\sigma}(v,c)$ is matched to the flip of $S_{\tau}(u,\sigma(v))$, then we can strictly improve the coupling by reallocating that mass to $S_{\tau}(u,\sigma(v))$.

	A similar argument shows that the optimal one-step coupling for \new{$G_2$} is the greedy coupling.
\end{proof}

%!TEX root = fullpaper.tex

\section{Missing Proofs from Section~\ref{sec:CM}}
\label{app:CM}

\subsection{Proof of Lemma~\ref{lem:thought}}
\label{app:thoughtexperiment}
\begin{proof}
Denote a minimizing choice of $\{p_{\alpha}\}$ and $\lambda_{sing},\lambda_{bad},\lambda_{good}$ for the $\gamma$-mixed coupling LP by $\{p^*_{\alpha}\}$ and $\lambda^*_{sing},\lambda^*_{bad},\lambda^*_{good}$. Then for any $(G,\sigma,\tau)\in\mathcal{C}$ and the greedy coupling $(\sigma,\tau)\mapsto(\sigma',\tau')$, observe that \begin{align}\E[d_H(\sigma',\tau') - 1] &\le -|\{c: \delta_c = 1\}| + \sum_{\substack{c: \sigma,\tau \\ \text{\Sing{c}}}} H(A_c,B_c;\vec{a}^c,\vec{b}^c) + \sum_{\substack{c: \sigma,\tau \\ \text{\Bad{c}}}} H(A_c,B_c;\vec{a}^c,\vec{b}^c) + \sum_{\substack{c: \sigma,\tau \\ \text{\Good{c}}}} H(A_c,B_c;\vec{a}^c,\vec{b}^c)\nonumber\\
&\le -|\{c: \delta_c = 1\}| + \sum_{\substack{c: \sigma,\tau \\ \text{\Sing{c}}}}(-1 + \lambda_{sing}) + \sum_{\substack{c: \sigma,\tau \\ \text{\Bad{c}}}} (-1 + 2\lambda_{bad}) + \sum_{\substack{c: \sigma,\tau \\ \text{\Good{c}}}} (-1 + \delta_c\lambda_{good})\nonumber\\
&= -k + \lambda_{sing}\cdot N_{sing}(\sigma,\tau) + 2\lambda_{bad}\cdot N_{bad}(\sigma,\tau) + \sum_{\substack{c: \sigma,\tau \\ \text{\Good{c}}}}\delta_c\cdot \lambda_{good}.\label{eq:basicconvexcombo}
\end{align} But because $\delta_c\ge 2$ for any $c\neq\sigma(v),\tau(v)$ for which $\sigma,\tau$ are in state \Good{c}, because $\sigma,\tau$ are always in state \Good{\sigma(v)}, \Good{\tau(v)}, and because $$N_{sing}(\sigma,\tau) + 2N_{bad}(\sigma,\tau) + \sum_{c:\sigma,\tau \ \text{\Good{c}}} \delta_c = \Delta(v),$$ we conclude that \eqref{eq:basicconvexcombo} is a convex combination of the terms $$-k + \lambda^*_{sing}\cdot \Delta(v), \quad -k + \lambda^*_{good}\cdot \Delta(v), \quad-k + \left(\frac{\gamma}{\gamma+1}\cdot \lambda^*_{bad} + \frac{1}{\gamma + 1}\cdot \lambda^*_{good}\right)\Delta(v).$$ So we conclude that $\E[d_H(\sigma',\tau') - 1]\le -k + \lambda^*_{\gamma}\Delta(v) < 0$ as long as $k> \lambda^*_{\gamma}\Delta$.
\end{proof}

\subsection{Proof of Upper Bound in Lemma~\ref{lem:probend}}
\label{app:probend}
\begin{proof}
Fix a color $c$ for which $\delta_c\neq 0$ and some $i\in[\delta_c]$. For $i\neq i_{max},j_{max}$, the pairs of flips $(S_{\sigma}(u_i,\tau(v)),S_{\tau}(u_i,\sigma(v)))$, $(S_{\sigma}(u_i,\tau(v)),\emptyset)$, and $(\emptyset,S_{\tau}(u_i,\sigma(v)))$ have probability mass $\min(p_{b_i},p_{a_i})$, $\max(0,p_{b_i}-p_{a_i})$, and $\max(p_{a_i} - p_{b_i},0)$, for a total of $\max(p_{a_i},p_{b_i})$. The remaining pairs of flips have probability masses which depend on whether $i_{max} = j_{max}$, as shown in Table~\ref{table:masses}.

%!TEX root = ./fullpaper.tex

\begin{table}[ht]
\centering
\caption{Probability masses for some coupled flips}
\label{table:masses}
\begin{tabularx}{\linewidth}{|c|c|Y|Y|}\hline
Flip in $\sigma$                   & Flip in $\tau$                    & $i_{max} = j_{max}$                                     & $i_{max}\neq j_{max}$                                                                       \\ \hline
$S_{\sigma}(v,c)$                  & $S_{\tau}(u_{i_{max}},\sigma(v))$ & $p_A$                                                   & $p_A$                                                                                       \\ \hline
$S_{\tau}(v,c)$                    & $S_{\sigma}(u_{j_{max}},\tau(v))$ & $p_B$                                                   & $p_B$                                                                                       \\ \hline
$S_{\sigma}(u_{i_{max}},\tau(v))$  & $S_{\tau}(u_{i_{max}},\sigma(v))$ & $\min(p_{a_{i_{max}}}-p_A,p_{b_{i_{max}}}-p_B)$         & $\min(p_{a_{i_{max}}}-p_A,p_{b_{i_{max}}})$                                                 \\ \hline
$(S_{\sigma}(u_{j_{max}},\tau(v))$ & $S_{\tau}(u_{j_{max}},\sigma(v))$ & N/A                                                  & $\min(p_{b_{j_{max}}}-p_B,p_{a_{j_{max}}})$                                                 \\ \hline
$(S_{\sigma}(u_{i_{max}},\tau(v))$ & $\emptyset$                       & $\max(0,p_{b_{i_{max}}} - p_B - p_{a_{i_{max}}} + p_A)$ & $\max(0,p_{b_{i_{max}}} - p_{a_{i_{max}}} + p_A)$                                           \\ \hline
$\emptyset$                        & $S_{\tau}(u_{i_{max}},\sigma(v))$ & $\max(0,p_{a_{i_{max}}} - p_A - p_{b_{i_{max}}} + p_B)$ & $\max(0,p_{a_{i_{max}}} - p_A - p_{b_{i_{max}}})$                                           \\ \hline
$S_{\sigma}(u_{j_{max}},\tau(v))$  & $\emptyset$                       & N/A                                                  & $\max(0,p_{a_{j_{max}}} - p_{b_{j_{max}}} + p_B)$                                           \\ \hline
$\emptyset$                        & $S_{\tau}(u_{j_{max}},\sigma(v))$ & N/A                                                  & $\max(0,p_{b_{j_{max}}} - p_B - p_{a_{j_{max}}})$                                           \\ \hline\hline
\multicolumn{2}{|c|}{Total}                             & $\max(p_{a_{i_{max}}} + p_A,p_{b_{i_{max}}} + p_B)$     & $\max(p_{a_{j_{max}}} + p_B,p_{b_{j_{max}}}) + \max(p_{b_{i_{max}}} + p_A,p_{a_{i_{max}}})$\\ \hline
\end{tabularx}
\end{table}

From these we can conclude that $$nk\cdot\Pr[(S,S') \ \text{terminating}] \le \left(\sum_{c:\delta_c > 0}p_{A_c} + p_{B_c}\right) + \left(\sum_{c:\delta_c > 0,i\in[\delta_c]}\max(p_{a^c_i},p_{b^c_i})\right) + \left(\sum_{c: \delta_c = 0}p_1\right).$$ The sum of the second and third summands is at most $k$. For the first summand, note that when $A_c,B_c$ are nonzero and $\delta_c > 0$, $A_c,B_c\ge 2$, so the first summand is at most $2p_2\Delta$. The desired upper bound follows.
\end{proof}

\subsection{Proof of Lemma~\ref{lem:reduction}}
\label{app:reduction}
\begin{proof}
Let $\lambda^*_{sing}, \lambda^*_{tree}, \lambda^*_{good}$ be the values for $\lambda_{sing},\lambda_{tree},\lambda_{good}$ of the minimizer of the $C\gamma$-mixed coupling LP from Definition~\ref{def:mixedlp}. For Kempe components $S,S'$ in $\sigma,\tau$ respectively, define $E^{S,S'}_{\sigma,\tau} = d_H(\sigma',\tau') - 1$, where $\sigma',\tau'$ is the pair of colorings obtained by flipping $S$ in $\sigma$ and $S'$ in $\tau$, and let $p^{S,S'}_{\sigma^{(T-1)},\tau^{(T-1)}}$ be the probability that $S,S'$ are flipped in one step of greedy coupling starting from $\sigma^{(T-1)}, \tau^{(T-1)}$. Then we have that \begin{equation}
	\E[d_H(\sigma^{(\Tstop)},\tau^{(\Tstop)})-1] = \sum_{T,\sigma,\tau}\Pr[\sigma^{(T-1)} = \sigma,\tau^{(T-1)}=\tau]\cdot Z(\sigma^{(T-1)},\tau^{(T-1)}),\label{eq:expdiff}
\end{equation} where $$Z(\sigma^{(T-1)},\tau^{(T-1)}) := \sum_{S,S'}\mathbb{I}[(S,S') \ \text{terminating}]\cdot p^{S,S'}_{\sigma^{(T-1)},\tau^{(T-1)}}\cdot E^{S,S'}_{\sigma^{(T-1)},\tau^{(T-1)}}$$ But note that for $v^{(T)},c^{(T)}$ not terminating, $E^{S,S'}_{\sigma^{(T-1)},\tau^{(T-1)}} = 0$ because the one-step coupling is just the identity coupling, so $Z(\sigma^{(T-1)},\tau^{(T-1)})$ is just the expected change in distance under one step of greedy coupling on the neighboring coloring pair $(G,\sigma^{(T-1)},\tau^{(T-1)})$. Therefore, for any $T\le\Tstop$, \begin{align}
	Z(\sigma^{(T-1)},\tau^{(T-1)}) &= \E[d_H(\sigma^{(T)},\tau^{(T)}) - 1] \nonumber\\
	&\le \frac{1}{nk}\Big((-1 + \lambda^*_{sing})\cdot N_{sing}(\sigma^{(T-1)},\tau^{(T-1)}) + (-1 + 2\lambda^*_{bad})\cdot N_{bad}(\sigma^{(T-1)},\tau^{(T-1)}) \nonumber\\
	&\qquad \quad + \sum_{\substack{c: \sigma^{(T-1)},\tau^{(T-1)} \\ \text{\Good{c}}}} (-1 + \delta_c\cdot \lambda^*_{good}) - |\{c: \delta_c = 0\}|\Big) \nonumber\\
	&= \frac{1}{nk}\cdot\Big(-k + \lambda^*_{sing}N_{sing}(\sigma^{(T-1)},\tau^{(T-1)}) + 2\lambda^*_{bad}N_{bad}(\sigma^{(T-1)},\tau^{(T-1)}) \nonumber\\
	&\qquad \quad + \sum_{\substack{c: \sigma^{(T-1)},\tau^{(T-1)} \\ \text{\Good{c}}}}\delta_c\cdot \lambda^*_{good}\Big).\label{eq:Z}
\end{align} Because $$N_{sing}(\sigma,\tau) + 2N_{bad}(\sigma,\tau) +\sum_{\substack{c: \sigma,\tau \\ \text{\Good{c}}}}\delta_c = \Delta(v)$$ for all neighboring coloring pairs $(G,\sigma,\tau)$, we conclude from \eqref{eq:expdiff} and \eqref{eq:Z} that \begin{equation}\E[d_H(\sigma^{(\Tstop)},\tau^{(\Tstop)}) - 1]=\frac{1}{nk}\left(\sum_{T,\sigma,\tau}\Pr[\sigma^{(T-1)} = \sigma,\tau^{(T-1)} = \tau]\right)\left(-k + \lambda^*_{sing}\Delta_{sing} + \lambda^*_{bad}\Delta_{bad} + \lambda^*_{good}\Delta_{good}\right)\label{eq:convexcombo}\end{equation} for some $\Delta_{sing},\Delta_{bad},\Delta_{good}\ge 0$ satisfying \begin{equation}\Delta_{sing} + \Delta_{bad} + \Delta_{good} = \Delta(v).\label{eq:addtodv}\end{equation} But note that \begin{align}
	\frac{\Delta_{bad}}{\Delta_{good}} &= \frac{\sum_{T,\sigma,\tau}\Pr[\sigma^{(T-1)} = \sigma,\tau^{(T-1)} = \tau]\cdot N_{bad}(\sigma,\tau)}{\sum_{T,\sigma,\tau}\Pr[\sigma^{(T-1)} = \sigma,\tau^{(T-1)} = \tau]\cdot N_{good}(\sigma,\tau)} \nonumber\\
	&\le C\cdot\frac{\E[N_{bad}(\sigma^{(\Tstop - 1)},\tau^{(\Tstop - 1)})]}{\E[N_{good}(\sigma^{(\Tstop - 1)},\tau^{(\Tstop - 1)})]} \nonumber\\
	&\le C\gamma\label{eq:dbadtodgood}
\end{align} for $C := \frac{k + 2p_2\Delta}{k - \Delta - 2}$, where the second inequality follows by hypothesis and the first inequality follows by the fact that for $s\in\{\text{bad},\text{good}\}$, \begin{align*}
	\E[N_{s}(\sigma^{(\Tstop - 1)},\tau^{(\Tstop - 1)})] &= \sum_{T,\sigma,\tau}\Pr[\sigma^{(T-1)} = \sigma,\tau^{(T-1)} = \tau]\cdot \Pr[(S_T,S'_T) \ \text{terminating}]\cdot N_{s}(\sigma,\tau) \\
	&\in \left[\frac{k - \Delta - 2}{nk},\frac{k + 2p_2\Delta}{nk}\right]\cdot \sum_{T,\sigma,\tau}\Pr[\sigma^{(T-1)} = \sigma,\tau^{(T-1)} = \tau]\cdot N_{s}(\sigma,\tau),
\end{align*} where we use the notation $x \in [a,b]\cdot y$ to denote the fact that $a\cdot y\le x\le b\cdot y$. The first step above follows by definition and the second step follows by Lemma~\ref{lem:probend}.

Finally, observe that \eqref{eq:addtodv} and \eqref{eq:dbadtodgood} imply that $-k + \lambda^*_{sing}\Delta_{sing} + \lambda^*_{bad}\Delta_{bad} + \lambda^*_{good}\Delta_{good}$ is a convex combination of $-k + \lambda^*_{sing}\Delta(v)$, $-k + \lambda^*_{good}\Delta(v)$, and $-k + \left(\frac{C\gamma}{C\gamma + 1}\lambda^*_{bad} + \frac{1}{C\gamma + 1}\lambda^*_{good}\right)\Delta(v)$, so in particular from \eqref{eq:convexcombo} we get that $$\E[d_H(\sigma^{(\Tstop)},\tau^{(\Tstop)}) - 1]\le\frac{1}{nk}\left(\sum_{T,\sigma,\tau}\Pr[\sigma^{(T-1)} = \sigma,\tau^{(T-1)} = \tau]\right)(-k + \lambda^*_{C\gamma}\Delta(v))\le \frac{-k + \lambda^*_{C\gamma}\Delta(v)}{k - \Delta - 2},$$ where the final step follows from the fact that $$\sum_{T,\sigma,\tau}\Pr[\sigma^{(T-1)} = \sigma,\tau^{(T-1)} = \tau]\cdot\Pr[(S_T,S'_T) \ \text{terminating}] = 1$$ and the lower bound of Lemma~\ref{lem:probend}.
\end{proof}

\subsection{Proof of Lemmas from Section~\ref{subsec:burnin}}
\label{app:burnin}
\begin{proof}[Proof of Lemma~\ref{lem:badtogood}]
Without loss of generality suppose that $(A_c,B_c,\vec{a}^c,\vec{b}^c) = (7,3,(3,3),(1,1))$. Let $u_1,u_2$ be the two $c$-colored neighbors of $v$, and denote the elements of $S_{\tau}(u_1,\sigma(v))$ and $S_{\tau}(u_2,\sigma(v))$ by $\{u_1,w^1_1,w^2_1\}$ and $\{u_2,w^1_2,w^2_2\}$ respectively. We know that the vertices $\{w^1_1,w^2_1,w^1_2,w^2_2\}$ are all distinct. With probability $\frac{4}{n}\cdot\frac{k - \Delta - 1}{k}$, the pair of flips $(S,S')$ chosen under the greedy coupling satisfies $S = S' = S_{\sigma}(w^i_j,c')$ for some $i,j\in\{1,2\}$ and $c'\in A_{\sigma}(w^i_j)\backslash\{\sigma(w^i_j)\}$ (note that $A_{\sigma}(w^i_j)\backslash\{\sigma(w^i_j)\}$ contains neither $\sigma(v)$ nor c). In this case, the flips are just of vertex $w^i_j$ from color $\sigma(w^i_j)$ to a different color not already present in its neighborhood, so the neighboring coloring pair $\sigma',\tau'$ resulting from the flips is in state \Good{c}.
\end{proof}

\begin{proof}[Proof of Lemma~\ref{lem:goodtobadend}]
First note that in order for $\sigma',\tau'$ to be in stage \Badend{c} given that $\sigma,\tau$ were in stage \Good{c}, it must be that condition (iii) of Definition~\ref{def:stages} holds. Furthermore, the pair of components $(S,S')$ flipped to get from $\sigma,\tau$ to $\sigma',\tau'$ cannot be terminal, so $S = S'$.

Suppose that $\delta_c > 2$. In this case, the probability that $\sigma,\tau$ leave stage \Good{c} for stage \Badend{c} is at most the probability that enough $c$-colored neighbors of $v$ are flipped so that $\delta_c$ becomes at most 2. A Kempe component $S$ outside of $D_c$ and containing at least $(\delta_c - 2)$ $c$-colored neighbors of $v$ must be flipped in both colorings to achieve this, and the probability the greedy coupling chooses any particular such $(S,S)$ is $p_{\delta_c - 2}/(nk)$. The number of such Kempe components is at most $\delta_c\cdot(k - 2)$, so by a union bound the probability that $\delta_c$ becomes 2 is at most $\frac{\delta_c p_{\delta_c - 2}\cdot (k-2)}{nk} < 3/n$.

On the other hand, if $\delta_c < 2$, then $\sigma,\tau$ are not in stage \Good{c} to begin with. So for the rest of the proof, we consider the case of $\delta_c = 2$. We will proceed by casework on $(A_c,B_c,\vec{a}^c,\vec{b}^c)$, which we will denote as $(A,B,(a_1,a_2),(b_1,b_2))$ for simplicity.

Let $\mathcal{E}$ denote the event that $\sigma,\tau$ transition to stage \Badend{c}. Denote the two $c$-colored neighbors of $v$ by $u_1, u_2$. We have that $\mathcal{E}\subseteq \mathcal{E}_1\cup\mathcal{E}_2$, where $\mathcal{E}_1$ is the event that $u_1$ or $u_2$ is flipped in both colorings to a new color, and $\mathcal{E}_2$ is the event that $u_1$ or $u_2$ are not flipped but $\sigma,\tau$ nevertheless transition to stage \Badend{c}. Obviously $\Pr[\mathcal{E}_1]\le 2/n$. We now proceed to bound $\Pr[\mathcal{E}_2]$.

\begin{case}
	If $a_i > 3$ or $b_i > 3$ for some $i = 1,2$, then $\Pr[\mathcal{E}_2]\le\frac{3}{n}$.\label{case:worst}
\end{case}

% tight example for this: C connected to an R-C as well as an R. 

\begin{proof}
	Without loss of generality, say that $a_1 > 3$. From the vertices of $S_{\tau}(u_1,\sigma(v))$ pick out $w,w'\neq u_1$ such that $w,w',u_1$ form a Kempe component. We have that event $\mathcal{E}_2\subseteq\mathcal{A}\cup\mathcal{B}$, where $\mathcal{A}$ is the event that all vertices in $S_{\tau}(u_1,\sigma(v))\backslash\{u_1,w,w'\}$ are flipped so that $S_{\tau'}(u_1,\sigma(v))\subseteq\{u_1,w,w'\}$, and $\mathcal{B}$ is the event that $w$ or $w'$ is flipped and no longer belongs to $S_{\tau'}(u_1,\sigma(v))$. Obviously $\Pr[\mathcal{B}]\le 2/n$. For $\mathcal{A}$, the $(a_1 - 3)$ neighbors of $u_1,w,w'$ in $S_{\tau}(u_1,\sigma(v))$ must be flipped at once, which by a union bound occurs with probability at most $\frac{1}{n}\cdot(a_1 - 3)\cdot p_{a_1 - 3}\le\frac{1}{n}$, where the inequality follows by \eqref{eq:pap}. So $\Pr[\mathcal{E}_2]\le\Pr[\mathcal{A}] + \Pr[\mathcal{B}]\le 3/n$.
\end{proof}

\begin{case}
	If $a_i = 0$ for some $i$ and $b_1,b_2\le 3$, or if $b_i = 0$ for some $i$ and $a_1,a_2\le 3$, then $P[\mathcal{E}_2]\le\frac{1}{n}$.
\end{case}

\begin{proof}
	Suppose without loss of generality that $a_1 = 0$ and $b_1,b_2\le 3$. By the definition of the greedy coupling and the fact that $c\neq\sigma(v),\tau(v)$, $a_1 = 0$ if and only if $S_{\tau}(u_2,\sigma(v))$ consists of $u_1,u_2,w$ for some $\sigma(v)$-colored $w\in N(u_1)\cup N(u_2)$. So $\mathcal{E}_2$ is a subset of the event that $w$ is flipped to any other color. Thus, $\Pr[\mathcal{E}_2]\le\frac{1}{n}$.
\end{proof}

\begin{case}
	If $1\le a_1,a_2,b_1,b_2\le 3$, and if $(a_1,a_2)$ and $(b_1,b_2)$ are both not among $\{(1,1),(3,3)\}$, then $P[\mathcal{E}_2]\le\frac{48}{nk}$.\label{case:insignificant}
\end{case}

\begin{proof}
Suppose $(a_1,a_2)$ and $(b_1,b_2)$ are both not among $\{(1,1),(3,3)\}$. Then $\mathcal{E}_2$ is a subset of the event that the pair of flips $(S,S)$ chosen increases or decreases at least one of $a_1,a_2$ and decreases or increases at least one of $b_1,b_2$, respectively. But for a flip $S$ to decrease some $a_i$ for $i\in\{1,2\}$, it must contain a member of $S_{\tau}(u_i,\sigma(v))$, and for a flip $S$ to increase some $b_j$ for $j\in\{1,2\}$, it must contain the color $c$ or $\tau(v)$. There are at most 3 members of $S_{\tau}(u_i,\sigma(v))$, so the probability of $(S,S)$ both increasing $a_i$ and decreasing $b_j$ is at most $\frac{3}{n}\cdot\frac{2}{k} = \frac{6}{nk}$, and by a union bound over the eight different choices of $i,j$, and increasing/decreasing, we conclude that $\Pr[\mathcal{E}_2]\le\frac{48}{nk}$.
\end{proof}

\begin{case}
	If $1\le a_1,a_2,b_1,b_2\le 3$ and exactly one of the tuples $(a_1,a_2)$ and $(b_1,b_2)$ is among $\{(1,1),(3,3)\}$, then $P[\mathcal{E}_2]\le\frac{4\Delta + 48}{nk}$.
\end{case}

\begin{proof}
Suppose $(b_1,b_2) = (1,1)$. $\mathcal{E}_2\subseteq\mathcal{X}\cup\mathcal{Y}$, where $\mathcal{X}$ is the event that the pair of flips $(S,S)$ chosen increases or decreases some $a_i$ and decreases or increases some $b_i$, respectively, and $\mathcal{Y}$ is the event that $(a_1,a_2)$ becomes $(3,3)$. We already know by Case~\ref{case:insignificant} that $\Pr[\mathcal{X}]\le \frac{48}{nk}$. Supposing without loss of generality that $a_1 < 3$, the event $\mathcal{Y}$ is a subset of the event that a neighbor of a vertex in $S_{\tau}(u_1,\sigma(v))$ is flipped to the color $c$ or $\sigma(v)$. There are at most $2\Delta$ such neighbors, so $\Pr[\mathcal{Y}]\le\frac{2\Delta}{n}\cdot\frac{2}{k} = \frac{4\Delta}{nk}$, and thus $\Pr[\mathcal{E}_2]\le\frac{4\Delta + 48}{nk}$.

Now suppose $(b_1,b_2) = (3,3)$. $\mathcal{E}_2\subseteq\mathcal{X}\cup\mathcal{Z}$ where $\mathcal{X}$ is the event defined above and $\mathcal{Z}$ is the event that $(a_1,a_2)$ becomes $(1,1)$. Supposing without loss of generality that $a_1 > 1$, $\mathcal{Z}$ is a subset of the event that one of the members of $S_{\tau}(u_1,\sigma(v))$ other than $u_1$ is flipped. There are at most two such vertices, so $\Pr[\mathcal{Z}]\le 2/n$ and $\Pr[\mathcal{E}_2]\le\frac{2k + 48}{nk}$.
\end{proof}

\begin{case}
	If $1\le a_1,a_2,b_1,b_2\le 3$ and $(a_1,a_2,b_1,b_2) = (1,1,1,1)$, then $P[\mathcal{E}_2]\le\frac{4\Delta}{nk}$.
\end{case}

\begin{proof}
	$\mathcal{E}_2$ is a subset of the event that one of the neighbors of $u_1$ or $u_2$ is flipped to the color $\sigma(v)$ or $\tau(v)$, so $\Pr[\mathcal{E}_2]\le\frac{2\Delta}{n}\cdot\frac{2}{k} =\frac{4\Delta}{nk}$.
\end{proof}

\begin{case}
	If $1\le a_1,a_2,b_1,b_2\le 3$ and $(a_1,a_2,b_1,b_2) = (3,3,3,3)$, then $P[\mathcal{E}_2]\le\frac{2}{n}$.
\end{case}

\begin{proof}
	$\mathcal{E}_2\subseteq\mathcal{S}\cup\mathcal{T}$, where $\mathcal{S}$ (resp. $\mathcal{T}$) is the event that all $\sigma(v)$-colored (resp. $\tau(v)$-colored) neighbors in $N(u_1)\cup N(u_2)$ in $\tau$ (resp. $\sigma$) are flipped to a different color. Consider an arbitrary $\sigma$-colored neighbor $w$ of $u_1$. $\mathcal{S}$ is a subset of the event that $w$ is flipped, so $\Pr[\mathcal{S}]\le 1/n$. We can bound $\Pr[\mathcal{T}]$ similarly, so $\Pr[\mathcal{E}_2]\le 2/n$.
\end{proof}

Of the upper bounds on $\Pr[\mathcal{E}_2]$ in all of the above cases, the bound of $3/n$ from Case~\ref{case:worst} is the greatest when $k\ge 1.833\Delta$, completing the proof of Lemma~\ref{lem:goodtobadend}.\end{proof}

%!TEX root = fullpaper.tex

\section{Missing Proofs from Section~\ref{sec:DPP}}\label{app:DPP}

\subsection{Proof of Lemma~\ref{lem:LP_equiv}}\label{app:lemLP}

\begin{proof}
We need to prove that for every $m$ and every realizable configuration $(A,B;\mathbf{a},\mathbf{b})$ of size $m$ different from the four excluded ones, we have 
\begin{align}\label{eq:goal}
H(A,B;\mathbf{a},\mathbf{b})\leq -1 + m\hat{\lambda}.
\end{align}
It is straightforward to check that for every $\alpha\in\{1,\dots,6\}$ the given assignment satisfies $\alpha \hat{p}_{\alpha}\leq 1$, $(\alpha-1) \hat{p}_{\alpha}\leq 1/3$ and $(\alpha-2)\hat{p}_\alpha \leq  (3\hat{\lambda}-5)/2$; we will use some of these inequalities in the proof. 

Let us first assume that $c\neq \sigma(v),\tau(v)$, so $A=1+\sum_i a_i$ and $B=1+\sum_{i} b_i$.
If $m=1$, then any realizable configuration has the form $(i+1,j+1,(i),(j))$.  Observe that~\eqref{eq:constr_1} correspond to the constraint  $H(i+1,j+1,(i),(j))\leq -1+\lambda$. We may assume that $i,j\leq 6$ and that $j\neq 1$ as otherwise we obtain weaker constraints, since $\hat{p}_\alpha=0$ for $\alpha\geq 7$. Thus~\eqref{eq:goal} follows for every configuration of size $1$ in Linear Program~\ref{def:prelp_DPP}, from the constraints in Linear Program~\ref{def:lp_DPP}.

If $m=2$, as $\hat{p}_\alpha=0$ for every $\alpha\geq 7$, there is a finite amount of non-trivial realizable non-extremal configurations $(A,B;\mathbf{a},\mathbf{b})$ of size $2$.
One can check by computer that for the given values of $\mathbf{\hat{p}}$ and $\hat{\lambda}$, any configuration of size $2$ in Linear Program~\ref{def:prelp_DPP} satisfies $H(A,B;\mathbf{a},\mathbf{b})\leq -1 +2\hat{\lambda}$, with equality if and only if the configuration is $(5,3,(2,2),(1,1))$ or $(3,5, (1,1),(2,2))$.

Suppose that $m=3$. By Lemma~\ref{lem:crude}, recall that
$$
H(A,B;\mathbf{a},\mathbf{b}) \leq  (A-2)p_A+(B-2)p_B +\sum_{i} (a_ip_{a_i} + b_i p_{b_i}- \min\{p_{a_i},p_{b_i}\}).
$$
Using the properties of $\mathbf{\hat{p}}$, we have $a_ip_{a_i} + b_i p_{b_i}- \min\{p_{a_i},p_{b_i}\}\leq 4/3$. Thus,
\begin{align*}
H(A,B;\mathbf{a},\mathbf{b})  &\leq (A-2)p_A+(B-2)p_B +\sum_{i}  (a_ip_{a_i} + b_i p_{b_i}- \min\{p_{a_i},p_{b_i}\})
\leq -1+3\hat{\lambda}\;.
%\frac{1}{2} - 3\left(\frac{11}{6}-\lambda\right) + 4 \\
%&=& \frac{9}{2} - 3\left(\frac{11}{6}-\lambda\right)= 3\lambda -1\;.
\end{align*}
If $m\geq 4$ and since $\hat{\lambda}>4/3$, we have
$$H(A,B;\mathbf{a},\mathbf{b}) \leq 3\hat{\lambda}-5 +4m/3 \leq -1 +m \hat{\lambda}\;.$$

We finally deal with the case $c\in\{\sigma(v),\tau(v)\}$. If $c=\tau(v)$, Remark~\ref{remark:specialcase} and $\alpha p_{\alpha}\leq 1$ implies $H(A,B;\mathbf{a},\mathbf{b})\leq -1+m\leq -1+m\hat{\lambda}$. If $c=\sigma(v)$ and $m=1$ then Remark 2.2 implies $H(A,B;\mathbf{a},\mathbf{b})=0\leq -1+\hat{\lambda}$. Finally, if $c=\sigma(v)$ and $m\geq 2$, then~\eqref{eq:Hforsigmav} holds for $\hat{\lambda}$ as $\alpha p_{\alpha}\leq 1$, $(\alpha-1) p_{\alpha}\leq 1/3$, so $H(A,B;\mathbf{a},\mathbf{b})\leq -1+m \hat{\lambda}$.

We conclude that $\mathbf{\hat{p}}$ is a feasible solution to Linear Program~\ref{def:prelp_DPP} with objective value $\hat{\lambda}$.
\end{proof}

\subsection{Proof of Lemma~\ref{lem:complement}}\label{app:lemcomplement}

\begin{proof}

Recall that the extremal configurations for our choice of flip parameters are $(3,2,(2),(1))$ and $(7,3, (3,3), (1,1))$, up to symmetries, and $U_c=\{u_1,\dots,u_{m}\}$, where $m=\delta_c$.

Assume first that $c\in C^i_{\sigma,\tau}(v)$ for some $i\in \{1,2\}$. Consider the sets of components
\begin{align*}
\cS_0&:= \{S\in \overline{\cD}: c\notin C_{\sigma_S,\tau_S}(v)\}\;,\\
\cS_2&:= \{S\in \overline{\cD}: c\in C^2_{\sigma_S,\tau_S}(v)\}\;.
\end{align*}
Note that when $i=1$, then for every $S\in \overline{\cD}\setminus  (\cS_0\cup \cS_2)$ we have $\xi_{\sigma,\tau}(v,c,S)\leq 0$; therefore,
$$
\overline{\nabla_B}(\sigma, \tau, c,\overline{\cD}) \leq 
\overline{\nabla_B}(\sigma, \tau, c,\cS_0)+ \overline{\nabla_B}(\sigma, \tau, c,\cS_2)\;.
$$
Note that when $i=2$, then for every $S\in \overline{\cD} \setminus \cS_0$ we have $\xi_{\sigma,\tau}(v,c,S)\leq 0$; therefore,
$$
\overline{\nabla_B}(\sigma, \tau, c,\overline{\cD}) \leq 
\overline{\nabla_B}(\sigma, \tau, c,\cS_0) \;.
$$

We proceed to bound $\overline{\nabla_B}(\sigma, \tau, c,\cS_0)$ for $i\in\{1,2\}$.
Without loss of generality, assume that $a_1>b_1$. Let $w\in S_\tau(u_1,\sigma(v))$ with $\tau(w)=\sigma(v)$; we note that  $w\notin U_c\cup\{v\}$ and that such a vertex always exists as $a_1\geq 2$.
Choose a color $c'\in [k]$ with $c'\notin \sigma(N(w))\cup\{\sigma(v),\tau(v)\}$. Let $S=S_\sigma(w,c') \in \overline{\cD}$. As $S=\{w\}$, $(\sigma_{S},\tau_S) $ has either a $(2,2, (1),(1))$ or a $(j+4,3,(j,3),(1,1))$ (with $j\in\{1,2\}$) configuration for $c$, i.e.  $c\notin C_{\sigma_S,\tau_S}(v)$. As there are at least $k-\Delta-2$ choices for $c'$ and as $p_{|S|}=p_1=1$, we have
$$
\overline{\nabla_B}(\sigma, \tau, c,\cS_0)\leq - \frac{\eta(k-\Delta-2)}{\Delta
}\cdot i\;.
$$

Now we bound $\overline{\nabla_B}(\sigma, \tau, c,\cS_2)$, provided that $i=1$. %Assume that $(\sigma,\tau)$ has a configuration $(3,3;1,1)$ and that $(\sigma,\tau)$ has a configuration $(2;1)$. 
Let $S\in \cS_2$, then $|S\cap (N(v)\setminus\{u_1\})|\geq 1$ and if $u\in S\cap (N(v)\setminus \{u_1\})$, then $\sigma_S(u)=c$. Thus, $S$  can be described as $S=S_\sigma(u,c)$ for $u\in N(v)$, implying that $|\cS_2|\leq \Delta$.  Moreover, $|S|\geq 2$ as at least two vertices need to change their color to transform an extremal $1$-configuration into an extremal $2$-configuration. Since $p_{|S|}\leq p_2\leq \frac{1}{3}$ and $\xi_{\sigma,\tau}(v,c,S)=1$, we have
$$
\overline{\nabla_B}(\sigma, \tau, c,\cS_2)\leq \frac{\eta}{3}\;.
$$

From the bounds on $\overline{\nabla_B}(\sigma, \tau, c,\cS_0)$ and $\overline{\nabla_B}(\sigma, \tau, c,\cS_2)$ derived above, we obtain that for $i\in\{1,2\}$ and $c\in C^i_{\sigma,\tau}(v)$
$$
\overline{\nabla_B}(\sigma, \tau, c,\overline{\cD})\leq- \frac{\eta \left(k- \frac{4\Delta}{3}-2\right)}{\Delta} \cdot i \leq -i \eta\left( \frac{k}{\Delta}-\frac{3}{2}\right) \;,
$$
and this proves the first statement.

%$$
%\overline{\nabla^i_B}(\sigma,\tau,c)\leq - \frac{\eta \cdot i}{\Delta}\cdot (k-\Delta-2)\leq - (4\eta/5)i\;.
%$$

\medskip

To prove the second statement, assume that $c\notin C_{\sigma,\tau}(v)$ and let $\cT:=\{S\in \overline{\cD}: c\in C_{\sigma_S,\tau_S}(v)\}$. 
Again, for every $S\in \overline{\cD} \setminus \cT$, we have $\xi_{\sigma,\tau}(v,c,S)\leq 0$. Therefore,
$$
\overline{\nabla_B}(\sigma, \tau, c,\overline{\cD}) \leq 
\overline{\nabla_B}(\sigma, \tau, c,\cT)\;.
$$
Define $U_c^S:=N(v)\cap (\sigma_S)^{-1}(c)$ with $m^S:=|U_c^S|$ and note that $m^S\leq 2$.
Consider the partition $\cT=\cT_1\cup\cT_2\cup \cT_3$ with
\begin{align*}
\cT_1&:= \{S\in \overline{\cD}: U_c\setminus U_c^S\neq \emptyset\}\;,\\
\cT_2&:= \{S\in \overline{\cD}: U_c^S\setminus U_c\neq \emptyset\}\setminus \cT_1\;,\\
\cT_3&:= \{S\in \overline{\cD}: U_c^S= U_c\}\;.
\end{align*}
For every $S\in \cT_3$, if $c\in C^1_{\sigma_S,\tau_S}(v)$, let $(A^S,B^S, (a^S_1),(b^S_1))$ be the extremal $1$-configuration for $c$ in $(\sigma_S,\tau_S)$ and if $c\in C^2_{\sigma_S,\tau_S}(v)$, let $(A^S,B^S,(a^S_1,a^S_{2}),(b^S_1,b^S_{2}))$ be the extremal $2$-configuration for $c$ in $(\sigma_S,\tau_S)$.
Recall that $(A,B;(a_1,\dots,a_m),(b_1,\dots,b_m))$ denotes the configuration for $c$ in $(\sigma,\tau)$. As it is non-extremal, there exists $x\in\{a,b\}$ and $j\in [m^S]$, such that $x_j\neq x^S_j$. Note that $x^S_j\leq 3$. 

Consider the partition $\cT_3=\cT_3^+ \cup\cT_3^- $ with
\begin{align*}
\cT_3^+&:= \{S\in \cT_3: x_j> x^S_j\}\;,\\
\cT_3^-&:= \{S\in \cT_3: x_j< x^S_j\}\;.
\end{align*}
To bound the size of $\cS'\in \{\cT_1,\cT_3^+\}$ we will proceed as follows. For every $S\in \cS'$, there is a vertex in a Kempe component of either $\sigma$ or $\tau$ that does not belong to the corresponding component in either $\sigma_S$ or $\tau_S$.  If there exists $R(\cS')\subseteq S_\sigma(v,c)\cup S_\tau(v,c)$ such that $S\cap R(\cS')\neq \emptyset$ for every $S\in \cS'$, then, any $S\in \cS'$ can be described as $S=S_\sigma(w,c')$ for $w\in R(\cS')$ and $c'\in [k]$, and $|\cS'|\leq |R(\cS')| k$.

If $\cS'= \cT_1$ and $S\in \cS'$, then observe that $|S\cap U_c|= |U_c\setminus U_c^S|\geq \max\{m-m^S,1\}$. Let $\ell=\min\{m^S+1, m\} $. If $R(\cT_1)=R_1=\{u_1,\dots, u_{\ell}\}$, it follows that $|S\cap R_1|\geq |S\cap U_c|- (m-(m^S+1)) \geq 1$ and 
%. Thus, $S=S(w_j,c')$ for some $j\in [m]$ and $c'\in [k]$ and 
$|\cT_1|\leq  (m^S+1) k\leq 3 k$. 

If $\cS'= \cT_3^+$ and $S\in \cS'$, recall that $x_j>x^S_j$ and set $\varphi=\sigma$ if $x=b$ and  $\varphi=\tau$ if $x=a$, and let $\pi\in\{\sigma,\tau\}\setminus  \{\varphi\}$.
Let $R(\cT_3^+)=R_3$ be an arbitrary set of $x^S_j$ vertices in  $S_{\varphi}(u_j,\pi(v)) \setminus \left\{u_j\right\}$. As $u_j\notin R_3$, we have $S\cap R_3\neq \emptyset$.  
Since there are $4$ choices for the extremal configuration, we have $|\cT_3^+|\leq 4 x^S_j k\leq  12 k$.
\medskip

%
%To upper bound the size of either $\cT_1$ or $\cT_3^+$, we choose a set of vertices $R$ such that $R\cap S\neq \emptyset$, then  $S=S_\varphi(u,c')$ for $u\in R$ and $c'\in [k]$.

To bound the size of $\cS'\in \{\cT_2,\cT_3^-\}$ we will proceed as follows. For every $S\in \cS'$, there is a vertex in the neighborhood of a Kempe component of either $\sigma$ or $\tau$, that belongs to the corresponding component in either $\sigma_S$ or $\tau_S$.  If there exists a set $N(\cS')$ of neighbors of $S_\varphi(v,c)$ such that $S\cap N(\cS')\neq \emptyset$ for every $S\in \cS'$, then, any $S\in \cS'$ can be described as $S=S_\varphi(w,c')$ for $w\in N(\cS')$ and a unique $c'\in\{c,\pi(v)\}$, and $|\cS'|\leq |N(\cS')|$.

If $\cS'= \cT_2$ and $S\in \cS'$, then let $N(\cT_2)=N_2=N(v)\setminus U_c$. Clearly $S\cap N_2 \neq \emptyset$ and 
$|\cT_2|\leq \Delta$.

If $\cS'= \cT_3^-$ and $S\in \cS'$,  recall that $x_j<x^S_j$ and set $\varphi=\sigma$ if $x=b$ and  $\varphi=\tau$ if $x=a$, and let $\pi\in\{\sigma,\tau\}\setminus  \{\varphi\}$.
Let $N(\cT_3^-)=N_3$ be the set of neighbors of $S_\varphi (u_j,\pi(v)) $, which satisfies $S\cap N_3\neq\emptyset$. 
%Thus, for each $u\in N_3$ there exists a unique $c'\in \{c,\pi(v)\}$ such that $S=S_\varphi(u,c')$. 
As $S\in \cT_3^-$, $|S|\leq x_j\Delta\leq (x^S_j-1)\Delta \leq 2\Delta$.  Since there are $4$ choices for the extremal configuration, we have $|\cT_3^-|\leq 8\Delta$.
\medskip

Since $p_{|S|}\leq 1$ and $\xi_{\sigma,\tau}(v,c,S)\leq 2$, we conclude the second statement of the lemma,
\begin{equation*}
\overline{\nabla_B}(\sigma, \tau, c,\overline{\cD})\leq \overline{\nabla_B}(\sigma, \tau, c,\cT) \leq \frac{2\eta}{\Delta}\left(3k+12k+\Delta+8 \Delta\right) = 2\eta\left(9+\frac{15k}{\Delta}\right)\;.
%\overline{\nabla_B}(\sigma, \tau, c,\overline{\cD})\leq \overline{\nabla_B}(\sigma, \tau, c,\cT) = 2\eta\left(\frac{3k}{\Delta}+\frac{12k}{\Delta}+\Delta+8\Delta\right) = 2\eta\left(9\Delta+\frac{15k}{\Delta}\right)\;.
\qedhere\end{equation*}
\end{proof}

%!TEX root = fullpaper.tex

\section{Proof of Theorem~\ref{thm:list_flip}}\label{app:list}

\begin{proof}

The proof follows the same lines as the proof of Theorem~\ref{thm:main} presented in Section~\ref{sec:DPP}, \new{although a similar analysis can be done using the multi-step coupling approach presented in Section~\ref{sec:CM}}. We will describe the proof strategy, stressing the parts where the argument is different for list-coloring and omitting the ones that are straightforward adaptations of the  coloring case.

Let $\sigma,\tau\in \Omega^L$ that differ only at a vertex $v$. For the rest of the proof, we will assume $c\neq \sigma(v),\tau(v)$; the case $c\in\{\sigma(v),\tau(v)\}$ produces weaker constraints and can be dealt similarly as in the non-list-coloring case (see Remark~\ref{remark:specialcase}).  
For $\varphi\in \{\sigma,\tau\}\subseteq \Omega^L$, $\pi\in \{\sigma,\tau\}\setminus \{\varphi\}$, $c\in L(v)$ and $U_c=\{u_1,\dots, u_{m}\}$ the set of neighbors of $v$ with color $c$ with $m=\delta_c$, we define the configurations $(A^L, B^L, (a^L_1,\dots,a^L_m);(b^L_1,\dots,b^L_m))$ for $c$ in $(\sigma,\tau)$ as before, with the sole difference that we also set $A^L=0$ if $S_\sigma(v,c)$ is un-flippable, $B^L=0$ if $S_\tau(v,c)$ is un-flippable, $a^L_i=0$ if $S_\tau(u_i,\sigma(v))$ is un-flippable and $b^L_i=0$ if $S_\sigma(u_i,\tau(v))$ is un-flippable.  As $c\neq \sigma(v),\tau(v)$, if $A^L\neq 	$, then $A^L=1+ a_1^L+\dots+a_m^L$, and similarly for $B^L$.
We define $i^L_{\max}$, $j^L_{\max}$, $a^L_{\max}$ and $b^L_{\max}$ analogously as before, and note that the latter two can be zero.
%As $c\neq \sigma(v),\tau(v)$, note that $A^L=1+ a_1^L+\dots+a_r^L$ if $S_\sigma(v,c)$ is flippable and $A^L=0$ otherwise. Let $B^L=1+b_1^L+\dots+b^L_r$ if $S_\tau(v,c)$ is flippable and $B^L=0$ otherwise. 
Define $q_i(L)$ and $q_i'(L)$ as in~\eqref{appeq:qi} and~\eqref{appeq:qprimei} 
 for the list version of the parameters.

According to this, we use the same definition of extremal configurations, metric $d$ on $\Omega^L$,  $d_H$ and $d_B$. Again, for any pair $\sigma',\tau'\in \Omega^L$, we have $d(\sigma',\tau') \leq d_H(\sigma',\tau')$, which implies that $d_B(\sigma',\tau')\geq 0$.
%For $c\in \mathbb{N}$ consider the sets
%$\cA^L_\varphi(c) := \{S_\varphi(v,c), \{S_\varphi(w_i,\pi(v))\}_{i\in [r]}\}$ and the multisets $\cA^L_\varphi:=\{\cA^L_\varphi(c):\,c\in \mathbb{N}\}$ and $\overline{\cA^L_\varphi}:=\cK_\varphi^L\setminus \cA^L_\varphi$. 
%As before, for every $S\in \cK^L_\sigma\cup \{\emptyset\}$, one can define $\bP^L_\sigma(S):=\bP(Y^L_{t+1}=\sigma_S\mid \, Y^L_t=\sigma)$. 
We use the same coupling as the one defined in Appendix~\ref{app:vigodareview} and define $\nabla^L$, $\nabla_H^L$ and $\nabla_B^L$ analogously as for colorings. Fix the flip parameters $\mathbf{\hat{p}}$ provided in Observation~\ref{obs:DPP_assig}.

We will prove an analogue of Corollary~\ref{cor:improvement} to bound $\nabla^L_H$ for list-colorings. As in Section~\ref{sec:metric}, we have
$$
\nabla^L_H(\sigma,\tau) = \sum_{c\in \mathbb{N}} \nabla_H^L(\sigma,\tau,c)\;.
$$
Suppose first that $m=0$. Then $c\in L(v)$ and $\nabla_H^L(\sigma,\tau,c)=-1$. 
%and otherwise, $\nabla_H^L(\sigma,\tau,c)=0$. 

If $m\geq 1$, the list analogues of equations \eqref{eq:H} and~\eqref{eq:mainvigodaineqcomponent} hold, so
\begin{align}\label{eq:list}
\nabla_H^L(\sigma,\tau,c) &\leq (A^L-a^L_{\max}-1)p_{a^L}+(B^L-b^L_{\max}-1)p_{b^L}\nonumber\\ 
&\;\;\;\;+ \sum_{i\in [m]} (a^L_i q_i(L)+ b^L_i q_i'(L)-\min\{q_i(L),q_i'(L)\})\;.
\end{align} 
We will bound each term $\nabla_H^L(\sigma,\tau,c)$ depending on whether $c\in L(v)$ or $c\notin L(v)$.

If $c\notin L(v)$, then it suffices to show that $\nabla_H^L(\sigma,\tau,c)\leq m \hat{\lambda}$. 
Note that $A^L=B^L=0$, $q_i(L)=p_{a_i^L}$ and  $q'_i(L)=p_{b_i^L}$ for every $i\in [m]$. Let $\{c_i^L,d_i^L\}= \{a_i^L,b_i^L\}$ with $p_{c_i^L}\geq p_{d_i^L}$.  Using~\eqref{eq:list}, we obtain 
\begin{align*}
\nabla_H^L(\sigma,\tau,c) &\leq \sum_{i\in [m]} (a^L_i p_{a_i^L}+ b^L_i p_{b_i^L}-\min\{p_{a_i^L},p_{b_i^L}\}) \\
&= \sum_{i\in [m]} c^L_i p_{c_i^L}+ (d^L_i-1) p_{d_i^L}
\leq \frac{4}{3}m< m \hat{\lambda}\;,
\end{align*} 
where we have used that $\alpha p_\alpha\leq 1$ and $(\alpha-1)p_\alpha\leq \frac{1}{3}$.

Now assume that $c\in L(v)$. We will compare these bounds with the ones we obtained in Section~\ref{sec:nablaH} by plugging the values of the configuration $(A^L, B^L, (a^L_1,\dots,a^L_m), (b^L_1,\dots,b^L_m))$. Observe that there are only two differences with respect to non-list-colorings; first, $A^L$ and $B^L$ can be zero, and second, $a^L_{\max}$ and $b^L_{\max}$ can be zero. Recall that $p_0=0$.  It is important to stress that, since $c\in L(v)$, $a^L_{\max}=0$ implies $A^L=0$, and similarly for $B^L$. Therefore, the only difference between~\eqref{eq:list} and the bound obtained from~\eqref{eq:H}, are the cases where either $A^L=0$ or $B^L=0$. If $A^L=a^L_{\max}=0$, then the total contribution of this part is zero and analogously for $B^L$. 
Therefore, the only interesting case is when $A^L=0$ and $a^L_{\max}\neq 0$; in this case $m\geq 2$. 
Since $A^L=0$ and $c\in L(v)$, there exists $j\in [m]$ such that $a^L_{j}=0$.  Consider the configuration of size $m-1$
\begin{align}\label{eq:reduced}
(A,B;(a^L_1,\dots,a^L_{j-1},a^L_{j+1} ,\dots, a^L_m),(b^L_1,\dots,b^L_{j-1},b^L_{j+1} ,\dots, b^L_m))\;.
\end{align}
where $A=1+ \sum_{i\neq j}a^L_i$ and $B=1+ \sum_{i\neq j}b^L_i$. Let $b_{\max}$ be the maximum of the $b^L_i$ with $i\neq j$ and note that $b_{\max}\leq b^L_{\max}$. Recall that $\mathbf{\hat{p}}$ is a feasible solution for Linear Program~\ref{def:prelp} with objective value $\frac{11}{6}$ and a feasible solution of Linear Program~\ref{def:prelp_DPP} with objective value $\frac{161}{88}$. If $m\geq 4$, then the configuration of size $m-1$ in~\eqref{eq:reduced} for $c$ is non-extremal and
% provided that~\eqref{eq:reduced} is not extremal configuration, we have
\begin{align*}
\nabla_H(\sigma,\tau,c)&\leq (A-a^L_{\max}-1) p_{A}+ (B-b_{\max}-1) p_{B} + \sum_{i\neq j} a_i^Lq_{i}+b_i^L q'_{i}-\min\{q_{i},q'_{i}\} \\
&\leq \frac{161}{88} (m-1)-1\;.
\end{align*}
If $1\leq m\leq 3$, then~\eqref{eq:reduced} can be extremal and
\begin{align*}
\nabla_H(\sigma,\tau,c)&\leq (A-a^L_{\max}-1) p_{A}+ (B-b_{\max}-1) p_{B} + \sum_{i\neq j} a_i^Lq_{i}+b_i^L q'_{i}-\min\{q_{i},q'_{i}\} \\
&\leq \frac{11}{6}(m-1)-1 \leq \frac{161}{88} (m-1)- \frac{131}{132}\;.
\end{align*}
For $i=i_{\max}^L$ we have $q_{i}= p_{a_{\max}}-p_A$ and $q_{i}(L)= p_{a_{\max}}$. Moreover, we have $q_{j}'= 0$ and $q_{j}'(L)\leq p_{b_j^L}$. We may assume that $b_{\max}^L\neq 0$, as otherwise we have $b=b_{\max}=0$ and the contribution of this part is zero, as before. Using these bounds and~\eqref{eq:list}, we obtain that for any such $c\in [k]$
\begin{align*}
\nabla^L_H(\sigma,\tau,c) &\leq  (B^L-b^L_{\max}-1) p_{B^L} + \sum_{i\in [m]} (a_i^L q_{i}(L)+b_i^L q'_{i}(L)-\min\{q_{i}(L),q'_{i}(L)\})\\
&\leq \nabla_H(\sigma,\tau,c) - (A-2a^L_{\max}-1) p_{A}+(B^L-b^L_{\max}-1) p_{B^L} +b^L_j p_{b^L_j}\\
&\leq \nabla_H(\sigma,\tau,c) + (A-1) p_{A}+(B^L-2) p_{B^L} +b^L_j p_{b^L_j}\\
&\leq \nabla_H(\sigma,\tau,c) + \frac{19}{12}\\
&\leq \frac{161}{88}\cdot m-1\;,
\end{align*}
where we have used that $b^L_{\max}\geq 1$, $A\geq a_{\max}^L+1$,  $\alpha p_\alpha\leq 1$, $(\alpha-1)p_\alpha\leq \frac{1}{3}$ and $(\alpha-2)p_\alpha\leq \frac{1}{4}$. Thus, Corollary~\ref{cor:improvement} also holds for $\nabla^L_H$.

Corollary~\ref{cor:improvement2} also holds for $\nabla_B^L$ as well, since all the negative contributions on the bound are given by Kempe components $S=S_\sigma(u,c)$ of size $1$, which are always flippable as $c\in L(u)$. The positive contributions of the Kempe components is still bounded by the same quantity since, in the worst case, they are all flippable. 

Using the same flip parameters and reasoning as in the proof of Theorem~\ref{thm:improvement}, it follows that for every $k$-list assignment $L$ with $k\geq (\frac{11}{6}-\eta)\Delta$, flip dynamics for $L$-colorings satisfies $\tau_{\text{mix}}(\epsilon) = O(n(
\log{n}+\log{\epsilon^{-1}}))$, concluding the proof of Theorem~\ref{thm:list_flip}.

\end{proof}

\end{document}